\documentclass[final,5p,times,twocolumn]{elsarticle}

\usepackage[T1]{fontenc}
\usepackage{courier}
\usepackage{microtype}
\usepackage{amsmath,amssymb,amsthm,mathtools,mathrsfs}
\usepackage{graphicx}
\usepackage{booktabs,tabularx,array}
\usepackage{enumitem}
\usepackage{xcolor}
\usepackage{url}
\usepackage[colorlinks=true,allcolors=blue!55!black]{hyperref}
\usepackage{placeins}

\journal{Nuclear Physics B}
\biboptions{numbers,sort&compress}

\numberwithin{equation}{section}

\graphicspath{{./}{figures/}}
\setlist{itemsep=2pt,topsep=4pt}
\allowdisplaybreaks[2]

\newtheorem{theorem}{Theorem}[section]
\newtheorem{proposition}[theorem]{Proposition}
\newtheorem{lemma}[theorem]{Lemma}
\newtheorem{corollary}[theorem]{Corollary}
\theoremstyle{definition}
\newtheorem{definition}[theorem]{Definition}
\theoremstyle{remark}

\newcommand{\dd}{\mathrm d}

\hypersetup{
 pdftitle={Resolved Maxwell--Boundary Normal Forms and Exact Reentrant Scaling in Accelerating AdS Black Holes},
 pdfauthor={Ruiliang Li},
 pdfsubject={Maxwell--boundary normal forms and exact reentrant phase selection in accelerating AdS black holes},
 pdfkeywords={accelerating AdS black holes, Maxwell coexistence, reentrant phase transition, thermodynamic volume, singularity theory, Newton normal form}
}

\begin{document}

\begin{frontmatter}

\title{Resolved Maxwell--Boundary Normal Forms and Exact Reentrant Scaling\\
in Accelerating AdS Black Holes}

\author[thu]{Ruiliang Li\corref{cor1}}
\ead{lirl23@mails.tsinghua.edu.cn}
\cortext[cor1]{Corresponding author}
\affiliation[thu]{organization={Tsinghua University},
                  city={Beijing},
                  postcode={100084},
                  country={China}}

\begin{abstract}
Reentrant phase transitions of accelerating anti-de Sitter black holes are
known numerically, but their apparent loss at small string tension has lacked
an analytic explanation.  In the single-string ensemble at fixed pressure,
charge, and tension, we solve the unrestricted two-phase Maxwell-turning
problem for the charged slowly accelerating C-metric.  Elimination forces the
two horizon coordinates to coincide and yields a closed one-parameter locus.
An exhaustive enumeration of the remaining equilibria establishes global
phase selection throughout the physical black-hole sector.  The locus exists
for \(0<\mu<0.202602\), with no positive lower threshold.  As
\(\mu\to0\), the pressure and temperature widths of the reentrant window
contract as the fourth and third powers of the tension, while the entropy gap
and latent heat diverge.  A two-chart blow-up gives a parameter-free limiting
profile and identifies thermodynamic-volume inversion as the turning mechanism
between two distinct noncritical phases.  At a Maxwell boundary with fixed
sheet incidence, projective Newton data of the thermodynamic jumps determine
the leading coexistence profile.  Their positive simple roots fix the turn
count, order, and curvature signs; the C-metric realizes the primitive
binomial class.
\end{abstract}


\end{frontmatter}

\section{Introduction}
\label{sec:introduction}

AdS black holes exhibit both the Hawking--Page transition and canonical
small/large phase coexistence
\cite{HawkingPage1983,Chamblin1999Fluctuations}.  Acceleration changes this
thermodynamic setting in two coupled ways.  The
conical defect contributes a thermodynamic string tension
\cite{GregoryScoins2019Chemistry}, while the slow-acceleration conditions cut
the equilibrium family by physical boundaries
\cite{AnabalonEtAl2019Thermodynamics,AbbasvandiEtAl2019Snapping}.  Four-dimensional
Born--Infeld AdS thermodynamics exhibited a zeroth-order segment that was later
recognized as part of a reentrant sequence
\cite{GunasekaranKubiznakMann2012Extended,AltamiranoEtAl2014Review}.
Reentrant and multicritical transitions were subsequently analyzed in
rotating AdS and Lovelock black holes
\cite{AltamiranoKubiznakMann2013Reentrant,AltamiranoEtAl2014Triple,
FrassinoEtAl2014MultipleReentrant}.
In the charged
AdS C-metric, the physical boundaries truncate the canonical swallowtail.
The resulting snapping transition, zeroth-order segment, bicritical junction,
and pressure-driven reentrant sequence were identified in
Ref.~\cite{AbbasvandiEtAl2019Snapping}; rotation produces a still finer
splitting of the transition pressures
\cite{AbbasvandiEtAl2019FinelySplit}.  Analytic control of the reentrant turn
at small tension has so far been lacking.  Its apparent disappearance from
unscaled phase diagrams could mark either a finite physical threshold or a
singular collapse of the relevant pressure and temperature scales.

This distinction, however, cannot be read off from the number of black-hole
branches alone.  A canonical branch is locally stable when the constrained Hessian of
its thermodynamic potential is positive, but coexistence compares the values
carried by two distinct stable branches.  Physical phase selection then
requires a third test: both branches must remain inside the slowly
accelerating, positive-temperature state space and their common value must
lie below every other admissible equilibrium.  For an off-shell canonical
potential
\begin{equation}
 \mathcal F_T(x)=M(x)-TS(x),
 \qquad
 \partial_x\mathcal F_T=S_x\,[T_H(x)-T],
 \label{eq:intro-factorization}
\end{equation}
the stationary cover is determined by the zeros of \(T_H-T\), whereas the
value difference between two stationary points is
\begin{equation}
 \mathcal F_T(x_j)-\mathcal F_T(x_i)
 =\int_{x_i}^{x_j}S_x(x)\,[T_H(x)-T]\,\dd x.
 \label{eq:intro-weighted-area}
\end{equation}
The entropy measure in this weighted-area identity is invisible to the
unweighted branch topology.  This is the same distinction that underlies
Maxwell constructions and generalized free-energy landscapes in black-hole
thermodynamics
\cite{SpallucciSmailagic2013Maxwell,WeiLiu2015Clapeyron,
LiWang2022Landscape,XuEtAl2024ComplexFreeEnergy}.

Our main physical result is an analytic solution of the Maxwell-turning
problem for the charged slowly accelerating AdS C-metric.  We work in the
standard north-regular single-string ensemble at fixed \((P,Q,\mu)\), with
the holographic normalization of time used in
Refs.~\cite{AnabalonEtAl2018Holographic,AnabalonEtAl2019Thermodynamics}.
Related work using the same normalization examined charged accelerating
thermodynamics, phase behavior, and holographic heat engines in
Ref.~\cite{ZhangLiYu2019Accelerating}.
The pressure interpretation of the cosmological constant and the first law
with variable conical tensions follow
Refs.~\cite{KastorRayTraschen2009,Dolan2011PressureVolume,
KubiznakMann2012PV,CveticEtAl2011Volume,KubiznakMannTeo2017Chemistry,
AppelsGregoryKubiznak2016Thermodynamics,
AppelsGregoryKubiznak2017Conical}.  Throughout the C-metric analysis, the
physical sector consists of slowly accelerating, entropy-regular,
positive-temperature black-hole equilibria in this fixed-charge ensemble.
No independent radiation or reference-background phase is included.

Let \(q=Q/\ell\), \(u=Ar_+\), \(z=eA\), and
\(P_t=3\mu^2/(8\pi Q^2)\).  A regular turn of a two-phase Maxwell component
obeys equal temperature, equal Gibbs free energy, and equal thermodynamic
volume.  We impose these equations on two unrestricted positive-domain
states.  Their elimination first proves
\begin{equation}
 u_-=u_+=\chi.
 \label{eq:intro-common-horizon}
\end{equation}
The remaining real ideal gives
\begin{align}
 \mu&=\frac{\chi}{(1+\chi)^2},
 &q^2&=\frac{\chi^2(1+\chi^2)}{(1+\chi)^4},
 \nonumber\\
 \frac{P_{\rm turn}}{P_t}&=1+\chi^2,
 &QT_{\rm turn}&=\frac{\mu}{\pi}.
 \label{eq:intro-exact-turn}
\end{align}
The charge--acceleration coordinates \(w_\pm=z_\pm^2\) are the two roots of
a quadratic over \(\mathbb Q(\chi)\).  A saturated subresultant
reconstruction enumerates the remaining stationary states at the same
controls.  It proves that the two coexisting solutions are strict canonical
minima and have the least Gibbs value among all admissible
positive-temperature black-hole equilibria.  Outer-horizon, angular,
normalization, and slow-acceleration inequalities hold throughout the open
locus.  Its curvature is positive and its latent heat is
\begin{equation}
 \frac{\mathcal L_{\rm turn}}Q
 =\frac{\sqrt{1-6\chi^2-3\chi^4}}
 {\chi(1+\chi^2)}.
 \label{eq:intro-latent-heat}
\end{equation}

The discriminant fixes the complete tension range,
\begin{align}
 0&<\mu<\mu_+,
 \nonumber\\
 \mu_+&=
 \frac{\sqrt{2\sqrt3-3}}
 {2\bigl(1+\sqrt{2\sqrt3-3}\bigr)},
 \nonumber\\
 \mu_+&\simeq0.202602 .
 \label{eq:intro-mu-plus}
\end{align}
At \(\mu_+\) the two phases coalesce at an ordinary critical point.  At the
other end the turning locus reaches \(\mu=0\); there is no positive lower
tension.  The apparent loss of reentrance comes from the singular estimates
\begin{align}
 (P_{\rm turn}-P_t)Q^2
 &=\frac{3\mu^4}{8\pi}+O(\mu^5),
 \nonumber\\
 Q(T_t-T_{\rm turn})
 &=\frac{\mu^3}{2\pi}+O(\mu^4).
 \label{eq:intro-collapse}
\end{align}
Ordinary plots suppress the turn by four powers of tension in pressure and
three powers in temperature.

Although the collapse is geometrically narrow, its thermodynamic effect is
strong.  The two coexisting black holes occupy different small-tension charts.  One has
\(z=O(\mu^2)\) and entropy \(O(\mu^{-2})\); the boundary-linked state has
\(z=O(\mu)\) and a lapse of order \(\mu\).  With
\begin{equation}
 \widehat P=\frac{q^2-\mu^2}{\mu^4},
 \qquad
 \widehat T=\frac{\tau_t-\tau}{2\mu^3},
 \label{eq:intro-blowup-controls}
\end{equation}
the Maxwell component converges on compact positive blow-up intervals to
\begin{equation}
 \widehat T=2\sqrt{\widehat P}-\widehat P.
 \label{eq:intro-blowup-profile}
\end{equation}
The limiting turn occurs at \((\widehat P,\widehat T)=(1,1)\).  The entropy
ordering remains fixed while the thermodynamic-volume jump and the Clapeyron
slope reverse sign.  Meanwhile the entropy jump, latent heat, and reduced
stationary-saddle action difference diverge.  Thus the control window
collapses while the thermodynamic discontinuities grow without bound.

Motivated by this example, we develop a classification intrinsic to phase
selection.  For two nondegenerate minima set
\(\Delta S=S_L-S_S\) and \(\Delta V=V_L-V_S\), and let
\begin{equation}
 D(T,P)=G_L(T,P)-G_S(T,P),
 \qquad D_T=-\Delta S,
 \qquad D_P=\Delta V.
 \label{eq:intro-gap-germ}
\end{equation}
The parameterized Morse lemma separates the quadratic wells from the
critical-value gap \(D\).  For a prescribed corner blow-up, fixed labeled
transverse selected-sheet face incidence, and fixed reservoir polarization,
we classify the leading pressure, temperature, jump, and gap-conormal traces
on the resolved Maxwell half-line.  The quotient retains the restrictions of
the physical admissibility margin to the two selected critical sheets.  A
full boundary-compatible state-control equivalence induces this quotient.
The leading pressure contact and the Newton principal part of
\(\Delta V/\Delta S\) determine a projective polynomial class \([R]\).  If the
pressure has contact order \(m\) and the volume jump has pole order \(n\), the
normalized coexistence curve in the integrable case \(m>n\) is the weighted
primitive of \(R\).  Its positive simple roots give the number and order of
regular turns and their curvatures.  We prove invariance and completeness in
this selected-sheet leading-profile quotient, stability of every simple-root
turn portrait on compact resolved intervals under the stated higher-order
remainders, and local punctured two-well realizability.  Multiple positive
roots lie on a positive-codimension discriminant.  The cases \(m=n\) and
\(m<n\) give logarithmic and power-divergent corners and cannot be absorbed
into the integrable class.  For the primitive binomial \(R(r)=1-r^k\), the
result reduces to
\begin{equation}
 \widehat T=
 \frac{\iota+k}{k}\widehat P^{\iota/m}
 -\frac{\iota}{k}\widehat P^{(\iota+k)/m},
 \qquad \iota=m-n,
 \label{eq:intro-binomial-normal-form}
\end{equation}
and the C-metric realizes the class
\((a,b,c;m,n,k)=(4,2,3;2,1,1)\).  The strict transforms of the normalized-time
face on its two Maxwell sheets are proportional to \(r\) and to a positive
unit, respectively, fixing the labeled transverse incidence used here.

This value-based viewpoint also delimits what branch-topology criteria can
decide.  Two ordinary extrema of \(T_H(r_+)\) generate a three-sheet
stationary cover, but a first-order transition additionally requires a
reversal of endpoint critical-value order, admissibility of both minima, and
comparison with every competing phase included in the ensemble.  These
requirements sharpen the recently proposed two-extrema criterion
\cite{ZhangEtAl2026Geometric} and delimit the subsequent real-domain
dictionary among geometric, topological, and complex classifications
\cite{ZhangEtAl2026Unifying}.  Those dictionaries encode stationary
multiplicity and stability order, whereas Maxwell coexistence also depends
on critical-value order and admissibility.  An analytic
iso-stationary deformation moves the Maxwell temperature while preserving
the full branch cover, state order, Morse indices, and Brouwer degree.  A
clipped cusp independently shows how a physical boundary deletes an ambient
Maxwell crossing.  These observations fit the catastrophe-theory treatment
of charged AdS black holes \cite{Chamblin1999Catastrophic} and the classical
distinction between caustics and Maxwell sets
\cite{Cerf1970Stratification,GolubitskyGuillemin1973,
ArnoldGuseinZadeVarchenko2012,Vassiliev2025BifurcationSets}, with an
additional admissibility stratum imposed by the gravitational state space.

Thermodynamic topology remains useful at the stationary level
\cite{WeiLiu2022Topology,WeiLiuMann2022Defects,WeiLiu2026Review,
Wu2023AcceleratingTopology,WuYangWei2025Extended}.  Recent work has likewise
separated boundary membership and index flow from finite-radius response
singularities of the zero-point branch geometry
\cite{Wu2026BranchStructure}.  That local distinction concerns the
stationary zero-point geometry and does not reconstruct critical-value order
or a Maxwell coexistence set.  Ref.~\cite{Li2026LegendreCovariant} proves
the Brouwer--Morse interpretation and Legendre covariance of the degree for a
specified reservoir problem.  The present work keeps the reservoir
polarization fixed and studies data that the degree does not contain:
critical-value order, Maxwell coexistence, and physical membership.  The
present analysis further establishes the model-specific elimination and
global selection theorem, the boundary-layer scaling, and the
leading-profile classification at a Maxwell boundary.

The absolute action of an accelerating black hole depends on its
renormalization prescription.  Hale et al.~\cite{HaleEtAl2025ChargedAccelerating}
found the standard and topological prescriptions to differ in the presence
of an overall deficit.  In the present fixed-charge ensemble
their difference is the branch-independent shift \(3\mu T/(4P)\).  It changes
absolute free energies but leaves branch differences, the Maxwell set,
winner order, thermodynamic jumps, latent heat, and turning curve unchanged.

The paper is organized as follows.  Section~\ref{sec:cmetric} gives the
fixed-\((P,Q,\mu)\) reduction, proves the exact Maxwell-turning theorem, and
derives its latent heat and endpoint laws.  The small-tension blow-up and
thermodynamic-volume inversion follow in
Sec.~\ref{subsec:cmetric-maxwell-blowup}.  The selected-sheet
Maxwell--boundary leading-profile classification follows in
Sec.~\ref{sec:maxwell-boundary-germs}.
Section~\ref{sec:constrained-families} fixes the general constrained-family
framework.  The corrected two-extrema criterion, value-order obstruction,
and boundary continuity results are established in
Sec.~\ref{sec:local-versus-global};
the decorated wall complex and its covariance are then developed in
Sec.~\ref{sec:decorated-complex}.  Detailed eliminations, root-exhaustion
arguments, wall incidences, and action-scheme checks are collected in the
appendices.

\section{Exact Maxwell turning of the charged accelerating AdS black hole}
\label{sec:cmetric}

The charged AdS C-metric supplies a physical test in which stationary
branches, their critical values, and the admissible domain change
independently.  Its geometry belongs to the accelerating sector of the
Pleba\'nski--Demia\'nski family and contains the charged C-metric of
Kinnersley and Walker
\cite{KinnersleyWalker1970,PlebanskiDemianski1976,GriffithsPodolsky2006,
GriffithsPodolsky2009Book}.
We use the
standard single-string parametrization and the standard holographic
normalization of time
\cite{AppelsGregoryKubiznak2017Conical,AnabalonEtAl2019Thermodynamics},
held fixed as part of the ensemble.  Standard and
topological action renormalization give different absolute potentials in the
presence of an overall conical deficit
\cite{HaleEtAl2025ChargedAccelerating}.
\ref{app:cmetric-renormalization} proves that their difference is common to
all branches at fixed reservoir data, so the phase-selection results are
unchanged.  Formulas for absolute values in this section use the standard
scheme and fixed \((P,Q,\mu)\), with \(Q>0\), while \(T\) is the reservoir
control.

\subsection{Single-string reduction and physical domain}
\label{subsec:cmetric-reduction}

In the conventions of
Refs.~\cite{AppelsGregoryKubiznak2017Conical,AnabalonEtAl2019Thermodynamics},
the metric and Maxwell potential are
\begin{align}
 \mathrm ds^2&=\frac{1}{\Omega^2}\left[
 -\frac{f(r)}{\alpha^2}\mathrm dt^2+\frac{\mathrm dr^2}{f(r)}
 +r^2\left(\frac{\mathrm d\theta^2}{g(\theta)}
 +g(\theta)\sin^2\theta\frac{\mathrm d\phi^2}{K^2}\right)\right],
 \label{eq:cmetric-metric}\\
 \Omega&=1+Ar\cos\theta,\nonumber\\
 f(r)&=(1-A^2r^2)\left(1-\frac{2m}{r}+\frac{e^2}{r^2}\right)
       +\frac{r^2}{\ell^2},\nonumber\\
 g(\theta)&=1+2mA\cos\theta+e^2A^2\cos^2\theta,\nonumber\\
 B&=-\frac{e}{\alpha}\left(\frac1r-\frac1{r_+}\right)\mathrm dt .
 \nonumber
\end{align}
Writing \(\Xi=1+e^2A^2\), the normalized time and pressure are
\begin{equation}
 \alpha^2=\Xi(1-A^2\ell^2\Xi),
 \qquad
 P=\frac{3}{8\pi\ell^2}.
 \label{eq:cmetric-alpha}
\end{equation}
We regularize the north pole and retain a string at the south pole,
\begin{equation}
 K=\Xi+2mA,\qquad
 \mu=\frac{mA}{K},\qquad
 Q=\frac eK.
 \label{eq:cmetric-single-string}
\end{equation}
The mass, entropy, Hawking temperature, and thermodynamic volume in this
normalization are
\begin{align}
 M&=\frac{m(1-A^2\ell^2\Xi)}{K\alpha},
 &
 S&=\frac{\pi r_+^2}{K(1-A^2r_+^2)},\nonumber\\
 T_H&=\frac{f'(r_+)}{4\pi\alpha},
 &
 V&=\frac{4\pi}{3K\alpha}\left[
 \frac{r_+^3}{(1-A^2r_+^2)^2}+mA^2\ell^4\Xi\right].
 \label{eq:cmetric-thermodynamics}
\end{align}
They are the nonrotating specialization of the full-cohomogeneity first law
of Ref.~\cite{AnabalonEtAl2019Thermodynamics}.  Its restriction to fixed
\((P,Q,\mu)\) is
\begin{equation}
 \mathrm dM=T_H\,\mathrm dS .
 \label{eq:cmetric-restricted-first-law}
\end{equation}

To expose the fixed-\((Q,\mu)\) reduction, we rewrite the single-string
relations used in Ref.~\cite{AbbasvandiEtAl2019Snapping} in terms of
\begin{equation}
 d=1-2\mu,\qquad z=eA,\qquad q=\frac{Q}{\ell},\qquad u=Ar_+ .
 \label{eq:cmetric-reduced-variables}
\end{equation}
Equations~\eqref{eq:cmetric-single-string} give
\begin{align}
 \Xi&=1+z^2,&
 K&=\frac{1+z^2}{d},&
 mA&=\frac{\mu(1+z^2)}{d},\nonumber\\
 AQ&=\frac{zd}{1+z^2},&
 \frac mQ&=\frac{\mu(1+z^2)^2}{zd^2},&
 A\ell&=\frac{zd}{q(1+z^2)}.
 \label{eq:cmetric-exact-reduction}
\end{align}
At fixed \((q,\mu)\), the outer horizon is the largest physical root
\(u\in(0,1)\) of
\begin{equation}
 \begin{split}
 H(u,z;q,d)={}&
 (1-u^2)\left[
 u^2-\frac{(1-d)(1+z^2)}d\,u+z^2\right]\\
 &+\frac{u^4q^2(1+z^2)^2}{z^2d^2}=0 .
 \end{split}
 \label{eq:cmetric-horizon-polynomial}
\end{equation}
The reduced lapse factor is
\begin{equation}
 \alpha^2=1+z^2-\frac{z^2d^2}{q^2}.
 \label{eq:cmetric-reduced-alpha}
\end{equation}
On a simple horizon sheet we use the constrained derivative
\begin{equation}
 \mathscr D_z=\partial_z-\frac{H_z}{H_u}\partial_u .
 \label{eq:cmetric-constrained-derivative}
\end{equation}

The dimensionless state functions
\begin{equation}
 \tau=4\pi Q T_H,\qquad
 s=\frac{S}{\pi Q^2},\qquad
 \mathfrak m=\frac MQ,\qquad
 \nu=\frac{3V}{4\pi Q^3}
\end{equation}
take the form
\begin{align}
 \tau&=\frac{zdH_u}{(1+z^2)u^2\alpha},
 &
 s&=\frac{u^2(1+z^2)}{z^2d(1-u^2)},\nonumber\\
 \mathfrak m&=\frac{\mu\alpha}{zd},
 &
 \nu&=\frac1\alpha\left[
 \frac{u^3(1+z^2)^2}{z^3d^2(1-u^2)^2}
 +\frac{\mu zd}{q^4}\right].
 \label{eq:cmetric-reduced-state-functions}
\end{align}
The first law becomes
\begin{equation}
 \mathscr D_z\mathfrak m=\frac{\tau}{4}\mathscr D_zs .
 \label{eq:cmetric-dimensionless-first-law}
\end{equation}
Consequently, for \(\vartheta=4\pi QT\),
\begin{equation}
 \frac{\mathcal G_T}{Q}=\mathfrak m-\frac{\vartheta s}{4},
 \qquad
 \mathscr D_z\left(\frac{\mathcal G_T}{Q}\right)
 =\frac{\tau-\vartheta}{4}\mathscr D_zs .
 \label{eq:cmetric-offshell}
\end{equation}
On an equilibrium sheet, \(\vartheta=\tau\), and we write its dimensionless
critical value as
\begin{equation}
 \mathfrak g:=\frac{G}{Q}
 =\mathfrak m-\frac{\tau s}{4}.
 \label{eq:cmetric-dimensionless-gibbs}
\end{equation}
The entropy-coordinate restriction implicit in
Eq.~\eqref{eq:cmetric-offshell} is unnecessary.  Before selecting any
horizon root, direct differentiation gives, on \(H=0\),
\begin{equation}
 H_us_z-H_zs_u
 =-\frac{4\mu u^2(1+z^2)}{z^3d^2}.
 \label{eq:cmetric-entropy-wedge}
\end{equation}
Combining this identity with the temperature in
Eq.~\eqref{eq:cmetric-reduced-state-functions} yields
\begin{equation}
 \mathscr D_zs
 =-\frac{4\mu}{z^2d\alpha\tau}<0
 \qquad
 (0<\mu<1/4,\ \alpha>0,\ \tau>0).
 \label{eq:cmetric-entropy-regularity}
\end{equation}
Thus entropy is a regular, strictly monotone coordinate on each connected
fixed-control component of the positive-temperature physical horizon sheet,
and \(\tau=\vartheta\) is equivalent there to stationarity of
\(\mathcal G_T\).  The Morse criterion is likewise
unambiguous, with a canonical minimum satisfying
\(\mathscr D_z\tau<0\).  At \(\mu=0\), where the accelerating coordinate
itself degenerates, the nonaccelerating limit is covered by \(r_+\) or \(S\).
\ref{app:cmetric-identities} proves the off-shell polynomial identity
behind Eq.~\eqref{eq:cmetric-entropy-wedge}.

The algebraic horizon condition is, however, only the first of several
admissibility tests.  The physical state space also requires
\begin{equation}
 0\leq\mu<\frac14,\qquad
 \alpha^2>0,\qquad
 g(\theta)>0\quad(0\leq\theta\leq\pi),
 \label{eq:cmetric-basic-admissibility}
\end{equation}
the existence of an outer black-hole horizon \(0<u<1\), and the absence of an
acceleration horizon on the conformal boundary.  With \(v=\cos\theta\), define
\begin{equation}
 C(v)=(v^2-1)\left(1+
 2\frac{\mu(1+z^2)}d\,v+z^2v^2\right)
 +\frac{q^2(1+z^2)^2}{z^2d^2}.
 \label{eq:cmetric-boundary-polynomial}
\end{equation}
The slowly accelerating chamber has \(C(v)>0\) for all \(v\in[-1,1]\).
In the present notation, the boundary-horizon parametrization of
Ref.~\cite{AbbasvandiEtAl2019Snapping} follows from \(C=C_v=0\) and reads
\begin{align}
 mA&=\frac{v(1+2z^2v^2-z^2)}{1-3v^2},\nonumber\\
 A\ell&=\frac{\sqrt{1-3v^2}}
 {(1-v^2)\sqrt{1-z^2v^2}}.
 \label{eq:cmetric-slow-face}
\end{align}
The original inequalities select the physical signs and interval of \(v\).
In particular, \(\alpha^2>0\) is weaker than slow acceleration
\cite{AbbasvandiEtAl2019Snapping}; using it as the only boundary condition
adds states with an extra horizon.

\subsection{Wall equations and the physical Maxwell test}
\label{subsec:cmetric-walls}

The fold surface is the physical image of
\begin{equation}
 H=0,\qquad \mathscr D_z\tau=0,
 \label{eq:cmetric-fold-system}
\end{equation}
under
\begin{equation}
 T=\frac{\tau}{4\pi Q},\qquad
 P=\frac{3q^2}{8\pi Q^2}.
\end{equation}
An ordinary canonical critical endpoint additionally satisfies
\begin{equation}
 \mathscr D_z^2\tau=0,\qquad
 \mathscr D_z^3\tau\ne0.
 \label{eq:cmetric-critical-system}
\end{equation}
The physical solutions of Eqs.~\eqref{eq:cmetric-fold-system} and
\eqref{eq:cmetric-critical-system} project to the ordinary critical curve
\(q=q_{\rm crit}(\mu)\).
For two distinct Euclidean equilibria \((u_i,z_i)\) at the same
\((q,\mu)\), the extended critical-value wall is obtained from
\begin{equation}
 H_i=0,\qquad
 \tau_1=\tau_2,\qquad
 \mathfrak m_1-\frac{\tau_1s_1}{4}
 =\mathfrak m_2-\frac{\tau_2s_2}{4}.
 \label{eq:cmetric-critical-value-system}
\end{equation}
It becomes a physical Maxwell point only after both states pass every
admissibility test, lie in entropy-regular patches, and are local canonical
minima,
\begin{equation}
 (\mathscr D_zs_i)(\mathscr D_z\tau_i)>0,
 \label{eq:cmetric-local-stability}
\end{equation}
and their common value is the lowest among all physical equilibria.
The distinction between Eq.~\eqref{eq:cmetric-critical-value-system} and this
global test is what prevents a free-energy crossing involving a clipped or
unstable sheet from being misidentified as coexistence.

At fixed \((Q,\mu)\), two equilibrium branches obey
\begin{equation}
 \mathrm d(\Delta G)=-\Delta S\,\mathrm dT+\Delta V\,\mathrm dP.
 \label{eq:cmetric-delta-G}
\end{equation}
On a regular coexistence surface,
\begin{equation}
 \frac{\mathrm dP}{\mathrm dT}=\frac{\Delta S}{\Delta V},
 \qquad
 \frac{\mathrm dT}{\mathrm dP}=\frac{\Delta V}{\Delta S}.
 \label{eq:cmetric-clapeyron}
\end{equation}
These identities provide consistency checks on the coexistence surface.
Figure~\ref{fig:cmetric-phase-slice} shows the stationary and
critical-value curves on a representative slice.

\begin{figure}[t]
 \centering
 \includegraphics[width=\linewidth]{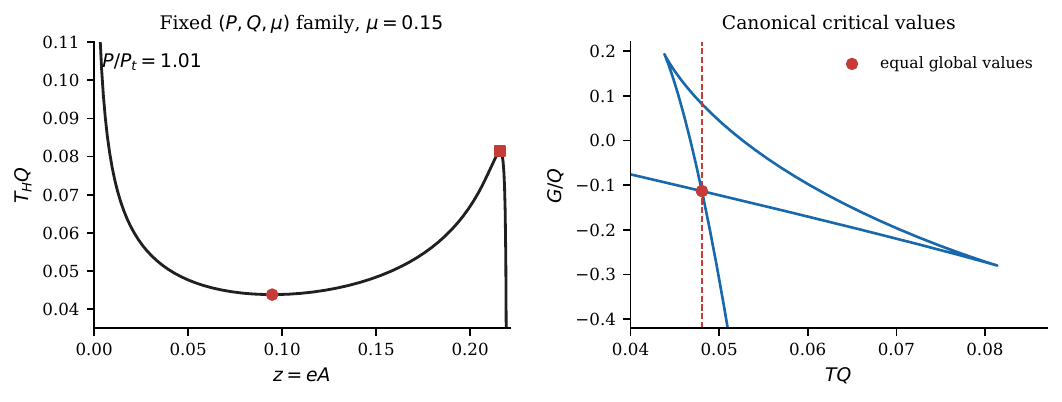}
 \caption{A representative numerical slice
 at \(Q=1\), \(\mu=0.15\), and \(P/P_t=1.01\).  The left panel shows the two
 canonical folds.  The right panel is the parametric critical-value curve;
 the solid circular marker represents two distinct physical
 equilibria with the same \((T,G)\).  This slice lies below the algebraic
 threshold \(\mu_+\) given by Eq.~\eqref{eq:cmetric-mu-plus-strong}.}
 \label{fig:cmetric-phase-slice}
\end{figure}

\subsection{The boundary corner and the snapping pressure}
\label{subsec:cmetric-X}

The point \(X\) identified in Ref.~\cite{AbbasvandiEtAl2019Snapping} is the
intersection of the normalized-time face, the bulk-extremal face, and the
slow-acceleration face.  Its ambient coordinates obey
\begin{equation}
 A\ell=\frac1{\sqrt{1+z^2}},
 \qquad
 mA=z\sqrt{1+z^2}.
 \label{eq:cmetric-X-ambient}
\end{equation}
Imposing Eq.~\eqref{eq:cmetric-single-string}, these relations reduce to
\begin{equation}
 \frac{\mu}{d}=\frac{z}{\sqrt{1+z^2}},
 \qquad q=\mu.
 \label{eq:cmetric-X-reduced}
\end{equation}
It follows that
\begin{equation}
 P_t=\frac{3\mu^2}{8\pi Q^2}.
 \label{eq:cmetric-Pt}
\end{equation}
The snapping pressure in Eq.~\eqref{eq:cmetric-Pt}, the snapping of the
swallowtail, the zeroth-order line, and the reentrant phase behavior were
established in Ref.~\cite{AbbasvandiEtAl2019Snapping}.  In the reduced
variables, these features meet at a wall-complex corner where the critical
slice loses a stationary sheet through the physical boundary.

Near this corner, two independent square-root scales appear.  At fixed
\(q=\mu\), let \(z_X\) be determined by
Eq.~\eqref{eq:cmetric-X-reduced}.  Then
\begin{equation}
 \alpha^2=\frac{2}{z_X}(z_X-z)
 +O\bigl((z_X-z)^2\bigr).
 \label{eq:cmetric-X-alpha-scaling}
\end{equation}
At the same point the horizon polynomial has a double root.  If
\(\xi=u-u_X\) and \(\rho_e\) is transverse to the extremal face, then
\begin{equation}
 H=c_{20}\xi^2+c_{01}\rho_e+\cdots,
 \qquad c_{20}c_{01}\ne0.
 \label{eq:cmetric-X-horizon-scaling}
\end{equation}
Hence \(\xi=O(\rho_e^{1/2})\), whereas
\(\alpha=O(\rho_\alpha^{1/2})\) with
\(\rho_\alpha=\alpha^2\).  Since \(T_H\) contains \(H_u/\alpha\),
\begin{equation}
 M=O(\rho_\alpha^{1/2}),\qquad
 V=O(\rho_\alpha^{-1/2}),\qquad
 T_H=O\left(\sqrt{\rho_e/\rho_\alpha}\right).
 \label{eq:cmetric-X-weighted-scaling}
\end{equation}
The limiting temperature is therefore path-dependent unless the relative
weight of the two faces is fixed.  These independent weights make \(X\) a
boundary-corner singularity of the selected physical family.

\subsection{The exact turning locus in the physical Maxwell set}
\label{subsec:cmetric-exact-turning}

Fix \((Q,\mu)\).  A regular finite-separation Maxwell turn is a smooth point
of a two-phase coexistence component at which the phases remain distinct,
\(\Delta S\neq0\), and \(\mathrm dT/\mathrm dP=0\).  The Clapeyron relation
\cite{WeiLiu2015Clapeyron} gives
\begin{equation}
 \frac{\mathrm dT}{\mathrm dP}=\frac{\Delta V}{\Delta S},
 \label{eq:cmetric-turn-definition}
\end{equation}
so a regular turn is equivalent to \(\Delta V=0\).  For distinct states on
the same connected entropy-regular horizon component considered here,
Eq.~\eqref{eq:cmetric-entropy-regularity} makes \(\Delta S\neq0\) automatic.
For the exact pair below the same conclusion also follows directly from
Eq.~\eqref{eq:cmetric-turning-entropy-gap}.  This excludes a critical
coalescence from the definition; the latter appears only as the upper
endpoint of the turning locus.

\begin{theorem}[Complete physical Maxwell-turning classification]
\label{thm:cmetric-exact-turning-main}
In the standard single-string fixed-\((P,Q,\mu)\) ensemble, define
\begin{equation}
 \mu=\frac{\chi}{(1+\chi)^2},\qquad
 0<\chi<\chi_c,\qquad
 \chi_c^2=\frac{2\sqrt3-3}{3}.
\end{equation}
There is a connected locus of strict finite-separation physical Maxwell
turning points with
\begin{align}
 q^2&=\frac{\chi^2(1+\chi^2)}{(1+\chi)^4},
 &\tau&=4\mu,\nonumber\\
 \mathfrak g_-&=\mathfrak g_+=0,
 &\nu_-&=\nu_+,
\end{align}
and the two phases are the roots of
\begin{align}
 u_-&=u_+=\chi,\nonumber\\
 0={}&(\chi^6+\chi^4-2\chi^2+1)w^2\nonumber\\
 &+(\chi^6+2\chi^4-\chi^2)w+\chi^6 .
\end{align}
Here \(w=z^2\), and \(z=eA\) is the charge--acceleration coordinate defined
in Eq.~\eqref{eq:cmetric-reduced-variables}; the horizon coordinate is
\(u=Ar_+\).
Both phases are admissible and entropy-regular, and each is a strict canonical
minimum.  Every other physical equilibrium at the same controls has strictly
larger Gibbs free energy.  The turning is nondegenerate,
with
\begin{equation}
 \left.\frac{\mathrm d^2\tau}{\mathrm dq^2}\right|_{\rm turn}
 =\frac{4(1-\chi)^2(1+\chi)^4}
 {\chi^3(1+\chi^2)^2}>0.
\end{equation}
Every such physical Maxwell turning in this ensemble lies on this locus,
whose parameter range is
\begin{equation}
 0<\mu<\mu_+,
 \qquad
 \mu_+=\frac{\sqrt{2\sqrt3-3}}
 {2(1+\sqrt{2\sqrt3-3})}.
\end{equation}
It has no positive lower endpoint and terminates by coalescence at the ordinary
critical point \(\mu=\mu_+\).  Its lower closure is the
Maxwell--admissibility corner resolved by the small-tension blow-up.
\end{theorem}

\begin{proof}
Begin with two independent positive-domain horizon states satisfying equal
temperature, Gibbs free energy, and thermodynamic volume.  The unrestricted
elimination in \ref{app:unrestricted-turning-classification} forces
\(u_-=u_+\); the common horizon coordinate therefore follows from the
two-horizon equations.  Substitution of \(u_-=u_+=\chi\) reduces the system to
the quadratic displayed in the theorem, and the rigidity analysis in
\ref{app:common-horizon-rigidity} exhausts its real positive sector.

The state-function factorizations and exact sign estimates in
\ref{app:cmetric-turning-certificate} establish entropy regularity,
admissibility, positive canonical Hessians, and the stated curvature.  The
saturated subresultant reconstruction in \ref{app:saturated-reconstruction}
enumerates every positive-temperature equilibrium at the same controls and
proves the strict global value order.  Together, these results establish
physical phase selection.  They also show that no physical
face is reached before the critical merger and that the lower closure occurs
only at \(\mu=0\).
\end{proof}

The equal-temperature,
equal-free-energy, and equal-volume equations admit a common
parametrization.  Introduce \(\chi\) by
\begin{equation}
 \mu=\frac{\chi}{(1+\chi)^2},
 \qquad
 d=\frac{1+\chi^2}{(1+\chi)^2}.
 \label{eq:cmetric-chi-mu}
\end{equation}
The coexisting phases have the same outer-horizon coordinate
\(u_-=u_+=\chi\), while \(w_\pm=z_\pm^2\) are the two roots of
\begin{align}
 F_\chi(w)&=A_\chi w^2+B_\chi w+\chi^6=0,\nonumber\\
 A_\chi&=\chi^6+\chi^4-2\chi^2+1,\nonumber\\
 B_\chi&=\chi^6+2\chi^4-\chi^2.
 \label{eq:cmetric-turning-polynomial}
\end{align}
Explicitly,
\begin{equation}
 w_\pm=
 \frac{\chi^2\left[1-2\chi^2-\chi^4
 \mathbin{\pm}(1-\chi^2)\sqrt{1-6\chi^2-3\chi^4}\right]}
 {2\left(\chi^6+\chi^4-2\chi^2+1\right)}.
 \label{eq:cmetric-turning-roots}
\end{equation}
Substitution in the horizon equation and the state functions gives
\begin{align}
 q^2&=\frac{\chi^2(1+\chi^2)}{(1+\chi)^4}
     =\mu^2(1+\chi^2),\nonumber\\
 \tau_-&=\tau_+=4\mu,
 &
 \mathfrak g_-&=\mathfrak g_+=0,\nonumber\\
 \nu_-&=\nu_+
 =\frac{(1-\chi)(1+\chi)^5}
 {\chi^3(1+\chi^2)^2}.
 \label{eq:cmetric-turning-data}
\end{align}
The common zero in Eq.~\eqref{eq:cmetric-turning-data} is a standard-scheme
normalization; equality and branch ordering are scheme invariant by
Proposition~\ref{prop:exact-scheme-robustness}.
Thus
\begin{equation}
 \frac{P_{\rm turn}}{P_t}=1+\chi^2,
 \qquad
 QT_{\rm turn}=\frac{\mu}{\pi}.
 \label{eq:cmetric-turning-controls}
\end{equation}
This stationary point is a strict local minimum of the coexistence
temperature.  Let \(p=q^2\) and set
\(\Delta Y=Y(w_+)-Y(w_-)\).  Along a regular Maxwell sheet,
\begin{equation}
 \frac{\mathrm d\tau}{\mathrm dp}=2\frac{\Delta\nu}{\Delta s}.
 \label{eq:cmetric-dimensionless-clapeyron-p}
\end{equation}
At the exact pair \(\Delta\nu=0\).  Exact implicit differentiation at fixed
\((\tau,\mu)\) gives
\begin{align}
 \left(\partial_p\nu\right)_{\tau,+}
 -\left(\partial_p\nu\right)_{\tau,-}
 &=\mathcal C_V(\chi)(w_+-w_-),\nonumber\\
 \mathcal C_V(\chi)&=
 \frac{(\chi-1)(\chi+1)^9A_\chi}
 {2\chi^9(1+\chi^2)^4}<0.
 \label{eq:cmetric-turning-curvature-factor}
\end{align}
The entropy jump reduces on the same quadratic to
\begin{equation}
 \Delta s=-\frac{A_\chi(1+\chi)^2(w_+-w_-)}
 {\chi^4(1+\chi^2)(1-\chi^2)}<0.
 \label{eq:cmetric-turning-entropy-jump}
\end{equation}
Consequently the curvature with respect to the physical pressure coordinate
is nonzero and positive.  In the dimensionless charge coordinate it takes
the closed form
\begin{equation}
 \left.\frac{\mathrm d^2\tau}{\mathrm dq^2}\right|_{\rm turn}
 =\frac{4(1-\chi)^2(1+\chi)^4}
 {\chi^3(1+\chi^2)^2}>0.
 \label{eq:cmetric-turning-positive-curvature}
\end{equation}
\ref{app:cmetric-turning-certificate} gives the quotient-ring
calculation and proves that its implicit Jacobian is nonzero throughout the
open interval.

The exact pair also determines the strength of the first-order transition.
We label the root \(w_-\) by \(L\), for the large-entropy phase,
and \(w_+\) by \(S\), for the small-entropy phase.

\begin{corollary}[Latent heat along the turning locus]
\label{cor:cmetric-turning-latent-heat}
Let
\begin{equation}
 D_\chi=1-6\chi^2-3\chi^4.
 \label{eq:cmetric-turning-Dchi}
\end{equation}
At every point of the open turning locus,
\begin{align}
 \frac{S_L-S_S}{\pi Q^2}
 &=\frac{(1+\chi)^2\sqrt{D_\chi}}
 {\chi^2(1+\chi^2)},
 \label{eq:cmetric-turning-entropy-gap}\\
 \frac{\mathcal L_{\rm turn}}Q
 =\frac{M_L-M_S}{Q}
 &=\frac{T_{\rm turn}(S_L-S_S)}Q
 =\frac{\sqrt{D_\chi}}{\chi(1+\chi^2)}>0.
 \label{eq:cmetric-turning-latent-heat}
\end{align}
The latent heat decreases strictly from infinity to zero as
\(\mu\) runs from \(0\) to \(\mu_+\).  At the two ends,
\begin{align}
 \frac{\mathcal L_{\rm turn}}Q
 &=\frac1\mu-2-5\mu+O(\mu^2),
 \qquad \mu\downarrow0,
 \label{eq:cmetric-latent-small-mu}\\
 \frac{\mathcal L_{\rm turn}}Q
 &=\mathcal A_+(\mu_+-\mu)^{1/2}
 +O\!\left((\mu_+-\mu)^{3/2}\right),
 \label{eq:cmetric-latent-upper-law}\\
 \mathcal A_+&=
 \left[
 \frac{12(1+\chi_c)^3}
 {\chi_c(1+\chi_c^2)(1-\chi_c)}
 \right]^{1/2}
 \simeq10.8537 .
 \nonumber
\end{align}
\end{corollary}

\begin{proof}
Equation~\eqref{eq:cmetric-turning-roots} gives
\begin{equation}
 w_+-w_-=
 \frac{\chi^2(1-\chi^2)\sqrt{D_\chi}}{A_\chi}.
 \label{eq:cmetric-turning-root-gap}
\end{equation}
Substitution in Eq.~\eqref{eq:cmetric-turning-entropy-jump} proves
Eq.~\eqref{eq:cmetric-turning-entropy-gap}.  Equal Gibbs values at a common
temperature give \(M_L-M_S=T_{\rm turn}(S_L-S_S)\); inserting
\(QT_{\rm turn}=\mu/\pi\) and
\(\mu=\chi/(1+\chi)^2\) proves
Eq.~\eqref{eq:cmetric-turning-latent-heat}.  Its logarithmic derivative with
respect to \(\chi\) is
\begin{equation}
 \frac{D_\chi'}{2D_\chi}-\frac1\chi
 -\frac{2\chi}{1+\chi^2}<0,
 \qquad D_\chi'=-12\chi(1+\chi^2),
\end{equation}
so the decrease is strict.  The lower expansion follows by inverting
\(\mu=\chi/(1+\chi)^2\).  At the upper end,
\(D_\chi=-D'_{\chi_c}(\chi_c-\chi)+
O((\chi_c-\chi)^2)\) and
\(\mu_+-\mu=[(1-\chi_c)/(1+\chi_c)^3]
(\chi_c-\chi)+O((\chi_c-\chi)^2)\), which gives
Eq.~\eqref{eq:cmetric-latent-upper-law}.
\end{proof}

\begin{corollary}[Reentrant winner ordering at every physical turn]
\label{cor:cmetric-reentrant-ordering}
Fix \(0<\mu<\mu_+\).  There is an \(\varepsilon_T(\mu)>0\) such that, for
\[
 T_{\rm turn}<T<T_{\rm turn}+\varepsilon_T(\mu),
\]
the fixed-\((T,Q,\mu)\) pressure path meets the physical Maxwell component
at exactly two nearby pressures
\begin{equation}
 P_1(T,\mu)<P_{\rm turn}(\mu)<P_2(T,\mu).
\end{equation}
As pressure increases through these two values, the global minimum in the
physical sector is selected in the order
\begin{equation}
 S\ \longrightarrow\ L\ \longrightarrow\ S .
 \label{eq:cmetric-local-reentrant-sequence}
\end{equation}
Both crossings are first order.  The local reentrant window exists for every
positive tension below \(\mu_+\); its unscaled temperature depth collapses
as \(O(\mu^3/Q)\) when \(\mu\downarrow0\).
\end{corollary}

\begin{proof}
The implicit Jacobian of the two-phase equations is nonzero and
Eq.~\eqref{eq:cmetric-turning-positive-curvature} makes the turning a strict
minimum of the coexistence temperature.  The Morse lemma on the regular
Maxwell component therefore gives the two stated intersections.  The exact
global-value gap and all physical margins are strict at the turning pair, so
they persist in a neighborhood.  With
\(\Delta G=G_S-G_L\), the envelope identity gives
\begin{equation}
 \left.\frac{\partial\Delta G}{\partial T}\right|_{P,Q,\mu}
 =-(S_S-S_L)>0.
\end{equation}
Thus \(L\) wins above the coexistence temperature and \(S\) below it.
A horizontal line immediately above a strict minimum lies below the
coexistence curve outside its two intersections and above it between them,
which proves Eq.~\eqref{eq:cmetric-local-reentrant-sequence}.  The latent
heat in Eq.~\eqref{eq:cmetric-turning-latent-heat} is nonzero throughout the
open interval.  Finally,
Eq.~\eqref{eq:cmetric-temperature-collapse} gives the small-tension scale.
\end{proof}

At the turning condition \(\Delta V=0\), the vanishing coexistence slope
occurs between two separated global minima.  The entropy jump and latent heat
are finite at every point of the open locus, vanish with the square-root
merger law as \(\mu\uparrow\mu_+\), and diverge as \(\mu\downarrow0\).
The lower limit is therefore a singular boundary-layer endpoint.

The discriminant of Eq.~\eqref{eq:cmetric-turning-polynomial} is
\begin{equation}
 \operatorname{disc}_wF_\chi
 =\chi^4(1-\chi^2)^2(1-6\chi^2-3\chi^4).
 \label{eq:cmetric-turning-discriminant}
\end{equation}
This connected exact turning locus exists for
\begin{equation}
 0<\chi<\chi_c,
 \qquad
 \chi_c^2=\frac{2\sqrt3-3}{3}.
 \label{eq:cmetric-chi-window}
\end{equation}
The map \(\chi\mapsto\mu\) is strictly increasing on this interval.  At
\(\chi=\chi_c\), the two roots merge at the ordinary critical endpoint and
give
\begin{align}
 0&<\mu<\mu_+,\nonumber\\
 \mu_+&=
 \frac{\sqrt{2\sqrt3-3}}
 {2\left(1+\sqrt{2\sqrt3-3}\right)},\nonumber\\
 \mu_+&\simeq0.202602 .
 \label{eq:cmetric-mu-plus-strong}
\end{align}
Equivalently, \(\mu_+\) is the unique root in \(0<\mu<1/4\) of
\begin{equation}
 p(\mu)=64\mu^4-48\mu^2+24\mu-3=0.
 \label{eq:cmetric-mu-polynomial}
\end{equation}
Indeed, \(p(0)<0<p(1/4)\), while
\(p'(\mu)=8(32\mu^3-12\mu+3)>0\) on the physical tension interval.

Set \(h=\sqrt{2\sqrt3-3}\).  The endpoint data are
\begin{align}
 z_c&=2-\sqrt3,
 &q_c&=\frac{1}{(1+h)(1+\sqrt3)},\nonumber\\
 u_c&=\frac{h}{\sqrt3},
 &(A\ell)_c&=\frac{1+\sqrt3}{4},\nonumber\\
 \alpha_c&=\sqrt3-1,
 &\frac{P_c}{P_t}&=\frac2{\sqrt3},\nonumber\\
 T_c&=\frac{\mu_+}{\pi Q},
 &S_c&=\frac{\pi Q^2}{q_c},\nonumber\\
 M_c&=\frac{\mu_+Q}{q_c},
 &G_c&=0 .
 \label{eq:cmetric-upper-exact-point}
\end{align}
The number-field calculation in \ref{app:cmetric-identities}
verifies
\begin{equation}
 \mathscr D_z\tau=\mathscr D_z^2\tau
 =\mathscr D_z\nu=0,
 \qquad
 \mathscr D_z^3\tau\ne0,
 \qquad
 \mathscr D_zs\ne0.
 \label{eq:cmetric-upper-certificate}
\end{equation}
Hence the upper endpoint of this turning locus is an
interior ordinary critical point with a vanishing limiting Clapeyron slope;
the boundary corner \(X\) governs its lower closure.

There is no positive lower tension on this exact connected turning locus.
Inverting Eq.~\eqref{eq:cmetric-chi-mu} gives
\begin{equation}
 \chi=\mu+2\mu^2+5\mu^3+14\mu^4+O(\mu^5),
\end{equation}
and therefore
\begin{equation}
 (P_{\rm turn}-P_t)Q^2
 =\frac{3\mu^4}{8\pi}
 \left(1+4\mu+14\mu^2+O(\mu^3)\right).
 \label{eq:cmetric-pressure-collapse}
\end{equation}
Its closure has lower endpoint \(\mu_-=0\).  The quartic collapse explains why a
finite-resolution phase diagram can suggest a nonzero lower threshold for
the turn.  The unrestricted classification in
Theorem~\ref{thm:cmetric-exact-turning-main} now excludes a second physical
turning component that could supply such a threshold.

The algebraic pair is physical throughout
Eq.~\eqref{eq:cmetric-chi-window}.  Since
\(A_\chi=(1-\chi^2)^2+\chi^6>0\), \(B_\chi<0\), and
\begin{equation}
 w_-+w_+-\chi^2
 =-\frac{\chi^6(\chi^2+2)}{A_\chi}<0,
\end{equation}
one has \(0<w_-<w_+<\chi^2\).  On either root,
\begin{align}
 \alpha_\pm&=
 \frac{\sqrt{w_\pm}(1-\chi^2)(\chi^2-w_\pm)}
 {(1+w_\pm)\chi^4}>0,\nonumber\\
 g(-1)&=\frac{(1+w_\pm)(1-\chi)^2}{1+\chi^2}>0.
 \label{eq:cmetric-turning-margins}
\end{align}
The outer-horizon, slow-acceleration, entropy-regularity, local-stability, and
global-minimum
proofs are given in
\ref{app:cmetric-turning-certificate}.  They prove that the two
\(G=0\) states are strict canonical minima, that neither reaches a physical
face before the critical merger, and that every remaining admissible
positive-temperature Euclidean equilibrium has positive potential.  Thus
Eq.~\eqref{eq:cmetric-turning-data} lies in the physical Maxwell set and
satisfies the admissibility, stability, and global-minimum conditions beyond
critical-value equality.
Figures~\ref{fig:cmetric-exact-reentrant-wedge} and
\ref{fig:cmetric-threshold-boundary} display the exact locus, its singular
small-tension scales, and the two endpoint limits.

\begin{figure}[t]
 \centering
 \includegraphics[width=\linewidth]{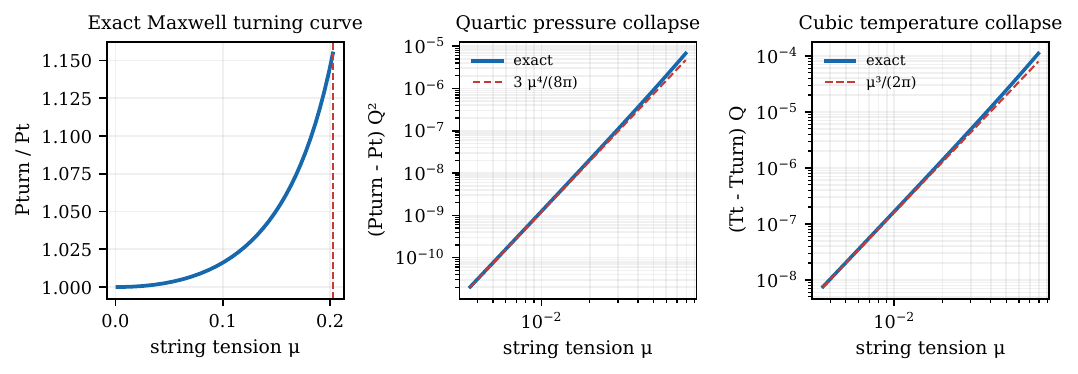}
 \caption{Exact turning locus and its small-tension scales.  The left
 panel shows the turning pressure relative to the snapping pressure.  The
 middle and right panels resolve the small-tension limits
 \(P_{\rm turn}-P_t\sim3\mu^4/(8\pi Q^2)\) and
 \(T_t-T_{\rm turn}\sim\mu^3/(2\pi Q)\).  Solid curves follow from the
 closed expressions; dashed curves are their leading asymptotic terms.}
 \label{fig:cmetric-exact-reentrant-wedge}
\end{figure}

\begin{figure}[t]
 \centering
 \includegraphics[width=\linewidth]{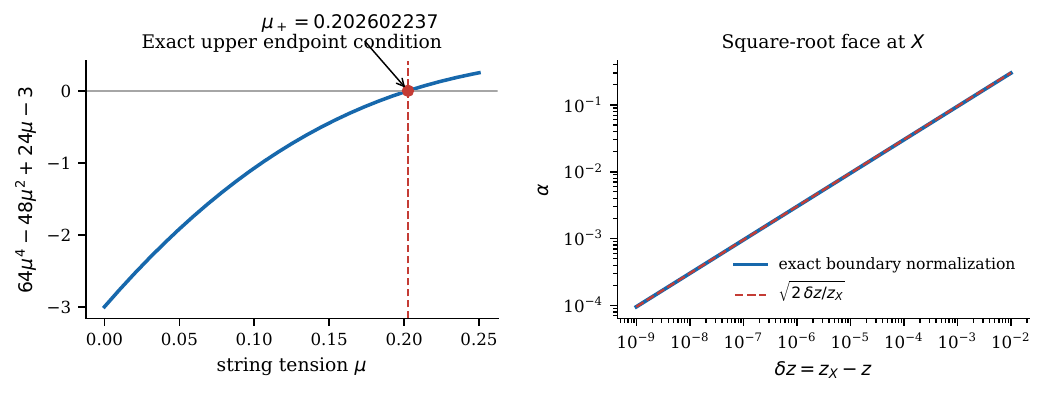}
 \caption{Exact endpoint and boundary behavior.  The left panel isolates
 the unique physical root of Eq.~\eqref{eq:cmetric-mu-polynomial}.  The right
 panel verifies
 \(\alpha\sim\sqrt{2(z_X-z)/z_X}\) on the corner slice \(q=\mu\).}
 \label{fig:cmetric-threshold-boundary}
\end{figure}

\subsection{Bicritical closure}
\label{subsec:cmetric-bicritical-closure}

The bicritical junction was identified in
Ref.~\cite{AbbasvandiEtAl2019Snapping}.  We derive its one-sided Maxwell
limit in closed form.  The corner temperature must be distinguished from the
reservoir temperature selected by a physical approach.  Parametrize the
bicritical limit by
\begin{equation}
 \mu=\frac{x}{\mathcal D_x},
 \qquad
 \mathcal D_x=1+2x-x^2,
 \qquad 0<x<\sqrt2-1.
 \label{eq:cmetric-bicritical-x}
\end{equation}
At \(q=\mu\), the limiting \(X\) phase and the coexisting large phase have
the common horizon coordinate \(u=x\), but
\begin{equation}
 z_X^2=\frac{x^2}{1-3x^2+x^4},
 \qquad
 z_L^2=\frac{x^4}{1-2x^2}.
 \label{eq:cmetric-bicritical-z}
\end{equation}
The Maxwell approach from \(q>\mu\) selects
\begin{align}
 \tau_t&=\frac{4\mu}{\sqrt{1-x^2}},
 &s_X&=\mathcal D_x,
 &s_L&=\frac{\mathcal D_x}{x^2},\nonumber\\
 \mathfrak g_X&=\mathfrak g_L
 =-\frac{x}{\sqrt{1-x^2}}.
 \label{eq:cmetric-bicritical-data}
\end{align}
At \(X\) itself, \(H_u\) and \(\alpha\) vanish simultaneously, so
\(T_H\propto H_u/\alpha\) depends on the relative approach weights.
Equation~\eqref{eq:cmetric-bicritical-data} gives the finite value selected
along the admissible Maxwell path.

Here \(G_X\) denotes the one-sided limiting critical value of the
disappearing small-black-hole sheet as \(q\downarrow\mu\) from the slowly
accelerating region; the excluded corner is not promoted to an independent
boundary phase.  Its volume diverges at the corner, so no smooth expansion in
an arbitrary pressure-transverse direction is available.  The exact statement
concerns the one-sided fixed-pressure cut \(P=P_t\).  Define the limiting
value continuation \(G_X^{\rm lim}(T;P_t)=-TS_X\) and
\(J=G_L-G_X^{\rm lim}\).  With
\(\delta\tau=\tau_t-\tau>0\), the envelope identity on the regular large
branch gives
\begin{equation}
 \frac{J}{Q}
 =\frac{\mathcal D_x(1-x^2)}{4x^2}\,\delta\tau+O(\delta\tau^2).
 \label{eq:cmetric-exact-jump-closure}
\end{equation}
The coefficient is strictly positive, so this one-sided continued value gap
vanishes linearly at the bicritical junction.  Determining the approach
exponent of the physical zeroth-order wall away from this cut requires
pressure-transverse data.  In physical variables,
\begin{align}
 J&=(S_L-S_X)(T_t-T)
 +O\!\left((T_t-T)^2\right),\nonumber\\
 S_L-S_X&=\pi Q^2\frac{\mathcal D_x(1-x^2)}{x^2}.
 \label{eq:cmetric-jump-physical}
\end{align}
The boundary-to-turning temperature depth is
\begin{equation}
 Q(T_t-T_{\rm turn})
 =\frac{\mu}{\pi}
 \left(\frac{1}{\sqrt{1-x^2}}-1\right)
 =\frac{\mu^3}{2\pi}+\frac{2\mu^4}{\pi}+O(\mu^5).
 \label{eq:cmetric-temperature-collapse}
\end{equation}

\subsection{Small-tension Maxwell boundary layer}
\label{subsec:cmetric-maxwell-blowup}

The fourth-order collapse of the pressure interval sets the outer scale of a
singular boundary layer
that contains an entire segment of the physical coexistence curve.  The two
coexisting black holes occupy different state-space charts in this limit.
One has \(z=O(\mu^2)\) and an entropy of order \(\mu^{-2}\), whereas the
boundary-linked phase has \(z=O(\mu)\), finite entropy, and a lapse factor of
order \(\mu\).  Resolving both charts gives a closed limiting Maxwell curve.

Put \(p=q^2\), and let \(\tau_t\) denote the temperature selected by the
one-sided Maxwell approach to the bicritical boundary.  Thus, with \(x\)
defined by
\begin{equation}
 \mu=\frac{x}{1+2x-x^2},
 \qquad
 \tau_t=\frac{4\mu}{\sqrt{1-x^2}},
 \label{eq:cmetric-blowup-boundary-temperature}
\end{equation}
one has
\begin{equation}
 \tau_t=4\mu+2\mu^3+8\mu^4+\frac{43}{2}\mu^5+O(\mu^6).
 \label{eq:cmetric-blowup-boundary-series}
\end{equation}
Introduce the boundary-layer controls
\begin{align}
 p&=\mu^2\bigl(1+\mu^2r^2\bigr),\nonumber\\
 \widehat P&=\frac{p-\mu^2}{\mu^4}=r^2,\nonumber\\
 \widehat T&=\frac{\tau_t-\tau}{2\mu^3}
 =\frac{2\pi Q}{\mu^3}(T_t-T).
 \label{eq:cmetric-blowup-controls}
\end{align}

\begin{theorem}[C-metric Maxwell boundary layer]
\label{thm:cmetric-maxwell-blowup}
Fix \(0<r_0<r_1<\infty\).  For sufficiently small positive tension, the
physical Maxwell component whose closure contains the bicritical boundary
has a unique two-phase solution in the blow-up charts below for every
\(r\in I=[r_0,r_1]\).  The notation \(O_I\) means uniformity on this fixed
compact interval.  Denote the low-\(z\), large-entropy phase by \(L\) and the
high-\(z\), boundary-linked phase by \(S\).  Uniformly on this interval,
\begin{align}
 u_L={}&\mu+2\mu^2+(3+2r)\mu^3
 +(2+8r)\mu^4+O_I(\mu^5),\nonumber\\
 z_L={}&\mu^2+4\mu^3+
 \left(11+4r+\frac{r^2}{2}\right)\mu^4\nonumber\\
 &+\left(24+24r+2r^2\right)\mu^5+O_I(\mu^6),
 \label{eq:cmetric-blowup-large-state}\\
 u_S={}&\mu+2\mu^2+(3+2r)\mu^3
 +(2+8r)\mu^4+O_I(\mu^5),\nonumber\\
 z_S={}&\mu+2\mu^2+\frac92\mu^3+11\mu^4\nonumber\\
 &+\left(\frac{227}{8}-\frac52r^2\right)\mu^5
 +\left(\frac{303}{4}-15r^2\right)\mu^6+O_I(\mu^7).
 \label{eq:cmetric-blowup-small-state}
\end{align}
Equivalently, in the polynomial coordinate \(w=z^2\),
\begin{align}
 w_L={}&\mu^4+8\mu^5+(r^2+8r+38)\mu^6\nonumber\\
 &+(8r^2+80r+136)\mu^7+O_I(\mu^8),\nonumber\\
 w_S={}&\mu^2+4\mu^3+13\mu^4+40\mu^5\nonumber\\
 &+(121-5r^2)\mu^6+(364-40r^2)\mu^7+O_I(\mu^8).
 \label{eq:cmetric-blowup-w-states}
\end{align}
Their common temperature and critical value satisfy
\begin{align}
 \tau_M(r,\mu)
 & =4\mu+2(1-r)^2\mu^3+8(1-r)\mu^4+O_I(\mu^5),
 \label{eq:cmetric-blowup-tau}\\
 \mathfrak g_M(r,\mu)
 & =\mu(r-1)-2\mu^2+O_I(\mu^3).
 \label{eq:cmetric-blowup-gibbs}
\end{align}
Consequently the rescaled coexistence curve converges in \(C^2\) on compact
subintervals of \(r>0\) to
\begin{equation}
 \boxed{\widehat T=2r-r^2+4\mu r+O_I(\mu^2),
 \qquad \widehat P=r^2.}
 \label{eq:cmetric-universal-maxwell-curve}
\end{equation}
In particular, its limiting profile is
\begin{equation}
 \widehat T=2\sqrt{\widehat P}-\widehat P.
 \label{eq:cmetric-universal-maxwell-profile}
\end{equation}
The two phases are strict canonical minima and are the global minima in the
physical sector of the fixed-\((P,Q,\mu)\) ensemble.  Across their coexistence
temperature, \(L\)
is selected on the high-temperature side and \(S\) on the low-temperature
side.
\end{theorem}

Substitution of the two chart ans\"atze into the unsquared Maxwell equations
gives a leading algebraic system with a unique physical solution and a
nonzero Jacobian.  A uniform implicit-function argument supplies the stated
series and their \(C^2\) control, while the compactness analysis exhausts the
remaining positive-temperature equilibria and fixes their value order.  The
details are given in \ref{app:cmetric-blowup-proof}; the limiting profile and
volume inversion are shown in Fig.~\ref{fig:cmetric-universal-maxwell-blowup}.

\begin{figure*}[t]
 \centering
 \includegraphics[width=0.86\textwidth]{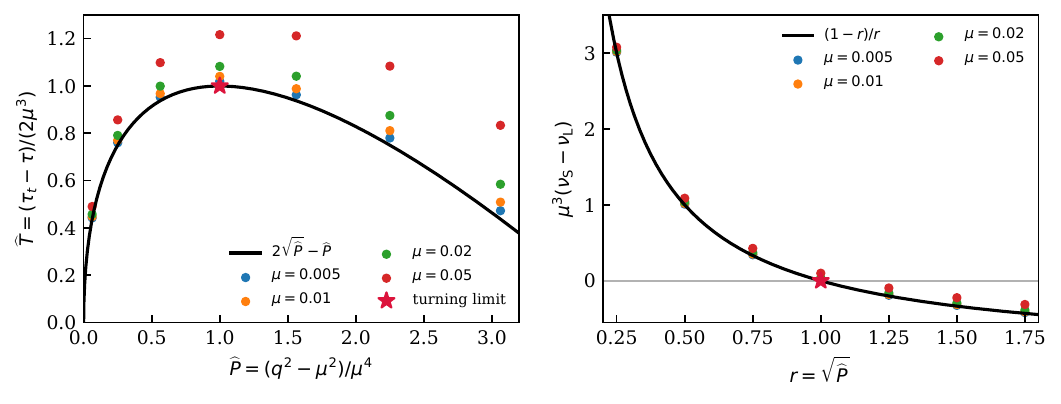}
 \caption{The small-tension Maxwell boundary layer.  The left panel shows
 the boundary-referenced coexistence curve.  Finite-tension solutions of the
 unsquared two-phase equations converge to
 \(\widehat T=2\sqrt{\widehat P}-\widehat P\); the star marks its limiting
 turn.  The right panel shows the rescaled thermodynamic-volume jump.  Its
 limiting zero is the star at \(r=1\), while the finite-tension zeros approach
 it from \(r>1\); the entropy jump remains negative.  The solid curves are the asymptotic formulas in
 Eqs.~\eqref{eq:cmetric-universal-maxwell-profile} and
 \eqref{eq:cmetric-blowup-jumps}.}
 \label{fig:cmetric-universal-maxwell-blowup}
\end{figure*}

The turn is generated by a thermodynamic-volume inversion while the entropy
ordering stays fixed.  On the Maxwell pair,
\begin{equation}
 s_S-s_L=-\frac1{\mu^2}\bigl(1+O_I(\mu)\bigr),
 \qquad
 \nu_S-\nu_L=\frac{1-r}{r\mu^3}+O_I(\mu^{-2}).
 \label{eq:cmetric-blowup-jumps}
\end{equation}
The large-entropy phase has the geometric volume
\(\nu_L=\mu^{-3}(1+O_I(\mu))\).  In the boundary-linked phase, the
acceleration contribution to the thermodynamic volume is amplified by
\(\alpha_S\sim\mu r\), giving
\(\nu_S=(r\mu^3)^{-1}(1+O_I(\mu))\).  Consequently, at leading order,
\(S\), despite its much smaller entropy, has the larger thermodynamic volume
for \(0<r<1\).  The limiting scaled volumes agree at \(r=1\), and their order
is restored for \(r>1\).  At finite tension the exact equality occurs at
\(r=r_{\rm turn}=1+2\mu+O(\mu^2)\).  The
dimensionless Clapeyron equation becomes
\begin{equation}
 \frac{\mathrm d\tau_M}{\mathrm dp}
 =\frac{2(r-1)}{r\mu}+O_I(1)
 =2\frac{\nu_S-\nu_L}{s_S-s_L},
 \label{eq:cmetric-blowup-clapeyron}
\end{equation}
so the change of coexistence slope is caused by the volume inversion.

The C-metric asymptotics realize
the selected-sheet quotient of
Corollary~\ref{thm:maxwell-corner-universality} with
\begin{equation}
 (a,b,c;m,n,k)=(4,2,3;2,1,1),
 \qquad \epsilon=\mu,\qquad \rho=r,
 \label{eq:cmetric-universality-weights}
\end{equation}
and
\begin{equation}
 c_p=c_s=c_v=\rho_*=1,\qquad \lambda=2.
 \label{eq:cmetric-universality-data}
\end{equation}
The selected physical face for this comparison is the normalized-time face.
Its strict transforms are taken separately in the two state charts. Using
Eq.~\eqref{eq:cmetric-blowup-lapse} and the solved pressure coefficient gives,
uniformly on every \(I\Subset(0,\infty)\),
\begin{equation}
 \beta_S^{\mathrm{res}}
 :=\frac{\alpha_S}{\mu}
 =r\bigl[1+O_I(\mu)\bigr],
 \qquad
 \beta_L^{\mathrm{res}}
 :=\alpha_L
 =1+O_I(\mu^2).
 \label{eq:cmetric-resolved-sheetwise-face}
\end{equation}
With \(\rho=r\), these margins give the labeled transverse sheetwise
incidence \(\mathfrak f_{S|L}\): the \(S\) phase reaches the selected face,
whereas \(L\) remains in the interior. The bulk-extremal and
slow-acceleration faces that also meet at \(X\) are additional fixed
incidence data. Their margins are strictly positive on every such \(I\), as
shown in \ref{app:cmetric-blowup-proof}; their joint contact at
\(r=0\) lies outside the single-face Newton datum.
The two-chart desingularization determines all the orders above.  The
pressure Newton edge gives
\(p-\mu^2=\mu^4r^2\) exactly.  The low-\(z\) chart carries the large entropy
\(s_L\sim\mu^{-2}\), whereas the boundary-linked phase has finite entropy.
Both geometric volumes scale as \(\mu^{-3}\), but the boundary-linked lapse
contributes the additional factor \(r^{-1}\).  Their difference has the
two-term Newton-face form \(r^{-1}-1\), as displayed in
Eq.~\eqref{eq:cmetric-blowup-jumps}.
Equations~\eqref{eq:cmetric-blowup-boundary-series} and
\eqref{eq:cmetric-blowup-tau} give the boundary matching condition.
Thus \(\iota=1\), \(\sigma=3\), and
Corollary~\ref{thm:maxwell-corner-universality} gives
\begin{equation}
 \widehat P=r^2,
 \qquad
 \widehat T\longrightarrow2r-r^2
 =2\sqrt{\widehat P}-\widehat P.
 \label{eq:cmetric-normal-form-specialization}
\end{equation}
It also fixes the limiting turning curvature to \(-1/2\).  Equation
\eqref{eq:cmetric-universal-maxwell-profile} therefore realizes the
\((m,n,k)=(2,1,1)\) class.  Its simple turn is stable on every compact
resolved interval under the remainder class of
Corollary~\ref{thm:maxwell-corner-universality} below.  Families with different
corner weights or a different leading Newton polynomial belong to different
classes.

\begin{proposition}[Reduced stationary-saddle barrier]
\label{prop:cmetric-reduced-barrier}
Let \(\mathfrak g_3\) be the critical value of the third equilibrium in the
boundary layer, and let \(\mathfrak g_M\) be the common value of \(L\) and
\(S\).  Uniformly for \(r\) in compact subsets of \((0,\infty)\),
\begin{align}
 \Delta\mathfrak g^\ddagger
 &:=\mathfrak g_3-\mathfrak g_M
 =\frac{2}{27\mu}+O_I(1),
 \label{eq:cmetric-reduced-gibbs-barrier}\\
 \Delta I_E^{\rm red}
 &:=\beta Q\Delta\mathfrak g^\ddagger
 =\frac{2\pi Q^2}{27\mu^2}
 +O_I\left(\frac{Q^2}{\mu}\right).
 \label{eq:cmetric-reduced-action-barrier}
\end{align}
At the same time,
\begin{equation}
 S_L-S_S=\frac{\pi Q^2}{\mu^2}\bigl[1+O_I(\mu)\bigr],
 \qquad
 T_M(S_L-S_S)=\frac{Q}{\mu}\bigl[1+O_I(\mu)\bigr].
 \label{eq:cmetric-boundary-layer-strength}
\end{equation}
\end{proposition}

\begin{proof}
The complete stationary classification in
\ref{app:cmetric-blowup-proof} gives
\(\mathfrak g_3=2/(27\mu)+O_I(1)\) and
\(\mathfrak g_M=O_I(\mu)\).  It also places the third equilibrium between
the two minima in the monotone entropy coordinate and gives it the opposite
canonical Morse index.  Hence its value is the barrier of the reduced
one-state potential.  Since
\(\tau_M=4\mu+O_I(\mu^3)\) and
\(\beta=4\pi Q/\tau_M\),
Eq.~\eqref{eq:cmetric-reduced-action-barrier} follows.  The remaining two
relations use Eq.~\eqref{eq:cmetric-blowup-jumps} and
\(T_M=\mu/(\pi Q)+O_I(\mu^3/Q)\).
\end{proof}

The pressure and temperature widths of the coexistence wedge vanish in the
unscaled control plane, whereas the entropy discontinuity, latent heat, and
stationary barrier grow.  Equation
\eqref{eq:cmetric-reduced-action-barrier} compares stationary values along
the constrained black-hole family.  A decay rate in the full gravitational
configuration space would additionally require the relevant off-shell saddle,
its negative modes, and the fluctuation determinant.

Equation~\eqref{eq:cmetric-universal-maxwell-profile} is smooth in the
natural boundary coordinate \(r=\sqrt{\widehat P}\), but has an infinite
slope at \(\widehat P=0\).  This square-root edge records the degeneration
of the normalized lapse at the bicritical boundary and lies outside the
interior critical locus.  The profile has a nondegenerate maximum of
\(\widehat T\), equivalently a nondegenerate minimum of the coexistence
temperature, at \(r=1\) in the limiting profile.  At finite tension, the
first correction moves it to
\begin{align}
 r_{\rm turn}&=1+2\mu+O(\mu^2),\nonumber\\
 \widehat P_{\rm turn}&=1+4\mu+O(\mu^2),\nonumber\\
 \widehat T_{\rm turn}&=1+4\mu+O(\mu^2).
 \label{eq:cmetric-blowup-turning-location}
\end{align}
Using the exact parameter \(\chi\) on the turning locus gives the stronger
curvature identity
\begin{equation}
 \left.
 \frac{\mathrm d^2\widehat T}{\mathrm d\widehat P^2}
 \right|_{\rm turn}
 =-\frac{(1-\chi)^2}
 {2(1+\chi)^2(1+\chi^2)^3}
 \longrightarrow-\frac12.
 \label{eq:cmetric-blowup-exact-curvature}
\end{equation}
Thus the \(O(\mu^4)\) pressure width and the \(O(\mu^3)\) temperature depth
are the two projections of a resolved, nondegenerate Maxwell geometry.  The
expansion is uniform only when \(r\) remains in a compact subset of
\((0,\infty)\).  Its closure at \(r=0\) is fixed by the exact bicritical
data, whereas a limit in which \(r\) tends to zero with \(\mu\), or grows
without bound, requires a different scaling.  Remote portions of the
Maxwell surface therefore require separate charts.
Theorem~\ref{thm:cmetric-exact-turning-main} proves that none of them contains an
additional finite-separation physical turn.

\paragraph{Numerical wall atlas.}

The exact curve lies inside the three-dimensional wall complex shown in
Fig.~\ref{fig:cmetric-wall-atlas}.  Each fixed-tension section includes all
real quartic horizon roots in \(0<u<1\), subject to the angular, horizon,
temperature, stability, and slow-acceleration inequalities.  Maxwell values
are then compared with every admissible Euclidean equilibrium at the same
controls.

\begin{figure*}[t]
 \centering
 \includegraphics[width=0.90\textwidth]{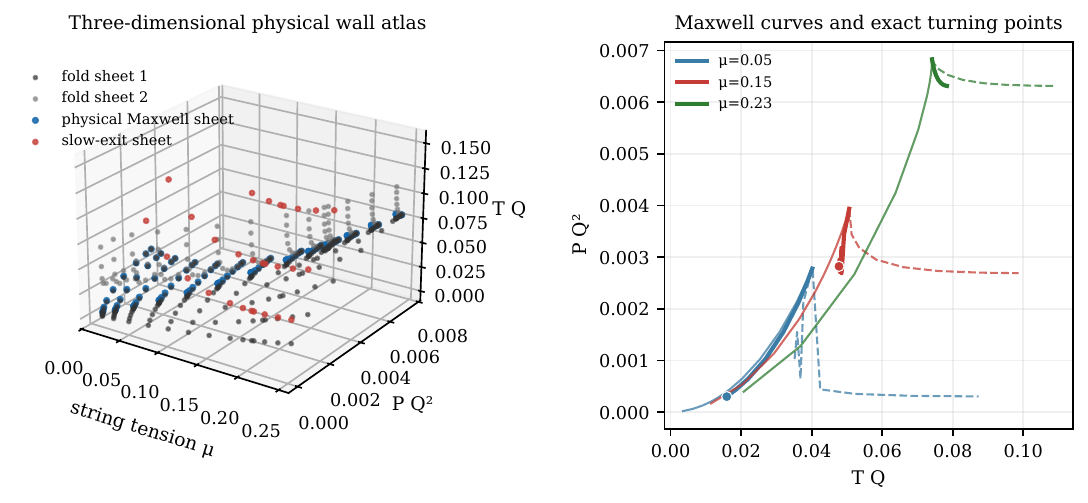}
 \caption{Sampled physical wall atlas at \(Q=1\).  The left panel shows the
 two fold sheets, the physical Maxwell sheet, and the slow-acceleration part
 of the admissibility-exit wall in \((T,P,\mu)\); the extremal part at
 \(T=0\) is omitted for visibility, and the displayed slow-exit points obey
 \(TQ<0.16\).  The right panel gives three
 fixed-tension sections.  Thick
 solid curves are physical Maxwell values, thin solid and dashed curves are
 spinodals, and marked points are the exact turnings.}
 \label{fig:cmetric-wall-atlas}
\end{figure*}

For \(q<\mu\), the connected state curve exits through the
slow-acceleration face.  For \(q>\mu\), it terminates at bulk extremality.
The equality \(q=\mu\) is the corner \(X\), where these faces meet the
normalization face.  On this slice, the lower temperature \(T_0\) used in the
standard phase diagram is the surviving canonical fold away from the
boundary-incidence stratum.  The upper value \(T_t\) is the
Maxwell--admissibility junction whose exact approach data are given in
Eq.~\eqref{eq:cmetric-bicritical-data}.
Table~\ref{tab:cmetric-chamber-ledger} summarizes the equilibrium structure,
phase selection, and active walls across these control regions.

\begin{table*}[t]
\centering
\begin{tabularx}{\textwidth}{@{}
 >{\raggedright\arraybackslash}p{0.20\textwidth}
 >{\raggedright\arraybackslash}p{0.27\textwidth}
 >{\raggedright\arraybackslash}p{0.27\textwidth}
 >{\raggedright\arraybackslash}X@{}}
\toprule
Control region & Equilibrium structure & Physical selection & Active wall\\
\midrule
\(q<\mu\), below the surviving fold
& No admissible Euclidean black-hole equilibrium
& No-black-hole chamber
& Slow-exit and fold boundaries\\
\(q<\mu\), above the surviving fold
& One minimum and one maximum
& Unique stable black hole
& One fold sheet\\
\(q=\mu\)
& Critical slice through the corner \(X\)
& Maxwell--admissibility junction at \(T_t\)
& Multi-face admissibility corner\\
\(\mu<q<q_{\rm crit}(\mu)\), outside the spinodals
& One admissible minimum
& Unique global phase
& None\\
\(\mu<q<q_{\rm crit}(\mu)\), between the spinodals
& Two minima separated by one maximum
& Lower critical value; equality on the physical Maxwell sheet
& Two fold sheets and one Maxwell sheet\\
\(q=q_{\rm crit}(\mu)\)
& Coalescing minimum--maximum pair
& Ordinary critical endpoint
& Cusp edge of the fold surface\\
\(q>q_{\rm crit}(\mu)\)
& One admissible Euclidean minimum
& Unique global phase
& None\\
\bottomrule
\end{tabularx}
\caption{Chamber structure in the sampled control window of the standard
single-string fixed-\((P,Q,\mu)\) ensemble.  The state-space exit changes
from the slow-acceleration face to bulk extremality at \(q=\mu\).}
\label{tab:cmetric-chamber-ledger}
\end{table*}

The exact elimination classifies every physical Maxwell turn at which the
phases remain distinct and proves global selection on the resulting locus.  On the displayed
fixed-tension sections intersecting \(0<\mu<\mu_+\), the Maxwell branch
emanating from the bicritical boundary meets the exact turning point.  The
atlas records the surrounding walls; uniqueness of the turn follows from the
positive-domain resultant.

\section{Resolved Maxwell--boundary profiles}
\label{sec:maxwell-boundary-germs}

\subsection{Value-gap reduction and resolved equivalence}
\label{subsec:maxwell-boundary-gap}

The difference of two competing critical values contains the information
needed to determine two-phase coexistence.  We study its leading profile on
the exceptional half-line of a prescribed corner blow-up.  The data consist
of two selected critical-value sheets, the physical admissibility margin
restricted to those sheets, and a fixed reservoir polarization.  This
selected-sheet problem is narrower than the right-equivalence problem for a
full boundary function germ \cite{Arnold1978Boundary} and differs from the
algebraic stationary bifurcations described by \(A\)-discriminants
\cite{NisseLimChang2024ADiscriminants}.

Through the Clapeyron equation, pressure contact of order \(m\) and a
volume-jump pole of order \(n\) give the leading factor
\(r^{m-n-1}R(r)\) after differentiation with respect to a resolved boundary
coordinate.  Positive simple roots of
\(R\) are therefore volume inversions and regular turns of the coexistence
curve.  The C-metric boundary layer realizes the primitive binomial case
\(R(r)=1-r\), \((m,n,k)=(2,1,1)\), whose weighted primitive is
\(2r-r^2\).

Let \(x_S\) and \(x_L\) be two disjoint nondegenerate local-minimum sheets
of a fixed thermodynamic potential. In reservoir coordinates
\((\vartheta,\pi)\), write their critical values as
\begin{equation}
 g_i(\vartheta,\pi;\epsilon)
 =\mathcal F(x_i(\vartheta,\pi;\epsilon),
             \vartheta,\pi;\epsilon),
 \qquad i=S,L,
 \label{eq:corner-critical-value-sheets}
\end{equation}
and define the oriented value gap
\begin{equation}
 D(\vartheta,\pi;\epsilon)=g_L-g_S.
 \label{eq:corner-value-gap}
\end{equation}
The jump convention agrees with the black-hole application below,
\(\Delta s=s_L-s_S\) and \(\Delta v=v_L-v_S\).  The labels are chosen so
that \(\Delta s>0\) on the Maxwell component.  With the pressure
normalization used below, the envelope identities read
\begin{equation}
 D_{\vartheta}=-\Delta s,
 \qquad
 D_{\pi}=\lambda\Delta v,
 \qquad \lambda>0.
 \label{eq:corner-gap-envelope}
\end{equation}
The constant \(\lambda\) only records the relative normalization of the
two reservoir coordinates. Equation
\eqref{eq:corner-gap-envelope} gives, on a regular component of \(D=0\),
\begin{equation}
 \frac{\dd\vartheta_M}{\dd\pi_M}
 =\lambda\frac{\Delta v}{\Delta s}.
 \label{eq:corner-universality-clapeyron}
\end{equation}

\begin{definition}[Resolved Maxwell--boundary profile]
\label{def:resolved-two-phase-germ}
A resolved Maxwell--boundary profile consists of the two
nondegenerate
minimum sheets in Eq.~\eqref{eq:corner-critical-value-sheets}, a regular
Maxwell component \(D=0\), and an oriented physical face
\(\beta=0\), \(\beta>0\), together with a corner chart
\begin{equation}
 (\epsilon,\rho)\longmapsto
 (\vartheta_M(\epsilon,\rho),\pi_M(\epsilon,\rho)),
 \qquad \epsilon>0,\quad \rho>0.
 \label{eq:corner-resolved-chart}
\end{equation}
For a control point \(\zeta=(\vartheta,\pi;\epsilon)\), define the selected
sheetwise admissibility margins
\begin{equation}
 b_i(\zeta)=\beta\bigl(x_i(\zeta),\zeta\bigr),
 \qquad i=S,L.
 \label{eq:corner-sheetwise-margins}
\end{equation}
Write
\(\zeta_M(\epsilon,\rho)
=(\vartheta_M(\epsilon,\rho),\pi_M(\epsilon,\rho);\epsilon)\)
and let
\begin{equation}
 \beta_i(\epsilon,\rho)
 =b_i\bigl(\zeta_M(\epsilon,\rho)\bigr),
 \qquad i=S,L,
 \label{eq:corner-sheet-face-pullback}
\end{equation}
be their Maxwell traces. We work in the fixed labeled sheetwise
face-incidence class
\begin{equation}
 \beta_S(\epsilon,\rho)=\upsilon_S(\epsilon,\rho)\rho,
 \qquad
 \beta_L(\epsilon,\rho)=\upsilon_L(\epsilon,\rho),
 \qquad \upsilon_S,\upsilon_L>0,
 \label{eq:corner-fixed-face-incidence}
\end{equation}
where \(\upsilon_S\) and \(\upsilon_L\) are positive \(C^2\) units on a
two-sided resolved collar of \(\rho=0\). Thus the labeled \(S\) minimum is the
unique boundary-linked member of the selected pair and meets the face
transversely, while the labeled \(L\) minimum remains uniformly in the
physical interior. We denote this incidence class by
\(\mathfrak f_{S|L}\). Here \(\rho=0\) is the lift of the physical face and
\(D_{\vartheta}\ne0\) on every compact subinterval of the resolved
positive half-line. All asymptotic statements are made after pullback to
this chart and are uniform in \(C^j\) on every
\(K\Subset(0,\infty)\) in the normalized coordinate introduced below.
The label in \(\mathfrak f_{S|L}\) and the sheetwise transverse contact
class are fixed throughout. Interchanging the boundary-linked sheet or
allowing a higher-order contact changes this incidence class.  Only the
restrictions \(b_S,b_L\) are retained.
\end{definition}

Geometrically, \(S\) is the phase that reaches the physical boundary, while
\(L\) remains in the interior.  The symbol \(\mathfrak f_{S|L}\) records which
labeled sheet meets the face and the transverse nature of that contact.
Selected-sheet data retain the face margin after evaluation on the two
equilibrium sheets, and fixed polarization keeps temperature and pressure in
their thermodynamic roles.  The classification compares the resulting
coexistence profiles under positive changes of units and the removal of terms
of higher corner weight.  The ordinary fixed-charge RN--AdS critical endpoint
lies outside this class because the phases coalesce and the two-well
hypothesis fails \cite{KubiznakMann2012PV}.
For the C-metric blow-up of
Sec.~\ref{subsec:cmetric-maxwell-blowup}, the normalized-time face meets the
boundary-linked \(S\) sheet transversely, while the \(L\) sheet remains in
the interior, and hence realizes \(\mathfrak f_{S|L}\).  Off-sheet extensions
such as \(\beta=\rho\) and \(\beta=\rho+y_S\) have the same restriction at
\(y_S=0\) and are identified in the selected-sheet problem.

The value-gap reduction retains all local two-phase selection information.

\begin{proposition}[Reduction to the value-gap germ]
\label{prop:corner-gap-reduction}
In disjoint state neighborhoods of two nondegenerate minima, a
parameter-preserving state diffeomorphism puts the potential in the form
\begin{equation}
 \mathcal F_i(y_i;\vartheta,\pi,\epsilon)
 =g_i(\vartheta,\pi;\epsilon)+\frac12 y_i^2,
 \qquad i=S,L.
 \label{eq:corner-two-well-Morse-form}
\end{equation}
Consequently the local Maxwell set and the ordering of the two minima are
determined by the sign of \(D\). A common increasing reparametrization of
the potential changes the gap, along \(D=0\), only by multiplication by a
positive unit. A boundary-compatible state-control equivalence changes each
selected-sheet admissibility margin only by multiplication by a positive
unit.
\end{proposition}

The sheetwise covariance is defined before passing to leading order.  A
label-preserving control map \(\Psi\), a common increasing transformation
\(\Theta\) of the two selected critical values, and positive sheet units
\(U_i\) obey
\begin{align}
 \widetilde g_i\bigl(\Psi(\zeta)\bigr)
 &=\Theta\bigl(g_i(\zeta),\zeta\bigr),
 \qquad \partial_1\Theta>0,
 \label{eq:corner-equivalence-potential}\\
 \widetilde b_i\bigl(\Psi(\zeta)\bigr)
 &=U_i(\zeta)b_i(\zeta),
 \qquad U_i>0,\qquad i=S,L.
 \label{eq:corner-equivalence-face}
\end{align}
The same \(\Theta\) acts on both wells and therefore preserves their value
order.
Equation~\eqref{eq:corner-equivalence-face} preserves which labeled sheet
is linked to the face and the transverse contact in
Eq.~\eqref{eq:corner-fixed-face-incidence}.  Every full boundary-compatible
state-control equivalence that maps the selected sheets label by label
induces Eqs.~\eqref{eq:corner-equivalence-potential} and
\eqref{eq:corner-equivalence-face}.  The quotient below retains the induced
Maxwell traces needed to classify a pressure turn; the covariance of the
complete wall complex is developed in Section~\ref{sec:decorated-complex}.

\subsection{Newton data and primitive classification}
\label{subsec:maxwell-corner-universality}

Let
\begin{equation}
 \begin{gathered}
 R(r)=\sum_{j=0}^{N}A_jr^{\kappa_j},
 \qquad 0=\kappa_0<\kappa_1<\cdots<\kappa_N,\\
 \kappa_j\in\mathbb N_0,
 \qquad |A_0|=1.
 \end{gathered}
 \label{eq:corner-Newton-polynomial}
\end{equation}
be a nonzero real polynomial with the zero coefficients omitted. We call
it the resolved Newton polynomial of the volume jump. Its positive
projective-dilation class is
\begin{equation}
 [R]=\{A R(B r): A>0,\ B>0\}.
 \label{eq:corner-projective-class}
\end{equation}
The sign of \(A_0\) is retained. Positive dilation changes the choice of
boundary coordinate, while positive multiplication changes the jump unit.

Let \(\rho=\rho_*r\), with \(\rho_*>0\), and suppose that, for constants
\(a,b,c,c_p,c_s,c_v,\delta>0\), integers \(m\geq1\), \(n\geq0\), and
the polynomial \(R\) in Eq.~\eqref{eq:corner-Newton-polynomial},
\begin{align}
 \pi_M-\pi_t
 &=c_p\epsilon^a\rho^m(1+R_p),
 \label{eq:corner-universality-pressure}\\
 \Delta s
 &=c_s\epsilon^{-b}(1+R_s),
 \label{eq:corner-universality-entropy}\\
 \Delta v
 &=-c_v\epsilon^{-c}\rho^{-n}
   \bigl[R(r)+R_v\bigr].
 \label{eq:corner-universality-volume}
\end{align}
The remainders obey, for every \(K\Subset(0,\infty)\),
\begin{equation}
 \|R_p\|_{C^2(K)}+
 \|R_s\|_{C^1(K)}+
 \|R_v\|_{C^1(K)}
 \leq C_K\epsilon^\delta.
 \label{eq:corner-universality-remainders}
\end{equation}
Here and below derivatives are taken with respect to \(r\). Set
\begin{equation}
 \iota=m-n,
 \qquad
 \sigma=a+b-c.
 \label{eq:corner-universality-weights}
\end{equation}
The finite list
\begin{equation}
\mathscr N=(a,b,c;m,n,[R])
\label{eq:corner-Newton-data}
\end{equation}
is the resolved Newton tuple. Throughout this subsection
\(\mathfrak f_{S|L}\) is held fixed, and \(\mathscr N\) records the
pressure and jump data that remain after this sheetwise incidence has been
specified.
Here \(a\) and \(m\) measure the pressure-layer scale and contact order,
\(b\) measures the entropy-gap scale, \(c\) and \(n\) measure the volume-gap
scale and pole order, and \([R]\) records the leading volume inversions.

The classification is complete in the following quotient. For
\(j=0,1,2\), \(\delta>0\), and \(K\Subset(0,\infty)\), let
\begin{equation}
 \mathcal R^j_\delta(K)
 =\left\{E_\epsilon:
   \|E_\epsilon\|_{C^j(K)}=O_K(\epsilon^\delta)\right\}.
 \label{eq:corner-remainder-quotient}
\end{equation}
Put \(\xi=\vartheta-\vartheta_M(\epsilon,\rho_*r)\). A
\emph{profile equivalence}, formally a polarization-preserving leading-profile
equivalence, between two
profiles with the same \(\mathfrak f_{S|L}\) consists of a
label-preserving resolved collar map \(\widehat\Psi\), positive constants
\begin{equation}
 A_\epsilon,A_r,A_\pi,A_\vartheta,A_D,A_s,A_v>0,
 \label{eq:corner-leading-positive-constants}
\end{equation}
and some \(\delta_0>0\), with the following properties. Both
\(\widehat\Psi\) and its inverse extend as \(C^2\) maps to the boundary
faces of the collar, their Jacobians are nonzero there, and their oriented
normal derivatives are positive. On the Maxwell trace,
\begin{align}
 \widetilde\epsilon
 &=A_\epsilon\epsilon(1+E_\epsilon),
 &E_\epsilon&\in\mathcal R^0_{\delta_0}(K),\nonumber\\
 \widetilde r
 &=A_r r(1+E_r),
 &E_r&\in\mathcal R^2_{\delta_0}(K).
 \label{eq:corner-equivalence-leading}
\end{align}
Write \(\widehat\Psi^*\) for pullback to the original Maxwell trace. For
an integrable profile let \(\vartheta_{\mathrm{ref}}=\vartheta_t\); for a
logarithmic or power-divergent profile use the value at \(r=1\). The four
thermodynamic traces are identified modulo the prescribed quotient when
\begin{align}
 \epsilon^{-a}
 \bigl[\widehat\Psi^*(\widetilde\pi_M-\widetilde\pi_t)
       -A_\pi(\pi_M-\pi_t)\bigr]
 &\in\mathcal R^2_{\delta_0}(K),
 \label{eq:corner-leading-pressure-quotient}\\
 \epsilon^{-\sigma}
 \bigl[\widehat\Psi^*(\widetilde\vartheta_M
       -\widetilde\vartheta_{\mathrm{ref}})
       -A_\vartheta(\vartheta_M-\vartheta_{\mathrm{ref}})\bigr]
 &\in\mathcal R^2_{\delta_0}(K),
 \label{eq:corner-equivalence-temperature}\\
 \epsilon^b
 \bigl[\widehat\Psi^*\widetilde{\Delta s}-A_s\Delta s\bigr]
 &\in\mathcal R^1_{\delta_0}(K),
 \nonumber\\
 \epsilon^c\rho^n
 \bigl[\widehat\Psi^*\widetilde{\Delta v}-A_v\Delta v\bigr]
 &\in\mathcal R^1_{\delta_0}(K).
 \label{eq:corner-equivalence-jump-conormal}
\end{align}
The first gap conormal is retained as well. Using \(D\) as the coordinate
normal to coexistence, it obeys
\begin{equation}
 \left.
 \frac{\partial(\widetilde D\circ\widehat\Psi)}{\partial D}
 \right|_{D=0}-A_D
 \in\mathcal R^1_{\delta_0}(K).
 \label{eq:corner-equivalence-gap-conormal}
\end{equation}
The sheetwise margins satisfy Eq.~\eqref{eq:corner-equivalence-face} with
positive units.  Positive nonzero boundary Jacobians and the stability of
\(\mathcal R^j_{\delta_0}(K)\), after decreasing \(\delta_0\) if necessary,
make these conditions an equivalence relation under composition and
inversion.  Its block-diagonal control scaling preserves the reservoir
polarization and hence the physical pressure direction.

Define
\begin{equation}
 J_q(r;r_0)=
 \begin{cases}
  (r^q-r_0^q)/q,&q\ne0,\\
  \log(r/r_0),&q=0,
 \end{cases}
 \label{eq:corner-primitive-monomial}
\end{equation}
and the Newton primitive
\begin{equation}
 \Phi_{\iota,R}(r;r_0)
 =\sum_{j=0}^{N}A_jJ_{\iota+\kappa_j}(r;r_0).
 \label{eq:corner-Newton-primitive}
\end{equation}
When \(\iota>0\), all exponents in this sum are positive and we use the
continuous extension \(r_0=0\). When \(\iota\leq0\), we use \(r_0=1\).

\begin{theorem}[Leading-profile classification at a Maxwell boundary]
\label{thm:maxwell-corner-Newton-classification}
Fix the sheet-incidence class
\(\mathfrak f_{S|L}\) of Definition~\ref{def:resolved-two-phase-germ}.
Suppose a resolved Maxwell--boundary profile in this category
satisfies
Eqs.~\eqref{eq:corner-universality-pressure}--
\eqref{eq:corner-universality-remainders} and the Clapeyron relation
\eqref{eq:corner-universality-clapeyron}. If \(\iota>0\), suppose in
addition that the boundary value fixes the integration constant through
\begin{equation}
 \lim_{\rho\downarrow0}\limsup_{\epsilon\downarrow0}
 \epsilon^{-\sigma}
 \left|\vartheta_M(\epsilon,\rho)
       -\vartheta_t(\epsilon)\right|=0.
 \label{eq:corner-universality-matching}
\end{equation}
Let
\begin{equation}
 \widehat\pi_\epsilon
 =\frac{\pi_M-\pi_t}
 {c_p\rho_*^m\epsilon^a},
 \qquad
 \mathcal A_\epsilon
 =\frac{\lambda m c_pc_v}{c_s}
  \rho_*^\iota\epsilon^\sigma.
 \label{eq:corner-general-normalization}
\end{equation}
For \(\iota>0\), set
\begin{equation}
 \widehat\vartheta_\epsilon(r)
 =\frac{\vartheta_t(\epsilon)
       -\vartheta_M(\epsilon,\rho_*r)}
       {\mathcal A_\epsilon}.
 \label{eq:corner-integrable-temperature-normalization}
\end{equation}
For \(\iota\leq0\), set instead
\begin{equation}
 \widehat\vartheta_\epsilon(r)
 =\frac{\vartheta_M(\epsilon,\rho_*)
       -\vartheta_M(\epsilon,\rho_*r)}
       {\mathcal A_\epsilon}.
 \label{eq:corner-singular-temperature-normalization}
\end{equation}
Then, in \(C^2(K)\) for every \(K\Subset(0,\infty)\),
\begin{equation}
 \widehat\pi_\epsilon(r)\longrightarrow r^m,
 \qquad
 \widehat\vartheta_\epsilon(r)
 \longrightarrow\Phi_{\iota,R}(r;r_0),
 \label{eq:corner-general-profile}
\end{equation}
where \(r_0=0\) for \(\iota>0\) and \(r_0=1\) for
\(\iota\leq0\). Equivalently, the limiting coexistence curve is
\begin{equation}
 \widehat\vartheta
 =\Phi_{\iota,R}
   \bigl(\widehat\pi^{1/m};r_0\bigr),
 \qquad \widehat\pi>0.
 \label{eq:corner-general-eliminated}
\end{equation}

The three boundary types are distinct. For \(\iota>0\), the Clapeyron
one-form is integrable and the normalized curve has a finite boundary
value. For \(\iota=0\), its leading term is \(A_0\log r\). For
\(\iota<0\), its leading term is \(A_0r^\iota/\iota\); a further logarithm
occurs exactly when \(\iota+\kappa_j=0\) for a monomial on the Newton face.
Thus the logarithmic and power-divergent profiles cannot satisfy the
finite-boundary matching condition without a cancellation that changes the
Newton tuple.

The tuple \(\mathscr N\) is invariant under profile equivalence in the fixed
category
\(\mathfrak f_{S|L}\). Conversely, two resolved leading profiles in this
same category with the same \(a,b,c,m,n\) and the same class \([R]\) are
equivalent in the quotient
Eqs.~\eqref{eq:corner-remainder-quotient}--
\eqref{eq:corner-equivalence-gap-conormal}.  Completeness holds for the fixed
incidence class \(\mathfrak f_{S|L}\), the selected Maxwell traces, and their
first gap conormal.

Let \(I=[r_-,r_+]\Subset(0,\infty)\), assume that neither endpoint is a
zero of \(R\), and suppose that \(R\) has exactly \(q\) simple zeros
\(r_1<\cdots<r_q\) in \(I\). For all sufficiently small \(\epsilon\),
the finite-\(\epsilon\) coexistence curve has exactly \(q\) regular turns
in \(I\), one at each
\begin{equation}
 r_{j,\epsilon}=r_j+O(\epsilon^\delta).
 \label{eq:corner-general-turn-location}
\end{equation}
Their normalized curvatures satisfy
\begin{equation}
 \left.
 \frac{\dd^2\widehat\vartheta_\epsilon}
      {\dd\widehat\pi_\epsilon^2}
 \right|_{r=r_{j,\epsilon}}
 =\frac{r_j^{\iota-2m+1}}{m^2}R'(r_j)
  +O(\epsilon^\delta).
 \label{eq:corner-general-turn-curvature}
\end{equation}
The sign sequence of \(R\) therefore fixes the complete slope and turn
sequence on \(I\). In particular, the finite Newton polynomial determines
the turn portrait on resolved compacta. Multiple positive roots are not
structurally stable under unrestricted perturbations of the leading
Newton coefficients. They form a positive-codimension semialgebraic
subset of the resultant discriminant; its generic smooth stratum in an
ordinary coefficient chart has codimension one. Such a root may split
into simple roots or leave the positive real axis under perturbation.
\end{theorem}

The proof begins by substituting
Eqs.~\eqref{eq:corner-universality-pressure}--
\eqref{eq:corner-universality-volume} into the Clapeyron relation, which gives
\(\dd\widehat\vartheta_\epsilon/\dd r=
r^{\iota-1}[R(r)+O_K(\epsilon^\delta)]\) in \(C^1(K)\); integration yields
the Newton primitive in Eq.~\eqref{eq:corner-general-profile}.  Positive
dilations and units preserve \(\mathscr N\), while persistence of simple roots
gives the turn locations and a second differentiation gives
Eq.~\eqref{eq:corner-general-turn-curvature}.  Changing the incidence or
contact class, or retaining off-sheet and higher-jet data, changes the
quotient.  The details within the fixed category are given in
\ref{app:wall-normal-forms}.

Representative profiles of the three boundary types and the integrable
binomial family are shown in Fig.~\ref{fig:maxwell-boundary-profiles}.

\begin{figure*}[t]
 \centering
 \includegraphics[width=\textwidth]{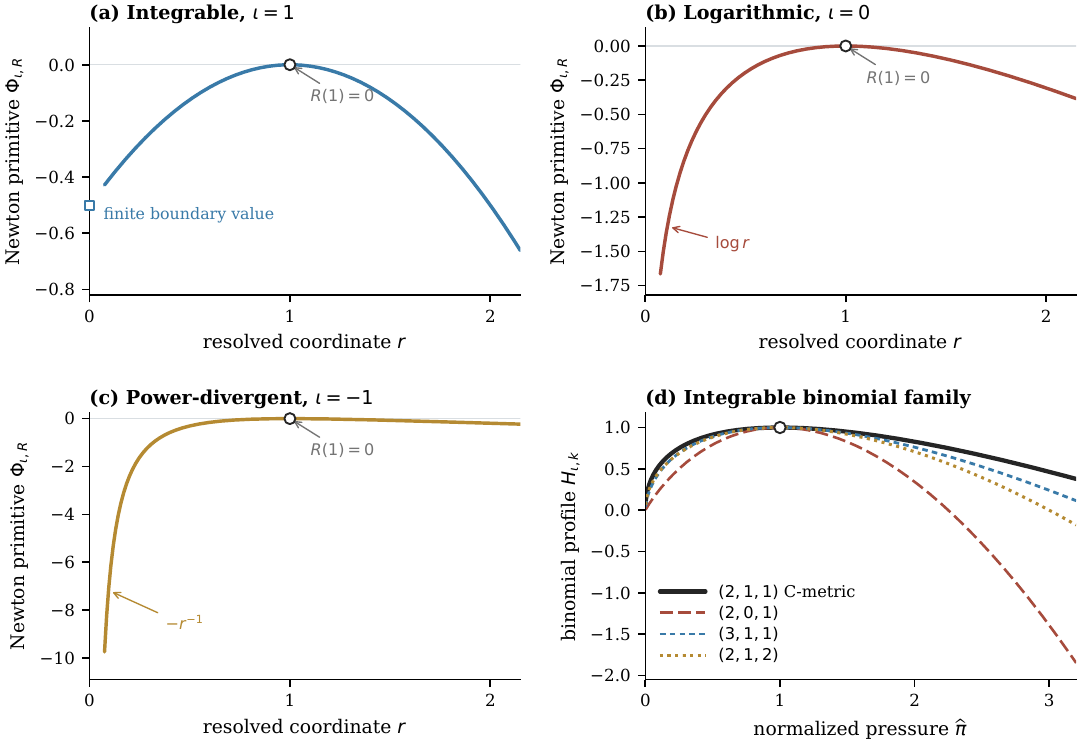}
 \caption{Representative resolved Maxwell--boundary profiles.  Panels
 (a)--(c) use \(R(r)=1-r\), with additive constants chosen so that
 \(\Phi_{\iota,R}(1;1)=0\).  The sign of \(\iota=m-n\) distinguishes a finite
 boundary value, logarithmic divergence, and power divergence; the open
 marker at \(r=1\) is the volume inversion \(R(1)=0\).  Panel (d) shows
 integrable binomial profiles as functions of \(\widehat\pi=r^m\).  Every
 curve turns at \((1,1)\).  The bold curve is the C-metric class
 \((m,n,k)=(2,1,1)\), for which
 \(\widehat\vartheta=2\sqrt{\widehat\pi}-\widehat\pi\).  The remaining curves
 are formal examples of other normal-form data.}
 \label{fig:maxwell-boundary-profiles}
\end{figure*}

Theorem~\ref{thm:maxwell-corner-Newton-classification} classifies profiles
after terms with a positive weight gap have been discarded by
Eq.~\eqref{eq:corner-universality-remainders}, while every monomial on the
leading Newton face is retained in \(R\).  This distinction matters because
the three-term polynomial
\(1-Ar^k+Br^{2k}\) may have zero, one double, or two simple positive roots;
the exponent gap \(k\) alone does not determine the turning portrait.

The binomial case realized by the accelerating black hole is more
constrained.

\begin{corollary}[Binomial Newton-face Maxwell corner]
\label{thm:maxwell-corner-universality}
Assume the hypotheses of
Theorem~\ref{thm:maxwell-corner-Newton-classification} with
\begin{equation}
 \iota=m-n>0,
 \qquad
 \sigma=a+b-c>0,
 \qquad
 R(r)=1-r^k,\quad k\geq1.
 \label{eq:corner-binomial-data}
\end{equation}
Equivalently, before normalization the leading volume factor is
\(A-B\rho^k\), with \(A,B>0\), and
\(\rho_*=(A/B)^{1/k}\). Define
\begin{align}
 r&=\frac{\rho}{\rho_*},
 &
 \widehat\pi_\epsilon
 &=\frac{\pi_M-\pi_t}
 {c_p\rho_*^m\epsilon^a},
 \nonumber\\
 \widehat\vartheta_\epsilon
 &=\frac{\iota(\iota+k)c_s}
 {k\lambda m c_pc_v\rho_*^\iota\epsilon^\sigma}
 (\vartheta_t-\vartheta_M).
 \label{eq:corner-universality-normalization}
\end{align}
Then, in \(C^2\) on every compact subinterval of \(r>0\),
\begin{align}
 \widehat\pi_\epsilon(r)&\longrightarrow r^m,
 \nonumber\\
 \widehat\vartheta_\epsilon(r)&\longrightarrow
 H_{\iota,k}(r)
 :=\frac{\iota+k}{k}r^\iota
   -\frac{\iota}{k}r^{\iota+k}.
 \label{eq:corner-universality-profile}
\end{align}
Eliminating \(r\) gives
\begin{equation}
 \boxed{
 \widehat\vartheta
 =\frac{\iota+k}{k}\widehat\pi^{\iota/m}
 -\frac{\iota}{k}\widehat\pi^{(\iota+k)/m}}
 \qquad (\widehat\pi>0).
 \label{eq:corner-universality-eliminated}
\end{equation}
The limiting curve has one turn at
\((\widehat\pi,\widehat\vartheta)=(1,1)\). On every fixed annulus
\(0<r_-<1<r_+<\infty\), the finite-\(\epsilon\) curve has, for all
sufficiently small \(\epsilon\), exactly one regular turn and
\begin{equation}
 r_\epsilon=1+O(\epsilon^\delta),
 \qquad
 \left.
 \frac{\dd^2\widehat\vartheta_\epsilon}
      {\dd\widehat\pi_\epsilon^2}
 \right|_{r=r_\epsilon}
 \longrightarrow-\frac{\iota(\iota+k)}{m^2}.
 \label{eq:corner-universality-turn}
\end{equation}
Thus the physical coexistence temperature has a nondegenerate local
minimum. The turn is exactly the inversion \(\Delta v=0\), while
\(\Delta s\) remains nonzero.
\end{corollary}

Every allowed choice of Newton data has a local representative in any prescribed
sheetwise incidence class of the form
Eq.~\eqref{eq:corner-fixed-face-incidence}.

\begin{proposition}[Local two-well realizability of the resolved Newton data]
\label{prop:corner-Newton-realizability}
For every choice of \(a,b,c>0\), \(m\geq1\), \(n\geq0\), polynomial
\(R\) satisfying Eq.~\eqref{eq:corner-Newton-polynomial}, and positive
\(C^2\) sheet units \(\upsilon_S,\upsilon_L\), there exists a local
punctured resolved Maxwell--boundary profile with margins
\(\beta_S=\upsilon_S\rho\), \(\beta_L=\upsilon_L\), satisfying
Eqs.~\eqref{eq:corner-universality-pressure}--
\eqref{eq:corner-universality-volume} with zero remainders. It may be
realized by two strict quadratic wells. If \(\iota>0\), the realization
extends continuously to a finite coexistence temperature at the lifted
boundary. For \(\iota\leq0\), its normalized coexistence temperature has
the logarithmic or power divergence stated in
Theorem~\ref{thm:maxwell-corner-Newton-classification}.  This local
thermodynamic realization supplies the prescribed two-well critical values
and selected-sheet margins.  A global Gibbs completion, an off-shell face
extension, or realization by gravitational field equations requires
additional data.
\end{proposition}

The difference \(\iota=m-n\) decides whether the Clapeyron one-form reaches
the boundary.  The complete Newton polynomial records the leading volume
inversions, while the weights \(a,b,c\) retain the physical sizes of the
pressure, entropy, and volume layers even when two systems share the same
normalized curve.  Neither feature is contained in the stationary degree or in
the Morse indices of the two phases.

\section{Constrained thermodynamic families}
\label{sec:constrained-families}

Consider an effective one-state family in a fixed thermodynamic ensemble.
Let \(X\subset\mathbb R\) be an open interval,
let \(\Lambda\) be a smooth control manifold, and let
\begin{equation}
 \mathcal F:X\times\Lambda\longrightarrow\mathbb R,
 \qquad
 \rho_a:X\times\Lambda\longrightarrow\mathbb R,
 \quad a=1,\ldots,q,
 \label{eq:family-and-faces}
\end{equation}
be of class \(C^r\), \(r\geq4\). At fixed control the open physical set is
\begin{equation}
 \mathcal A_\lambda
 =\{x\in X:\rho_a(x,\lambda)>0\text{ for all }a\}.
 \label{eq:admissible-set}
\end{equation}
The functions in \eqref{eq:family-and-faces} are defined on an ambient
collar of the closure of the physical set. Consequently an ambient
stationary sheet continues across a physical face.  This distinguishes loss
of physical membership from destruction of a stationary point.

The data form the triple \((L,\mathfrak p,\mathcal A)\). Here
\(L\) is the equilibrium Legendrian, \(\mathfrak p\) is the reservoir
polarization that fixes the controls and the thermodynamic potential, and
\(\mathcal A\) is the admissible domain. We keep this triple fixed. Changing
the ensemble or promoting a boundary configuration to a thermodynamic phase
defines a different selection problem.
The fixed-reservoir role of the polarization and its Brouwer--Morse degree
were analyzed in Ref.~\cite{Li2026LegendreCovariant}.  Here the polarization
is held fixed, while \(\mathcal A\), critical-value order, and global
selection enter as independent data.

\subsection{Stationary sheets and physical phases}
\label{subsec:stationary-sheets}

Fix a relatively compact control window \(U\Subset\Lambda\). The ambient
critical locus is
\begin{equation}
 \mathscr C
 =\{(x,\lambda)\in X\times U:\mathcal F_x(x,\lambda)=0\}.
 \label{eq:critical-locus}
\end{equation}
At a nondegenerate stationary point define
\begin{equation}
 \operatorname{ind}(x,\lambda)=
 \begin{cases}
  0,&\mathcal F_{xx}(x,\lambda)>0,\\
  1,&\mathcal F_{xx}(x,\lambda)<0.
 \end{cases}
 \label{eq:one-dimensional-index}
\end{equation}
The point is active when it lies in \(\mathcal A_\lambda\), locally stable
when \(\operatorname{ind}=0\), metastable when it is an active local minimum
but not the lowest active minimum, and globally selected when it has the
lowest value among all active local minima. Equal lowest values give
coexistence.  The selection problem considered below contains no additional
nonstationary reference background.  If prescribed reference-value sheets
participate, their equality walls with every stationary value, and their
mutual equality walls when more than one is present, must first be adjoined to
the critical-value complex and their value order added to the chamber datum.

Physical faces act only as admissibility filters; points on
\(\partial\mathcal A_\lambda\) are not counted as independent boundary
equilibria. A theory that includes genuine boundary phases
has boundary critical points and additional codimension-one events
\cite{Arnold1978Boundary,BorodzikNemethiRanicki2016Boundary,
BorodzikBuczynska2025Families} and defines a different class from the one
treated here.

\subsection{Properness and genericity}
\label{subsec:properness-genericity}

The general analysis uses the following hypotheses.
\begin{enumerate}
\item The restriction of the critical projection
\(\operatorname{pr}:\mathscr C\to U\) is proper and has finite fibers. In particular, no
stationary point escapes through an end of \(X\) while the controls remain
in a compact subset of \(U\).
\item Away from strata of codimension at least two, the stationary
projection has only ordinary \(A_2\) folds, equality of two distinct
critical values is transverse, and a nondegenerate stationary sheet meets
one regular physical face transversely.
\item At an ordinary fold \(p=(x_0,\lambda_0)\),
\begin{equation}
 \mathcal F_x(p)=\mathcal F_{xx}(p)=0,\qquad
 \mathcal F_{xxx}(p)\ne0,\qquad
 d_\lambda\mathcal F_x(p)\ne0
 \label{eq:ordinary-fold}
\end{equation}
in a direction normal to the projected fold.
\item The fold, critical-value, and admissibility strata have a locally
finite Whitney stratification on \(U\), with all incidences used by a path
homotopy included in the stratification. Real-analytic or
o-minimal-definable data on a compact window are standard sufficient
settings. Finite differentiability does not imply this property for a general
\(C^r\) family, so we assume it
\cite{Cerf1970Stratification,GoreskyMacPherson1988Stratified}.
\item Codimension-two points are limited, on the control window under
study, to the normal forms and transverse products explicitly entered in
the incidence data. In particular, no unrecorded nonisolated critical
set or branch escape is allowed.
\end{enumerate}

These assumptions define the sense in which the wall complex is complete.
The reconstructed portrait consists of stationary-sheet existence, state
order, Morse index, critical-value order, physical-face signs, metastability,
and the global winner.  Quantitative free-energy gaps, barrier heights,
transition rates, and finite-\(G_N\) dynamics require further model-specific
data.
Figure~\ref{fig:information-layers} separates the stationary, value-order,
and admissibility data retained below.

\begin{figure}[!b]
 \centering
 \includegraphics[width=\linewidth]{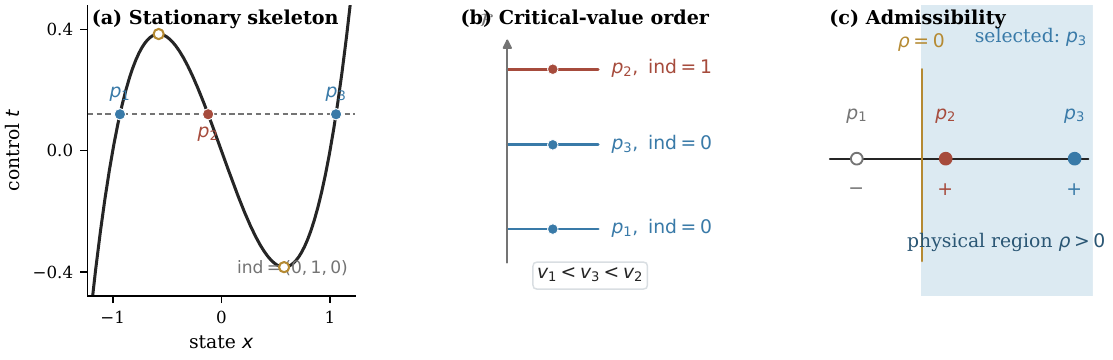}
 \caption{The three information layers retained by the decorated chamber
 datum. (a) The cubic stationary skeleton at fixed control has the
 state-ordered Morse word \((0,1,0)\); gold circles mark its folds.
 (b) Critical values carry an independent order, here
 \(v_1<v_3<v_2\), which is not the state order.
 (c) A physical face excludes \(p_1\) without destroying its ambient
 stationary sheet. After the face filter is applied, \(p_3\) is the
 selected active minimum. Neither the value order nor this exclusion can
 be inferred from the stationary skeleton alone.}
 \label{fig:information-layers}
\end{figure}

\section{Phase selection with a moving physical domain}
\label{sec:local-versus-global}

\subsection{The weighted-area obstruction}
\label{subsec:weighted-obstruction}

Consider a canonical one-state family
\begin{equation}
 \mathcal F_T(x;\eta)=M(x;\eta)-T S(x;\eta),
 \qquad S_x(x;\eta)>0,
 \label{eq:canonical-family}
\end{equation}
where \(\eta\) denotes the controls other than the reservoir temperature.
The on-shell temperature is
\begin{equation}
 T_H(x;\eta)=\frac{M_x(x;\eta)}{S_x(x;\eta)}.
 \label{eq:on-shell-temperature}
\end{equation}

\begin{proposition}[Weighted-area identity]
\label{prop:weighted-area}
For every \(x_i<x_j\),
\begin{equation}
 \mathcal F_T(x_j;\eta)-\mathcal F_T(x_i;\eta)
 =\int_{x_i}^{x_j}S_x(x;\eta)
       \bigl[T_H(x;\eta)-T\bigr]\,dx.
 \label{eq:weighted-area-identity}
\end{equation}
If the endpoints are stationary, their equality of free energies is a
weighted-area condition. The stationary cover, folds, and Morse indices
depend on the zero set of \(T_H-T\) and its local derivatives, whereas the
value difference also depends on the positive weight \(S_x\).
\end{proposition}

\begin{proof}
The derivative of the canonical potential factorizes as
\begin{equation}
 \partial_x\mathcal F_T=M_x-TS_x=S_x(T_H-T).
 \label{eq:gradient-factorization}
\end{equation}
Integration gives \eqref{eq:weighted-area-identity}. At a stationary point,
\begin{equation}
 \mathcal F_{T,xx}=S_x\partial_xT_H,
 \label{eq:hessian-temperature}
\end{equation}
and at a fold,
\begin{equation}
 \mathcal F_{T,xxx}=S_x\partial_x^2T_H.
 \label{eq:third-temperature}
\end{equation}
The positivity of \(S_x\) preserves the local signs but does not remove the
weight from \eqref{eq:weighted-area-identity}.
\end{proof}

Maxwell equal-area constructions in black-hole thermodynamics and the
extended-space Clapeyron relation were developed in
Refs.~\cite{SpallucciSmailagic2013Maxwell,WeiLiu2015Clapeyron}.  The weight
in Eq.~\eqref{eq:weighted-area-identity} is the information that is lost
when only the graph or zero set of \(T_H-T\) is retained.

Ref.~\cite{ZhangEtAl2026Geometric} proposed an if-and-only-if
two-extrema diagnostic for the canonical small/large black-hole transition.
The statement below isolates the endpoint value order, admissibility, and
global-minimality conditions under which that diagnostic applies.

\begin{proposition}[Corrected two-extrema criterion]
\label{prop:corrected-two-extrema}
Fix the controls \(\eta\), and suppose two ordinary folds of \(T_H(x;\eta)\)
bound a connected temperature interval \(I=(T_-,T_+)\) with exactly three
interior stationary sheets
\begin{equation}
 x_S(T)<x_U(T)<x_L(T).
 \label{eq:three-sheet-order}
\end{equation}
Assume \(S_x>0\), with \(x_S\) and \(x_L\) strict minima and \(x_U\) a
strict maximum.  Define
\begin{equation}
 D(T)=\mathcal F_T(x_L(T);\eta)
      -\mathcal F_T(x_S(T);\eta).
 \label{eq:two-extrema-value-gap}
\end{equation}
Then
\begin{equation}
 D'(T)=-\bigl[S(x_L(T);\eta)-S(x_S(T);\eta)\bigr]<0.
 \label{eq:two-extrema-gap-monotonicity}
\end{equation}
Consequently the two stable sheets have at most one Maxwell equality in
\(I\).  They have one transverse equality in \(I\) if and only if
\begin{equation}
 \lim_{T\downarrow T_-}D(T)>0,
 \qquad
 \lim_{T\uparrow T_+}D(T)<0.
 \label{eq:two-extrema-endpoint-signs}
\end{equation}
That equality is a physical first-order transition only if both sheets are
admissible there and their common value is no larger than that of every other
active local minimum, including every competing phase admitted by the chosen
ensemble.  Thus two
nondegenerate extrema guarantee the local three-sheet cover; they do not by
themselves guarantee Maxwell equality or global phase coexistence.
\end{proposition}

Along either stationary sheet, the envelope theorem gives
\(\dd\mathcal F_T(x_i(T);\eta)/\dd T=-S(x_i(T);\eta)\).  Subtracting the two
identities yields Eq.~\eqref{eq:two-extrema-gap-monotonicity}; since
\(S_x>0\) and \(x_L>x_S\), the derivative is strictly negative.  Continuity
then gives a unique zero exactly when the endpoint signs reverse, and the
nonzero derivative makes that zero transverse.  Admissibility and comparison
with every other active minimum are separate global tests.  The fold-endpoint
limits are detailed in \ref{appsubsec:two-extrema-proof}.

In families where the endpoint signs, admissibility, and global comparison
follow from separate structure, the
diagnostic is valid.  For a general constrained thermodynamic family, the
two extrema settle only the local branch multiplicity.  The constructions
below show independently how the endpoint value order can move and how an
admissibility face can delete one of the competitors.

If two distinct strict interior minima coexist in a connected one-state
family and an interior maximum lies between them, then \(T_H-T\) has at least
three zeros.  Rolle's theorem therefore gives at least two extrema of \(T_H\).
Disconnected physical components and boundary states fall outside this
necessity statement.

\subsection{Iso-stationary families with different Maxwell sets}
\label{subsec:iso-stationary-counterexample}

\begin{theorem}[Local stationary data are insufficient]
\label{thm:local-data-insufficient}
There are real-analytic canonical families with identical equilibrium
cover, fold set, state ordering, Morse indices, Morse word, and Brouwer
degree, but with different Maxwell sets and opposite global winners.
Consequently, no rule based only on those stationary data determines the
Maxwell set or the critical-value order.
\end{theorem}

The analytic construction and its global-value check are given in
\ref{appsubsec:iso-stationary-proof}. In that construction the
stationary equation is the \(\varepsilon\)-independent cubic
\(x^3-x=t\), whereas the unique Maxwell wall moves as
\begin{equation}
 t_M(\varepsilon)=-\frac{2}{15}\varepsilon+O(\varepsilon^2).
 \label{eq:Maxwell-shift-main}
\end{equation}
The families at \(\varepsilon=\delta\) and \(-\delta\) therefore select
opposite outer minima at \(t=0\) while retaining the same stationary data.

In some untruncated cusp families, two extrema of \(T_H\) do accompany a
crossing.  The extrema, branch cover, indices, and degree determine the
stationary skeleton, whereas the critical values locate the crossing and
fix the value order.  An admissibility filter can then remove a competing
sheet and eliminate the physical transition from the same ambient cusp.
Figure~\ref{fig:obstruction-Maxwell} displays the stationary skeleton and
the independently moving Maxwell wall in the analytic counterexample.

\begin{figure}[t]
 \centering
 \includegraphics[width=\linewidth]{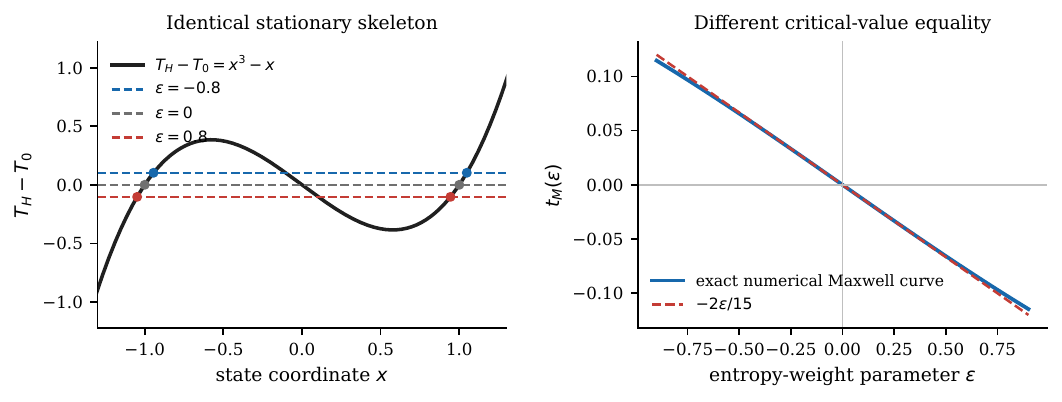}
 \caption{The iso-stationary counterexample.  The left panel shows the
 common cubic stationary skeleton and the distinct Maxwell temperatures for
 three entropy weights.  The two marked outer roots are the competing
 minima.  The right panel compares the numerically evaluated Maxwell curve
 with its analytic tangent \(t_M=-2\varepsilon/15\).  All curves are generated
 directly from Eqs.~\eqref{eq:epsilon-entropy}--\eqref{eq:D-epsilon-t}.}
 \label{fig:obstruction-Maxwell}
\end{figure}

\subsection{Admissibility can remove an ambient Maxwell crossing}
\label{subsec:clipped-cusp}

The obstruction in Theorem~\ref{thm:local-data-insufficient} moves the
Maxwell wall without changing the stationary skeleton. A second example
shows that an admissibility filter can remove the physical transition
altogether while leaving the ambient catastrophe unchanged.

\begin{proposition}[The clipped cusp]
\label{prop:clipped-cusp}
Consider
\begin{equation}
 \mathcal F(x;t)=\frac{x^4}{4}-\frac{x^2}{2}+tx,
 \qquad X=(-2,2).
 \label{eq:clipped-cusp-potential}
\end{equation}
On the full physical interval
\(\mathcal A^{\rm full}=(-3/2,3/2)\), the two stable sheets undergo a
physical Maxwell crossing at \(t=0\). On the clipped interval
\(\mathcal A^{\rm clip}=(-1/2,3/2)\), the same ambient family has exactly
one active stable sheet for \(|t|<1/8\) and hence no physical Maxwell
crossing in that window.
\end{proposition}

The root-location calculation is given in
\ref{appsubsec:clipped-cusp-proof}. It leaves the ambient cover,
folds, indices, and critical values unchanged; only the face-sign data
change.

The two examples isolate complementary obstructions, since a positive entropy
weight changes the value order on a fixed stationary cover, whereas an
admissibility filter changes which entries of that order participate in
physical selection. Both pieces of information are absent from the
stationary topology.

\subsection{Continuity and zeroth-order transitions}
\label{subsec:no-jump}

Here it is important to distinguish two minimization problems. The first minimizes a
continuous off-shell potential over a closed physical domain. The second
minimizes only over active interior stationary sheets. The first has a
continuity theorem; the second can jump when a minimizing sheet is deleted
at an excluded boundary.

The following is the minimum-value form of Berge's maximum theorem
\cite{Berge1963Topological,AliprantisBorder2006Infinite}, together with the
locally inf-compact extension used below.

\begin{proposition}[Continuity of the closed-domain equilibrium value]
\label{prop:no-jump}
Let \(V\) be a metric control space and let
\(K:V\rightrightarrows\mathbb R\) take nonempty compact values. Suppose
that
\begin{equation}
 d_H\!\left(K(\lambda),K(\lambda_0)\right)\longrightarrow0
 \quad\text{as}\quad \lambda\longrightarrow\lambda_0,
 \label{eq:Hausdorff-continuity}
\end{equation}
and that \(\mathcal F\) is jointly continuous on a neighborhood of the
local graph of \(K\). Then
\begin{equation}
 g_{\rm cl}(\lambda)
 =\min_{x\in K(\lambda)}\mathcal F(x,\lambda)
 \label{eq:equilibrium-value}
\end{equation}
is continuous at \(\lambda_0\).

The compactness hypothesis can be replaced by a locally uniform compact
sublevel condition. More precisely, the same conclusion holds when \(K\)
has nonempty closed values, has a closed graph at \(\lambda_0\), is lower
hemicontinuous at \(\lambda_0\), and there are a neighborhood \(N\), a
compact set \(C\), and a number \(c\) such that
\begin{equation}
 \inf_{x\in K(\lambda)}\mathcal F(x,\lambda)<c,\qquad
 \{x\in K(\lambda):\mathcal F(x,\lambda)\leq c\}\subset C
 \quad(\lambda\in N).
 \label{eq:uniform-inf-compactness}
\end{equation}
\end{proposition}

The proof is given in \ref{appsubsec:closed-domain-proof}.

Where the active stationary-minimum set is nonempty, define
\begin{equation}
 g_{\rm stat}(\lambda)
 =\min\{\mathcal F(x,\lambda):
 x\in\mathcal A_\lambda,\ \mathcal F_x=0,\ \mathcal F_{xx}>0\}.
 \label{eq:stationary-selector-value}
\end{equation}
Continuity can fail for this selector because the candidate correspondence
in \eqref{eq:stationary-selector-value} need not be lower hemicontinuous at
an admissibility wall.

\begin{proposition}[A transverse boundary deletion produces a finite jump]
\label{prop:boundary-deletion-jump}
For
\begin{equation}
 \mathcal F_\kappa(x)=(x^2-1)^2+\kappa(3x-x^3),
 \qquad 0<\kappa<\frac43,
 \label{eq:boundary-deletion-example}
\end{equation}
and the open physical condition \(x>\lambda\), the minimum over active
stationary minima jumps upward by \(4\kappa\) as \(\lambda\) crosses
\(-1\). On every fixed compact interval \([\lambda,R]\), \(R>1\), the
closed-domain value remains continuous.
\end{proposition}

The factorization and one-sided minimization are worked out in
\ref{appsubsec:closed-domain-proof}.

\begin{corollary}[Boundary criteria for zeroth-order transitions]
\label{cor:jump-diagnostic}
A finite discontinuity in a reported equilibrium free energy cannot arise
from the minimum of a jointly continuous potential over a nonempty,
Hausdorff-continuous compact closed physical family. It therefore identifies
at least one of the following mechanisms: the feasible set becomes empty; an
excluded open boundary deletes the winner; the feasible correspondence loses
a closed graph, upper or lower hemicontinuity, or Hausdorff continuity; a
minimizer arrives from a noncompact region; the off-shell potential loses
joint continuity, as can occur under a discontinuous change of action
normalization; or a competing phase has been omitted from the ensemble.
\end{corollary}

An admissibility crossing produces a jump precisely when the departing sheet
is the global winner and its limiting value differs from that of the surviving
winner. If the two limiting values agree, the boundary wall meets the physical
Maxwell wall and the jump closes continuously.

\section{The decorated wall complex and reconstruction}
\label{sec:decorated-complex}

The exact C-metric phase diagram combines three geometrically distinct
events: creation or loss of stationary sheets at folds, exchange of
critical-value order on a Maxwell wall, and loss of physical membership at
an admissibility face.  A decorated chamber datum records all three, so
continuation through a control window preserves the information needed for
global phase selection.  This record determines the qualitative portrait of
branch existence, value order, and physical membership; model-dependent
magnitudes come from the thermodynamic potential.

\subsection{Walls, chambers, and decorations}
\label{subsec:walls-chambers}

Let \(\mathscr C_{\mathrm{reg}}\) be the nondegenerate part of
\(\mathscr C\), and let \(\operatorname{pr}\colon X\times U\to U\) be the projection.
The fold, critical-value, and admissibility walls are
\begin{align}
 \Delta
 &=\operatorname{pr}\{(x,\lambda):\mathcal F_x=\mathcal F_{xx}=0\},
 \label{eq:fold-wall}\\
 \mathcal M_{\mathrm{cv}}
 &=\overline{\left\{\lambda:
 \begin{array}{l}
 (x_i,\lambda),(x_j,\lambda)\in\mathscr C_{\mathrm{reg}},
 \quad x_i\neq x_j,\\[-2pt]
 \mathcal F(x_i,\lambda)=\mathcal F(x_j,\lambda)
 \end{array}\right\}},
 \label{eq:value-wall}\\
 \mathcal B_a
 &=\operatorname{pr}\{(x,\lambda)\in\mathscr C_{\mathrm{reg}}:
          \rho_a(x,\lambda)=0\},
 \qquad
 \mathcal B=\bigcup_{a=1}^q\mathcal B_a.
 \label{eq:admissibility-wall}
\end{align}
The closure in \eqref{eq:value-wall} retains incidences with higher strata.
On local nondegenerate sheets \(x_i(\lambda)\), define
\begin{equation}
 v_i(\lambda)=\mathcal F(x_i(\lambda),\lambda).
 \label{eq:critical-value-sheet}
\end{equation}
The regular equality wall of sheets \(i\) and \(j\) is transverse when
\begin{equation}
 d(v_i-v_j)\neq0.
 \label{eq:value-transversality}
\end{equation}
Because \(\mathcal F_x=0\) on each sheet, the envelope identity gives
\begin{equation}
 dv_i=\mathcal F_\lambda(x_i(\lambda),\lambda)\,d\lambda.
 \label{eq:envelope-critical-value}
\end{equation}
A regular face crossing is transverse when
\begin{equation}
 d\bigl[\rho_a(x_i(\lambda),\lambda)\bigr]\neq0.
 \label{eq:face-transversality}
\end{equation}

The physical Maxwell set is the subset
\begin{equation}
 \mathcal M_{\mathrm{eq}}\subset\mathcal M_{\mathrm{cv}}
 \label{eq:physical-Maxwell-subset}
\end{equation}
on which the equal sheets are both active minima and their common value is
no larger than every other active-minimum value. Thus
\(\mathcal M_{\mathrm{cv}}\) records all critical-value reorderings, whereas
only part of it describes physical coexistence.

Set
\begin{equation}
 \Sigma=\Delta\cup\mathcal M_{\mathrm{cv}}\cup\mathcal B.
 \label{eq:total-wall}
\end{equation}
For a chamber \(C\), a connected component of \(U\setminus\Sigma\), define
\begin{equation}
 \mathfrak P_C
 =(E_C,\prec_x,\prec_{\mathcal F},\operatorname{ind},\epsilon).
 \label{eq:decorated-chamber-datum}
\end{equation}
Here \(E_C\) is the finite set of stationary sheets, \(\prec_x\) is their
state order, \(\prec_{\mathcal F}\) is their critical-value order,
\(\operatorname{ind}\colon E_C\to\{0,1\}\), and
\begin{equation}
 \epsilon_i=
 \bigl(\operatorname{sgn}\rho_1(x_i,\lambda),\ldots,
       \operatorname{sgn}\rho_q(x_i,\lambda)\bigr)
 \label{eq:face-sign-vector}
\end{equation}
is the face-sign vector. A sheet is active exactly when every component of
\(\epsilon_i\) is positive.

\begin{theorem}[Chamber constancy]
\label{thm:chamber-constancy}
Under the assumptions of Section~\ref{sec:constrained-families},
\(\mathfrak P_C\) is constant on each chamber, up to continuation of sheet
labels. Hence the active stable set, metastable set, global winner, and
presence or absence of an active stationary phase are constant on \(C\).

Because \(X\) is a global interval, state order canonically labels the sheets
and excludes within-chamber monodromy in the category considered here.
\end{theorem}

Properness and finite fibers make the critical projection over a chamber a
finite covering.  Since no Hessian, critical-value difference, or physical
margin vanishes away from \(\Sigma\), the Morse indices, value order, and face
signs are constant on each lifted sheet.  The metastable set and the global
winner then follow from these finite data.  The covering argument is given in
\ref{app:transversality}.

For a chamber in which no active zero lies on a state-space endpoint, define
\begin{equation}
 W_C=\sum_{i\in E_C^{\mathrm{act}}}(-1)^{\operatorname{ind}_i}.
 \label{eq:active-signed-count}
\end{equation}
This is the Brouwer degree of \(\mathcal F_x\) on the admissible components
when those components are truncated to a compact domain with nonzero
boundary gradient. Without these boundary hypotheses,
\eqref{eq:active-signed-count} is an active signed count, not automatically
a degree.

\subsection{Wall actions and minimality}
\label{subsec:wall-actions-minimality}

\begin{proposition}[Regular wall actions]
\label{prop:wall-effects}
At a regular point of one wall, away from codimension-two incidences:
\begin{enumerate}
\item crossing \(\Delta\) inserts or deletes two state-adjacent points of
opposite Morse parity; their face signs agree, and, if the fold value differs
from every other critical value, their values are adjacent in
\(\prec_{\mathcal F}\);
\item crossing \(\mathcal M_{\mathrm{cv}}\) preserves sheets, indices, and
face signs and exchanges two adjacent entries of
\(\prec_{\mathcal F}\); the winner changes exactly when these entries are
the two lowest active minima;
\item crossing \(\mathcal B_a\) preserves the ambient sheets, indices, and
critical-value order but toggles one face sign of one sheet.
\end{enumerate}
The active signed count is unchanged at a fold or value wall. At a face
wall, a contribution \((-1)^{\operatorname{ind}_i}\) is added or removed only if every
other face sign of that sheet is positive.
\end{proposition}

The wall-action table and its geometric illustration are collected in
\ref{appsubsec:operator-relations}, Table~\ref{tab:wall-actions},
and Fig.~\ref{fig:wall-operators}.

\begin{proposition}[Minimality by wall type]
\label{prop:wall-minimality}
No one of \(\Delta\), \(\mathcal M_{\mathrm{cv}}\), and \(\mathcal B\) can
be removed from a reconstruction of the full stationary, critical-value-order,
and admissibility portrait \eqref{eq:decorated-chamber-datum}.  This minimality is
by information type; an equilibrium-only reconstruction may omit value walls
involving unstable sheets or metastable minima above the winner.
\end{proposition}

The fold, shifted-Maxwell, and physical-exit counterexamples proving
Proposition~\ref{prop:wall-minimality} are given in
\ref{app:transversality}.  Reconstruction of complete critical-value order or
barrier order requires the extended set \(\mathcal M_{\mathrm{cv}}\).

\subsection{Codimension-two coherence}
\label{subsec:codimension-two-coherence}

Codimension-two strata supply the relations needed to compare different
wall-crossing paths. Branch labels must be distinguished from ranks in
\(\prec_{\mathcal F}\): a value swap commutes with a face toggle in
branch-labeled notation, while in rank notation it transports the rank
at which the toggle acts. A transverse equality of three nondegenerate
critical values produces the braid relation among adjacent value swaps.
The operator identities and the six-sector calculation are given in
\ref{appsubsec:operator-relations}.

The versal \(A_3\) cusp and a transverse fold--face incidence supply
representative non-product relations.  At the latter, the two face walls are
\(s=\pm\sqrt t\), so the fold insertion and face toggle need not commute.
The local normal forms and chamber activity patterns are derived in
\ref{app:wall-normal-forms} and displayed in
Fig.~\ref{fig:codimension-two-coherence}.  A reconstruction records every
codimension-two stratum present in the chosen control window, including
simultaneous folds, fold--remote-value resonances, multiple face crossings,
and intersections of physical faces.

\subsection{Conditional reconstruction and covariance}
\label{subsec:reconstruction-covariance}

\begin{theorem}[Conditional stratified reconstruction]
\label{thm:path-reconstruction}
Suppose that the following data are specified on \(U\):
\begin{enumerate}
\item the oriented Whitney-stratified wall set \(\Sigma\), with branch and
face incidence labels on every regular codimension-one stratum;
\item one decorated datum \(\mathfrak P_{C_0}\) in a base chamber;
\item the fold, value-swap, or face-toggle operator on every oriented
regular wall; and
\item the complete local continuation relation at every codimension-two
stratum met by the path homotopies under consideration.
\end{enumerate}
Composition of wall operators reconstructs the qualitative stationary,
value-order, and admissibility portrait along every generic path beginning
in \(C_0\). Two paths with the same endpoints give the same result if they
are joined, relative to their endpoints, by a generic homotopy that meets
only the recorded codimension-two strata.  Path independence therefore holds
within each such generic homotopy class.
\end{theorem}

Along a generic path, successive wall crossings compose the corresponding
fold, value-swap, and face-toggle operators.  A generic homotopy changes the
crossing word only through the recorded codimension-two moves, whose local
relations preserve the endpoint chamber data.  This gives path independence in the
stated homotopy class; the stratified continuation proof is given in
\ref{app:wall-normal-forms}.

The relevant covariance is narrower than covariance under an arbitrary
Legendre transformation. Consider two constrained families
\((\mathcal F,\rho_a,X,\Lambda)\) and
\((\mathcal F',\rho'_a,X',\Lambda')\). A phase-compatible equivalence
consists of a control diffeomorphism
\(\psi\colon\Lambda\to\Lambda'\), state diffeomorphisms
\(\phi_\lambda\colon X\to X'\), a permutation \(\varsigma\) of the physical
faces, and functions satisfying
\begin{align}
 \mathcal F'\bigl(\phi_\lambda(x),\psi(\lambda)\bigr)
 &=a(\lambda)\mathcal F(x,\lambda)+b(\lambda),
 \qquad a(\lambda)>0,
 \label{eq:phase-compatible-potential}\\
 \rho'_a\bigl(\phi_\lambda(x),\psi(\lambda)\bigr)
 &=u_a(x,\lambda)\rho_{\varsigma(a)}(x,\lambda),
 \qquad u_a(x,\lambda)>0.
 \label{eq:phase-compatible-faces}
\end{align}

\begin{theorem}[Phase-compatible covariance]
\label{thm:phase-compatible-covariance}
Under \eqref{eq:phase-compatible-potential} and
\eqref{eq:phase-compatible-faces}, \(\psi\) carries
\(\Delta\), \(\mathcal M_{\mathrm{cv}}\),
\(\mathcal M_{\mathrm{eq}}\), and \(\mathcal B\) to their primed
counterparts. It preserves Morse indices, critical-value order,
admissibility, metastability, the global winner, the active signed count,
and the wall-incidence relations. State order is transported by
\(\phi_\lambda\) and reverses if \(\phi_\lambda\) reverses orientation.

An ensemble enlarged by prescribed nonstationary reference sheets requires
the extended value walls described in
Section~\ref{subsec:stationary-sheets}; phase compatibility then also requires
the same positive affine map on every reference value.
\end{theorem}

\ref{app:transversality} proves the theorem by differentiating
\eqref{eq:phase-compatible-potential}. The positive affine relation is
sufficient but not necessary; it may be replaced by
\(\mathcal F'=\Theta(\mathcal F,\lambda)\) with
\(\partial\Theta/\partial\mathcal F>0\), provided that the same increasing map
acts on every competing value sheet and the state and face maps remain as
above. In contact language the corresponding structure is the triple
\((L,\mathfrak p,\mathcal A)\) comprising the equilibrium Legendrian, reservoir
polarization, and admissible domain. A general Legendre transformation can
change any of these objects and need not preserve phase dominance
\cite{GhoshBhamidipati2019Contact,Bravetti2019Contact}.
Reference~\cite{Li2026LegendreCovariant} established the covariance of the
fixed-reservoir Brouwer--Morse degree.  The refinement here decorates that
stationary invariant by critical-value order and physical membership, and
restricts the allowed action map so that the order of all competing phases
is preserved.

The active signed count \eqref{eq:active-signed-count} is a coarse
projection of \(\mathfrak P_C\). A Maxwell crossing can change the global
phase while leaving that count, the branch cover, and every Morse index
unchanged.

\section{Conclusions}
\label{sec:conclusions}

In the standard single-string ensemble, the charged slowly accelerating AdS
C-metric has a reentrant Maxwell turn at every positive string tension below
\begin{equation}
 \mu_+=
 \frac{\sqrt{2\sqrt3-3}}
 {2\bigl(1+\sqrt{2\sqrt3-3}\bigr)}.
\end{equation}
There is no positive lower threshold.  The unrestricted two-horizon
equations force the coexisting solutions to have a common value
\(Ar_+=\chi\), while their charge--acceleration coordinates remain
distinct.  The turn is then fixed by
\begin{equation}
 \mu=\frac{\chi}{(1+\chi)^2},\qquad
 \frac{P_{\rm turn}}{P_t}=1+\chi^2,
 \qquad QT_{\rm turn}=\frac{\mu}{\pi}.
\end{equation}
Both members of the pair are outer, slowly accelerating, entropy-regular
strict canonical minima.  Algebraic reconstruction of all remaining
stationary states shows that their common Gibbs value is lower than that of
every other physical equilibrium at the same controls.  The turn has positive
curvature and nonzero latent heat throughout the open interval; it ends at an
ordinary critical merger at \(\mu_+\) and at a Maxwell--admissibility corner
at zero tension.

The lower endpoint resolves the apparent disappearance of reentrance in
finite-resolution phase diagrams.  The displacement from the snapping line
is fourth order in pressure and third order in temperature,
\begin{equation}
 (P_{\rm turn}-P_t)Q^2\sim\frac{3\mu^4}{8\pi},
 \qquad
 Q(T_t-T_{\rm turn})\sim\frac{\mu^3}{2\pi}.
\end{equation}
After the corresponding blow-up, the full boundary-linked Maxwell segment
has the limiting profile
\begin{equation}
 \widehat T=2\sqrt{\widehat P}-\widehat P.
\end{equation}
The two phases live in different state-space charts.  Their entropy order
remains fixed, but their thermodynamic volumes exchange order at the turn.
The Clapeyron equation identifies volume inversion as the local mechanism of
the reentrant reversal.  At the same time the entropy gap, latent heat, and
stationary barrier diverge.

The C-metric also realizes the leading-profile classification at a Maxwell
boundary with fixed sheet incidence and reservoir polarization.  For two
nondegenerate minima, the critical-value gap \(D=G_L-G_S\) obeys
\(D_T=-\Delta S\) and \(D_P=\Delta V\).  The leading pressure contact and the
projective Newton class of \(\Delta V/\Delta S\) determine the normalized
coexistence profile.  In the integrable regime \(m>n\), its regular turns are
the positive simple roots of the Newton polynomial; \(m=n\) and \(m<n\)
produce logarithmic and power-divergent boundary classes.  For the C-metric,
\((a,b,c;m,n,k)=(4,2,3;2,1,1)\) has the primitive binomial shape
\((m,n,k)=(2,1,1)\), and its two-chart incidence yields the weighted primitive
\(2\sqrt{\widehat P}-\widehat P\).

The wall complex separates stationary folds, critical-value order, and
physical membership.  Two extrema of the Hawking temperature determine the
local three-sheet cover, but a first-order transition also requires endpoint
value reversal, admissibility, and comparison with every competing phase.
Likewise, a finite jump of the selected free energy signals a failure of the
continuity, compactness, or nonemptiness hypotheses for the feasible family,
deletion of a limiting state, a discontinuous action normalization, or an
omitted competitor.

The topological action shift found by Hale et al.\
\cite{HaleEtAl2025ChargedAccelerating} reduces in the present fixed-charge
ensemble to \(3\mu T/(4P)\) on every canonical branch.
Absolute free energies therefore depend on the prescription, but all branch
differences and the phase-selection results proved here do not.  The exact
turning locus, the global winner, \(\Delta S\), \(\Delta V\), the Clapeyron
slope, latent heat, curvature, and stationary barrier are identical
in the two prescriptions.

The exact result identifies the apparent low-tension disappearance of the
reentrant turn as a boundary-scaling effect.  Near a moving physical boundary,
reconstruction of the phase diagram requires the stationary structure, the
critical-value differences, and the admissibility of the competing branches.

It remains open which other projective Newton classes occur in complete
gravitational solutions.  Rotating accelerating black holes, Lovelock systems
with multiple reentrant transitions, and Born--Infeld AdS black holes provide
natural test cases because their additional scales may change the leading
thermodynamic-jump polynomial
\cite{AbbasvandiEtAl2019FinelySplit,FrassinoEtAl2014MultipleReentrant,
GunasekaranKubiznakMann2012Extended}.  Determining their classes requires the
local stationary branches together with the physical Maxwell component and
its boundary incidence.

\appendix
\renewcommand{\thetheorem}{\Alph{section}.\arabic{theorem}}
\numberwithin{figure}{section}
\numberwithin{table}{section}
\renewcommand{\thefigure}{\Alph{section}.\arabic{figure}}
\renewcommand{\thetable}{\Alph{section}.\arabic{table}}
\section{Transversality, chamber constancy, and covariance}
\label{app:transversality}

\subsection{Finite sheet coverings}
\label{appsubsec:finite-covering}

\begin{lemma}
\label{lem:finite-critical-covering}
Let \(C\) be a chamber. The restriction
\(\operatorname{pr}\colon\mathscr C_{\mathrm{reg}}\vert_C\to C\) is a finite covering.
If \(X\) is an interval, the state order trivializes this covering.
\end{lemma}

\begin{proof}
At every point of \(\mathscr C_{\mathrm{reg}}\),
\(\mathcal F_{xx}\neq0\). The implicit-function theorem therefore makes
\(\operatorname{pr}\) a local diffeomorphism. By assumption,
\(\operatorname{pr}\) is proper and its
fibers are finite. A proper local diffeomorphism with finite fibers is a
finite covering.

Suppose now that \(X\) is an interval. The number of stationary points in a
fiber is constant on \(C\); write them as
\begin{equation}
 x_1(\lambda)<x_2(\lambda)<\cdots<x_n(\lambda).
 \label{eq:ordered-critical-points}
\end{equation}
Each \(x_k\) agrees locally with one of the smooth implicit-function
sections. It is therefore continuous and locally smooth on all of \(C\).
The sections \(x_1,\ldots,x_n\) trivialize the covering. In particular, a
loop in \(C\) cannot permute the sheets while preserving their strict
linear order.
\end{proof}

\begin{proof}[Proof of Theorem~\ref{thm:chamber-constancy}]
Lemma~\ref{lem:finite-critical-covering} makes the stationary-sheet number
constant. The sign of \(\mathcal F_{xx}\) cannot change without crossing
\(\Delta\). Distinct stationary points cannot exchange state order without
colliding, and their critical values cannot exchange order without crossing
\(\mathcal M_{\mathrm{cv}}\). Finally, a component
\(\rho_a(x_i(\lambda),\lambda)\) of a face-sign vector cannot change sign
without crossing \(\mathcal B_a\). Every entry of
\(\mathfrak P_C\) is locally constant and hence constant on the connected
chamber. The active stable set, metastable set, and winner are functions of
this finite record and are therefore constant as well.

For a state space without a global order, the first part of
Lemma~\ref{lem:finite-critical-covering} still gives a finite covering, but
its continuation around a loop can permute sheets. This is the stated
monodromy qualification.
\end{proof}

\subsection{Regular wall actions and signed count}
\label{appsubsec:regular-wall-actions}

\begin{proof}[Proof of Proposition~\ref{prop:wall-effects}]
At a regular fold, smooth changes of the state and controls give the local
form
\begin{equation}
 \mathcal F(u;s,\eta)=\mathcal F_0(s,\eta)
 +\sigma\left(\frac{u^3}{3}-su\right),
 \qquad \sigma\in\{+1,-1\},
 \label{eq:fold-normal-form-app}
\end{equation}
after smooth state and control changes and addition of a
control-dependent function
\cite{Golubitsky1978Catastrophe}. For \(s>0\), its stationary points are
\(u_\pm=\pm\sqrt s\). Their Hessians
\begin{equation}
 \mathcal F_{uu}(u_\pm;s,\eta)=2\sigma u_\pm
 \label{eq:fold-Hessians}
\end{equation}
have opposite signs. Their critical values approach the same fold value as
\(s\downarrow0\). If the fold is not on an admissibility wall, continuity
makes every face sign of \(u_+\) equal to the corresponding sign of
\(u_-\) for sufficiently small \(s>0\). If the fold value is unequal to
every remote critical value, the two newborn values lie in one value gap
and are adjacent in \(\prec_{\mathcal F}\). Their contributions to
\eqref{eq:active-signed-count} cancel when active and are both absent when
inactive.

At a regular critical-value wall, \eqref{eq:value-transversality} implies
that \(v_i-v_j\) changes sign. No sheet, Hessian, or face sign changes.
Because no third value is equal at a regular wall point, \(i\) and \(j\)
are adjacent in value order. The selected phase changes only if both are
active minima and no active minimum lies below them.

At a regular admissibility wall,
\eqref{eq:face-transversality} changes one component of one face-sign
vector. The ambient point, its Hessian, and its critical value continue
smoothly. Its signed contribution is added or removed exactly when all
other face signs are positive.
\end{proof}

\begin{lemma}
\label{lem:active-degree}
Let \(J\subset X\) be a finite union of bounded open intervals whose
endpoints are not zeros of \(\mathcal F_x(\,\cdot\,,\lambda)\). If every
zero in \(J\) is nondegenerate, then
\begin{equation}
 \deg(\mathcal F_x,J,0)
 =\sum_{x_i\in J}\operatorname{sgn}\mathcal F_{xx}(x_i,\lambda)
 =\sum_{x_i\in J}(-1)^{\operatorname{ind}_i}.
 \label{eq:degree-local-sum}
\end{equation}
\end{lemma}

\begin{proof}
This is the regular-value formula for the one-dimensional Brouwer degree.
The boundary condition makes the degree well defined. A minimum contributes
\(+1\) and a maximum contributes \(-1\).
\end{proof}

\subsection{Independence of the wall types}
\label{appsubsec:minimality-proof}

\begin{proof}[Proof of Proposition~\ref{prop:wall-minimality}]
Each wall type controls information that the other two do not contain.

For the fold type, take
\begin{equation}
 \mathcal F(u;s)=\frac{u^3}{3}-su.
 \label{eq:fold-necessity-example}
\end{equation}
There is no stationary point for \(s<0\) and there are two for \(s>0\).
If \(\Delta=\{s=0\}\) is omitted, sheet number and Morse data are not
constant on the purported chamber.

For the critical-value type, restrict the family of
Theorem~\ref{thm:local-data-insufficient} to a small \(t\)-interval
containing \(t_M(\varepsilon)\). No fold or face crossing occurs there, but
the active minima exchange dominance. Equation \eqref{eq:Maxwell-shift}
also shows that the common stationary skeleton cannot recover the wall
position.

For the admissibility type, use
\(\mathcal F_\kappa\) from \eqref{eq:boundary-deletion-example} with the
single face
\begin{equation}
 \rho(x,\lambda)=x-\lambda
 \label{eq:exit-necessity-face}
\end{equation}
near \(\lambda=-1\). The ambient stationary set, indices, and values do not
change, but the lower minimum is removed when the face crosses it. Omitting
\(\mathcal B\) loses both active phase existence and the winner change.
\end{proof}

\subsection{Phase-compatible equivalence}
\label{appsubsec:covariance-proof}

\begin{proof}[Proof of Theorem~\ref{thm:phase-compatible-covariance}]
Differentiate \eqref{eq:phase-compatible-potential} with respect to the
state coordinate at fixed control:
\begin{equation}
 \mathcal F'_{x'}(\phi_\lambda(x),\psi(\lambda))
       \partial_x\phi_\lambda(x)
 =a(\lambda)\mathcal F_x(x,\lambda).
 \label{eq:covariance-first-derivative}
\end{equation}
Since \(\partial_x\phi_\lambda\neq0\), stationary points correspond. At a
stationary point, a second differentiation gives
\begin{equation}
 \mathcal F'_{x'x'}
 =\frac{a(\lambda)}{(\partial_x\phi_\lambda)^2}
  \mathcal F_{xx}.
 \label{eq:covariance-Hessian}
\end{equation}
The prefactor is positive, so the Morse index is preserved. At a fold,
where the first two state derivatives vanish, a third differentiation gives
\begin{equation}
 \mathcal F'_{x'x'x'}
 =\frac{a(\lambda)}{(\partial_x\phi_\lambda)^3}
  \mathcal F_{xxx}.
 \label{eq:covariance-third}
\end{equation}
The fiber-preserving diffeomorphism
\((x,\lambda)\mapsto(\phi_\lambda(x),\psi(\lambda))\) also preserves the
transversality of the projected fold. Thus \(\Delta\) maps to
\(\Delta'\).

For two stationary points,
\begin{equation}
 \mathcal F'(x'_i,\lambda')-\mathcal F'(x'_j,\lambda')
 =a(\lambda)\bigl[
   \mathcal F(x_i,\lambda)-\mathcal F(x_j,\lambda)\bigr].
 \label{eq:covariance-value-difference}
\end{equation}
The positivity of \(a\) preserves both equality and order of critical
values. At a value equality,
\(d[a(v_i-v_j)]=a\,d(v_i-v_j)\), so regular-wall transversality is
preserved as well.

Equation \eqref{eq:phase-compatible-faces} preserves every face sign. At a
face zero,
\begin{equation}
 d\bigl(u_a\rho_{\varsigma(a)}\bigr)
 =u_a\,d\rho_{\varsigma(a)},
 \label{eq:covariance-face-transversality}
\end{equation}
so the face-wall transversality and incidence labels are also preserved.
All components of the chamber datum, the active-minimum order, the winner,
and the active signed count correspond. Since the entire stratified
incidence diagram is transported by a diffeomorphism, its local wall
relations correspond as well.  In an enlarged complex containing prescribed
reference-value sheets, the same argument applies only when those values are
transformed by the same positive affine map.
\end{proof}

\section{Wall normal forms and coherence relations}
\label{app:wall-normal-forms}

\subsection{Operator relations}
\label{appsubsec:operator-relations}

Let \(\mathbf c^+_{pq}\) and \(\mathbf c^-_{pq}\) insert and delete a fold
pair, including its state slots, value slots, indices, and common face
signs. Let \(\mathbf m_{pq}\) exchange the value order of two fixed branch
labels, and let \(\mathbf a_{pa}\) toggle face \(a\) of branch \(p\). Where
the partial fold maps are defined, their action on the finite chamber datum
gives
\begin{equation}
 \mathbf c^-_{pq}\mathbf c^+_{pq}=\mathrm{id},\qquad
 \mathbf m_{pq}^2=\mathrm{id},\qquad
 \mathbf a_{pa}^2=\mathrm{id}.
 \label{eq:operator-inverses}
\end{equation}

If two events involve disjoint branch and face labels, they alter disjoint
entries of \(\mathfrak P_C\) and commute. For a value swap and a face toggle
involving the same branch, the branch-labeled operators still commute:
one changes only \(\prec_{\mathcal F}\), and the other changes only
\(\epsilon\). In rank notation the branch occupying rank \(i\) moves to
rank \(i+1\) under \(\mathbf m_i\). Therefore
\begin{equation}
 \mathbf m_i\mathbf a_{i,a}
 =\mathbf a_{i+1,a}\mathbf m_i,\qquad
 \mathbf m_i\mathbf a_{i+1,a}
 =\mathbf a_{i,a}\mathbf m_i,
 \label{eq:Maxwell-admissibility-square}
\end{equation}
The Maxwell--admissibility square therefore follows from the rank convention.
On the side where one competitor is
inactive, its equality wall remains in \(\mathcal M_{\mathrm{cv}}\) but not
in \(\mathcal M_{\mathrm{eq}}\).

Suppose three distinct nondegenerate sheets have equal critical values at
\(\lambda_0\), and assume
\begin{equation}
 \operatorname{rank}d(v_1-v_2,v_2-v_3)\big|_{\lambda_0}=2.
 \label{eq:triple-value-rank}
\end{equation}
In a transverse two-dimensional slice, \(v_1-v_2\) and \(v_2-v_3\) are
coordinates. The third equality is their sum. The three equality lines cut
the slice into the six possible total orders. The two crossing words
connecting opposite sectors are
\(\mathbf m_i\mathbf m_{i+1}\mathbf m_i\) and
\(\mathbf m_{i+1}\mathbf m_i\mathbf m_{i+1}\); both induce the same
permutation of the three branch labels. Hence
\begin{equation}
 \mathbf m_i\mathbf m_{i+1}\mathbf m_i
 =\mathbf m_{i+1}\mathbf m_i\mathbf m_{i+1},\qquad
 \mathbf m_i\mathbf m_j=\mathbf m_j\mathbf m_i
 \quad(|i-j|>1).
 \label{eq:braid-and-commutation}
\end{equation}

Table~\ref{tab:wall-actions} records the three regular codimension-one
actions in the same notation.

\begin{table}[t]
\centering
\small
\caption{Changes carried by the three regular wall types. The statement
about \(W_C\) at an admissibility wall concerns a sheet that is otherwise
admissible.}
\label{tab:wall-actions}
\setlength{\tabcolsep}{3pt}
\begin{tabular}{p{0.12\linewidth}p{0.26\linewidth}
                p{0.34\linewidth}p{0.12\linewidth}}
\hline
{\raggedright Wall\par} &
{\raggedright Stationary and admissibility data\par} &
{\raggedright Value order and selection\par} &
{\raggedright \(W_C\)\par}\\
\hline
{\raggedright \(\Delta\)\par} &
{\raggedright one adjacent pair with opposite indices and equal face signs
is inserted or deleted\par} &
{\raggedright two adjacent values are inserted or deleted; metastability
can change\par} &
{\raggedright fixed\par}\\
{\raggedright \(\mathcal M_{\mathrm{cv}}\)\par} &
{\raggedright all sheets, indices, and face signs are unchanged\par} &
{\raggedright adjacent values swap; the winner changes only at the lowest
active-minimum swap\par} &
{\raggedright fixed\par}\\
{\raggedright \(\mathcal B\)\par} &
{\raggedright ambient sheets and indices are unchanged; one face sign
toggles\par} &
{\raggedright active order and the winner can change; a no-phase chamber can
appear\par} &
{\raggedright changes by \(\pm(-1)^{\operatorname{ind}}\)\par}\\
\hline
\end{tabular}
\end{table}

\begin{figure}[t]
 \centering
 \includegraphics[width=\linewidth]{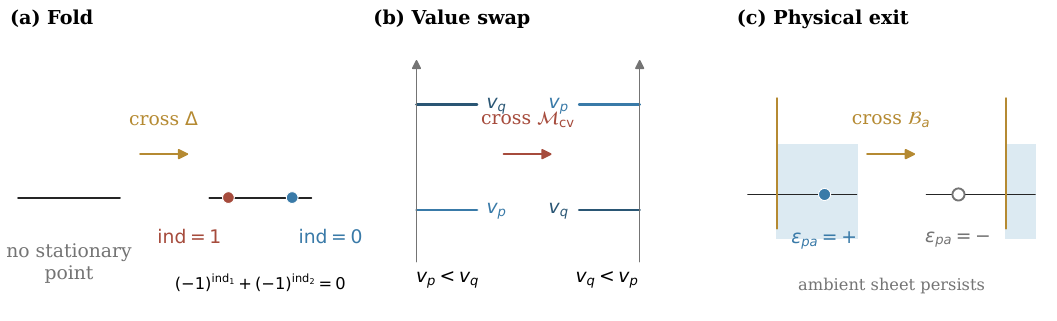}
 \caption{The three regular codimension-one wall actions. Crossing
 \(\Delta\) inserts a state-adjacent pair of opposite Morse parity;
 crossing \(\mathcal M_{\mathrm{cv}}\) swaps two adjacent critical values;
 crossing \(\mathcal B_a\) changes one face sign while the ambient sheet
 persists. Filled and open points denote active and inactive sheets.}
 \label{fig:wall-operators}
\end{figure}

\subsection{The cusp movie}
\label{appsubsec:cusp-movie}

The versal \(A_3\) potential germ can be written as
\begin{equation}
 \mathcal F(u;a,b)=\mathcal F_0(a,b)
 +\sigma\left(\frac{u^4}{4}+\frac{a u^2}{2}+bu\right),
 \qquad \sigma\in\{+1,-1\},
 \label{eq:cusp-normal-form-app}
\end{equation}
after smooth changes of state and control coordinates and addition of a
control-dependent function \cite{Golubitsky1978Catastrophe}. No
state-dependent multiplication of the potential is used, because such a
multiplication would not in general preserve stationary values. Its
stationary equation is
\begin{equation}
 u^3+au+b=0.
 \label{eq:cusp-stationary-cubic}
\end{equation}
A double root obeys
\begin{equation}
 u^3+au+b=0,\qquad 3u^2+a=0.
 \label{eq:cusp-double-root}
\end{equation}
Eliminating \(u\) gives
\begin{equation}
 4a^3+27b^2=0,\qquad a\leq0.
 \label{eq:cusp-discriminant-app}
\end{equation}
The interior \(4a^3+27b^2<0\) has three stationary points, and the exterior
has one.

Writing the double root as \(u=r\) gives the explicit parametrization
\begin{equation}
 a=-3r^2,\qquad b=2r^3.
 \label{eq:cusp-fold-parametrization}
\end{equation}
For \(r<0\), the double root is the left pair \(u_1=u_2=r\) and the simple
root is \(-2r>0\). For \(r>0\), the double root is the right pair
\(u_2=u_3=r\) and the simple root is \(-2r<0\). Thus the two fold arcs
coalesce different state-adjacent pairs.

For \(\sigma=+1\), the three points at \(b=0\), \(a<0\), are
\begin{equation}
 u_1=-\sqrt{-a},\qquad u_2=0,\qquad u_3=\sqrt{-a}.
 \label{eq:cusp-symmetric-roots}
\end{equation}
The outer points are minima with the common value
\(\mathcal F_0-a^2/4\), while the middle point is a maximum with value
\(\mathcal F_0\). For \(\sigma=-1\), all indices reverse. Thus the exact
normal form contains a value-equality curve \(b=0\) terminating at the cusp.

A path entering the cusp interior through one fold arc creates a
state-adjacent pair. A path leaving through the other arc deletes a
different adjacent pair. The remaining one-point configuration agrees with
the configuration obtained by continuation outside the cusp. After the two
fold crossings,
the surviving sheet may carry a different genealogical label. Coherence is
encoded by the specified cusp-continuation isomorphism between genealogical
labels.

\subsection{Fold--boundary incidence}
\label{appsubsec:fold-boundary}

\begin{lemma}[Relative fold--face normal form]
\label{lem:relative-fold-face}
Suppose \(p\) is an ordinary \(A_2\) fold on a regular physical face.
Assume that \(\rho_u(p)\ne0\) and that, after choosing the fold unfolding
coordinate \(t\), the graph of the face has an independent control
direction \(s\). Then local state and control diffeomorphisms and
multiplication of \(\rho\) by a positive unit put the pair in the form
\begin{align}
 \mathcal F(u;t,s,\eta)
 &=\mathcal F_0(t,s,\eta)
 +\sigma\left(\frac{u^3}{3}-tu\right),
 \label{eq:fold-boundary-normal-form-app-potential}\\
 \rho(u;t,s,\eta)&=\upsilon(u,t,s,\eta)(u-s),
 \qquad \upsilon>0.
 \label{eq:fold-boundary-normal-form-app}
\end{align}
\end{lemma}

\begin{proof}
The parametrized \(A_2\) normal-form theorem first gives the displayed
form of \(\mathcal F\), with the remaining controls passive. Since
\(\rho_u(p)\ne0\), the implicit-function theorem writes the face as
\(u=h(t,\zeta,\eta)\). Independence of the face direction means that one
passive control can be chosen with
\(\partial h/\partial\zeta\ne0\) at \(p\). The base change
\(s=h(t,\zeta,\eta)\) is then a local control diffeomorphism and leaves the
fold state coordinate unchanged. Hadamard's lemma gives
\(\rho=\upsilon(u-s)\). If \(\upsilon(p)<0\), replace
\(u\) and \(s\) by \(u'=-u\) and \(s'=-s\), and replace
\(\sigma\) by \(\sigma'=-\sigma\). The fold potential retains the displayed
\(A_2\) form, while
\(\rho=(-\upsilon)(u'-s')\) has a positive unit. This coordinate change
preserves the original physical inequality \(\rho>0\); no reversal of the
face defining function is used.
\end{proof}

The fold is \(t=0\). For \(t>0\), the stationary points are
\begin{equation}
 u_-(t)=-\sqrt t,\qquad u_+(t)=\sqrt t.
 \label{eq:fold-boundary-roots}
\end{equation}
The face incidences \(\rho(u_\pm;t,s)=0\) are
\begin{equation}
 s=-\sqrt t,\qquad s=\sqrt t.
 \label{eq:fold-boundary-walls-app}
\end{equation}
With the convention \(\rho>0\), the activity table is
\begin{equation}
\begin{array}{c|c}
\text{control region} & \text{active newborn sheets}\\
\hline
s<-\sqrt t & u_-,u_+\\
-\sqrt t<s<\sqrt t & u_+\\
s>\sqrt t & \varnothing .
\end{array}
\label{eq:fold-boundary-ledger}
\end{equation}
Both face walls end on the fold. Hence there is no globally defined
commuting square obtained by moving a face toggle past a fold insertion.
The three-region incidence structure, transported through the local coordinate
equivalence, supplies the required coherence identification.

\begin{figure}[t]
 \centering
 \includegraphics[width=\linewidth]{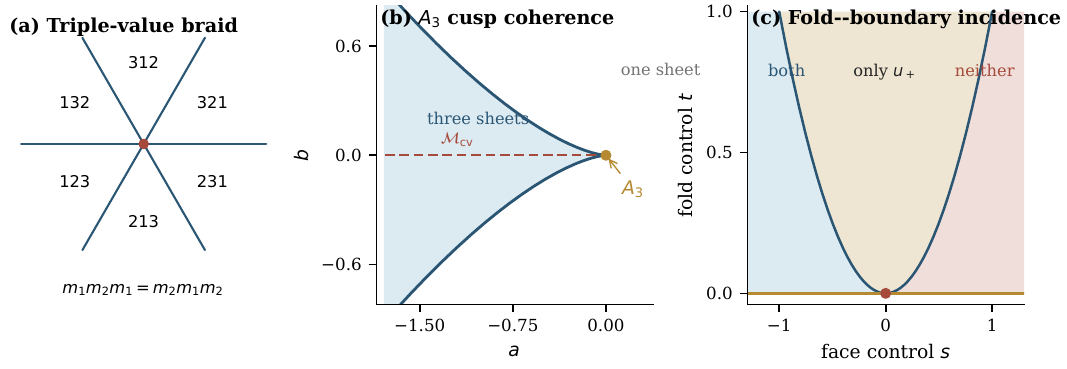}
 \caption{Codimension-two coherence. (a) Three transverse value-equality
 walls divide a normal slice into the six strict orders and implement the
 braid relation. (b) The fold discriminant \(4a^3+27b^2=0\) of the versal
 \(A_3\) potential bounds the three-sheet region, while \(b=0\) is the
 equality wall of the outer stationary values. (c) In the relative
 fold--face normal form, \(s=\pm\sqrt t\) bound chambers with two, one, and
 no active newborn sheets and terminate together at \(t=0\).}
 \label{fig:codimension-two-coherence}
\end{figure}

\subsection{Proofs for the resolved leading-profile classification}
\label{app:maxwell-corner-universality-proof}

\begin{proof}[Proof of Proposition~\ref{prop:corner-gap-reduction}]
The two stationary sheets are nondegenerate and have Morse index zero.
The parameterized Morse lemma \cite{GolubitskyGuillemin1973} can therefore be
applied in disjoint state neighborhoods, with the reservoir variables and
\(\epsilon\) retained as
parameters. A further positive rescaling of each state coordinate gives
Eq.~\eqref{eq:corner-two-well-Morse-form}. The comparison of the two
minimum values is consequently the comparison of \(g_S\) and \(g_L\), and
their equality is \(D=0\).

Let a common increasing reparametrization act on the two critical values,
\(g_i\mapsto \Theta(g_i,\vartheta,\pi,\epsilon)\), with
\(\Theta_z>0\). Taylor's
formula gives
\begin{equation}
 \Theta(g_L,\vartheta,\pi,\epsilon)
 -\Theta(g_S,\vartheta,\pi,\epsilon)
 =D\int_0^1
 \Theta_z(g_S+tD,\vartheta,\pi,\epsilon)\,\dd t.
 \label{eq:app-gap-positive-unit}
\end{equation}
The integral is a positive unit. The same conclusion follows immediately
from the selected-value relation
Eq.~\eqref{eq:corner-equivalence-potential}. Finally,
Eq.~\eqref{eq:corner-equivalence-face} maps each selected-sheet
admissibility margin to a positive unit times that margin. Hence both the
sign of the value gap and the physical membership of each selected sheet
are preserved. No assertion about an off-sheet extension of the face
function is needed.
\end{proof}

\begin{proof}[Proof of Theorem~\ref{thm:maxwell-corner-Newton-classification}]
On a fixed \(K\Subset(0,\infty)\), division of
Eqs.~\eqref{eq:corner-universality-volume} and
\eqref{eq:corner-universality-entropy} gives
\begin{equation}
 \frac{\Delta v}{\Delta s}
 =-\frac{c_v}{c_s}\epsilon^{b-c}\rho^{-n}
  \left[R(r)+E_v(\epsilon,r)\right],
 \qquad
 \|E_v\|_{C^1(K)}=O_K(\epsilon^\delta).
 \label{eq:app-corner-general-jump-ratio}
\end{equation}
Here multiplication by \((1+R_s)^{-1}\) preserves the estimate because
\(R_s=O_K(\epsilon^\delta)\). Differentiating
Eq.~\eqref{eq:corner-universality-pressure}, using
\(\rho=\rho_*r\), yields
\begin{equation}
 \frac{\dd\pi_M}{\dd r}
 =m c_p\rho_*^m\epsilon^a r^{m-1}
  \left[1+E_p(\epsilon,r)\right],
 \qquad
 \|E_p\|_{C^1(K)}=O_K(\epsilon^\delta).
 \label{eq:app-corner-general-pressure-derivative}
\end{equation}
It is strictly positive for small \(\epsilon\). Combining the last two
equations with Eq.~\eqref{eq:corner-universality-clapeyron} gives
\begin{equation}
 \frac{\dd\vartheta_M}{\dd r}
 =-\mathcal A_\epsilon r^{\iota-1}
  \left[R(r)+E(\epsilon,r)\right],
 \qquad
 \|E\|_{C^1(K)}=O_K(\epsilon^\delta),
 \label{eq:app-corner-general-temperature-derivative}
\end{equation}
where \(\mathcal A_\epsilon\) is defined in
Eq.~\eqref{eq:corner-general-normalization}. Thus
\begin{equation}
 \frac{\dd\widehat\vartheta_\epsilon}{\dd r}
 =r^{\iota-1}\left[R(r)+E(\epsilon,r)\right].
 \label{eq:app-corner-normalized-derivative}
\end{equation}

For \(\iota>0\), every exponent \(\iota+\kappa_j\) is positive and
\begin{equation}
 \int_0^r s^{\iota-1}R(s)\,\dd s
 =\sum_{j=0}^{N}\frac{A_j}{\iota+\kappa_j}
  r^{\iota+\kappa_j}
 =\Phi_{\iota,R}(r;0).
 \label{eq:app-corner-integrable-primitive}
\end{equation}
The matching condition
\eqref{eq:corner-universality-matching} fixes the additive constant.
Integrating Eq.~\eqref{eq:app-corner-normalized-derivative} first between
two positive points, then taking the lower point to zero after
\(\epsilon\downarrow0\), proves uniform \(C^0\) convergence on \(K\).
The derivative estimate gives \(C^2(K)\) convergence. For
\(\iota\leq0\), integration from \(1\) to \(r\) fixes the constant without
a boundary limit and gives
\begin{equation}
 \int_1^r s^{\iota-1}R(s)\,\dd s
 =\sum_{j=0}^{N}A_jJ_{\iota+\kappa_j}(r;1).
 \label{eq:app-corner-singular-primitive}
\end{equation}
This proves Eq.~\eqref{eq:corner-general-profile}. Since
\(\widehat\pi_\epsilon\to r^m\) in \(C^2(K)\) and the limiting pressure
derivative is positive, eliminating \(r\) proves
Eq.~\eqref{eq:corner-general-eliminated}.

The leading term of the integrand at \(r=0\) is
\(A_0r^{\iota-1}\). It is integrable precisely for \(\iota>0\). At
\(\iota=0\) its primitive is \(A_0\log r\), and at \(\iota<0\) it is
\(A_0r^\iota/\iota\). Any monomial for which
\(\iota+\kappa_j=0\) contributes the additional logarithm in
Eq.~\eqref{eq:corner-primitive-monomial}. Positive multiplication of the
temperature coordinate cannot convert a bounded primitive into a logarithm
or a power, or convert a logarithm into a power. The three boundary types
are therefore inequivalent.

Invariance in the leading-profile quotient follows because the two-sided
\(C^2\) collar extension and the nonzero positive normal Jacobian in the
definition of equivalence send
the defining functions of \(\epsilon=0\) and \(r=0\) to positive units
times those defining functions.  The pressure trace in
Eq.~\eqref{eq:corner-leading-pressure-quotient} therefore preserves the
\(\epsilon\)-valuation \(a\) and the pressure contact order \(m\).  The
two conormal relations in
Eq.~\eqref{eq:corner-equivalence-jump-conormal} preserve, respectively,
the entropy valuation \(b\), the volume valuation \(c\), and the volume
pole order \(n\).  A factor such as \(\epsilon^t\) in the transverse
Jacobian would change these valuations, but it is excluded both by the
nonvanishing boundary Jacobian and by
Eq.~\eqref{eq:corner-equivalence-gap-conormal}.

At leading order,
\(\widetilde r=A_rr\) and the conormal volume relation has an
\(r\)-independent positive coefficient.  After the positive jump amplitudes
and the changes of \(\rho_*\) are absorbed into \(c_v\), its action on the
leading volume factor is therefore
\begin{equation}
 R(r)\longmapsto A R(B r),
 \qquad A>0,\quad B>0.
 \label{eq:app-corner-polynomial-action}
\end{equation}
Hence it preserves \([R]\). A common increasing transformation of the
selected critical values multiplies \(\dd D\) on \(D=0\) by a positive
unit, by
Eq.~\eqref{eq:app-gap-positive-unit}; the constant leading conormal in
Eq.~\eqref{eq:corner-equivalence-gap-conormal} prevents an additional
\(r\)-dependent Newton factor. Independent positive changes of the two
reservoir units alter only the constants
\(c_p,c_s,c_v,\lambda\). This proves invariance of \(\mathscr N\).

Conversely, equality of \([R]\) supplies positive constants \(A,B\) in
Eq.~\eqref{eq:app-corner-polynomial-action}. The dilation
\(r\mapsto B^{-1}r\), followed by positive rescalings of pressure,
temperature, entropy jump, and volume jump, identifies the leading terms of
the four Maxwell traces. Equality of \(a,b,c,m,n\) identifies their
valuations. Taking the smaller of the two positive remainder exponents,
their differences satisfy
Eqs.~\eqref{eq:corner-leading-pressure-quotient}--
\eqref{eq:corner-equivalence-jump-conormal}. The primitive calculation
above supplies the temperature relation in all three regimes. These
dilations and rescalings extend with their inverses to the resolved collar.
They may be chosen compatibly with the envelope identities.  In particular,
take \(A_D=A_\vartheta A_s\), and choose the positive pressure unit, or
equivalently the target normalization \(\widetilde\lambda\), so that
\(A_v=A_D\lambda/(A_\pi\widetilde\lambda)\).  Using \(D\) and
\(\widetilde D\) as normal coordinates, extend the collar map by
\(\widetilde D=A_DD\).  Then the entropy and volume conormals agree with the
chosen trace scalings and
Eq.~\eqref{eq:corner-equivalence-gap-conormal} holds with zero remainder.

For the selected-sheet margins, the positive dilation maps \(\rho=0\) to
\(\widetilde\rho=0\) and preserves the physical side.
Equation~\eqref{eq:corner-fixed-face-incidence} gives
\begin{equation}
 U_S=\frac{\widehat\Psi^*\widetilde\beta_S}{\beta_S}>0,
 \qquad
 U_L=\frac{\widehat\Psi^*\widetilde\beta_L}{\beta_L}>0,
 \label{eq:app-corner-sheetwise-face-units}
\end{equation}
and both ratios extend as positive \(C^2\) units to the resolved face.
This proves completeness of the leading-profile quotient conditional on
\(\mathfrak f_{S|L}\).

It remains to prove the assertion about turns. Since
\(\dd\pi_M/\dd r>0\) and \(\Delta s\ne0\), a turn is equivalent, by the
Clapeyron equation, to
\begin{equation}
 R(r)+E(\epsilon,r)=0.
 \label{eq:app-corner-perturbed-root-equation}
\end{equation}
Choose pairwise disjoint closed intervals \(I_j\Subset I\), each containing
one simple root \(r_j\), so small that \(R'\) has a fixed nonzero sign on
\(I_j\). On the complement of their interiors, compactness gives
\begin{equation}
 \min\left|R(r)\right|=d>0.
 \label{eq:app-corner-root-gap}
\end{equation}
The \(C^1\) bound on \(E\) makes the left-hand side of
Eq.~\eqref{eq:app-corner-perturbed-root-equation} strictly monotone on
each \(I_j\), with opposite signs at its endpoints, and keeps its modulus
larger than \(d/2\) on the complement. It therefore has exactly one zero
in each \(I_j\) and no others. The implicit-function estimate gives
\(r_{j,\epsilon}=r_j+O(\epsilon^\delta)\).

At the perturbed root, Eqs.~\eqref{eq:app-corner-normalized-derivative} and
\eqref{eq:app-corner-general-pressure-derivative} give
\begin{align}
 \widehat\vartheta_\epsilon'(r_{j,\epsilon})&=0,
 \nonumber\\
 \widehat\vartheta_\epsilon''(r_{j,\epsilon})
 &=r_j^{\iota-1}R'(r_j)+O(\epsilon^\delta),
 \nonumber\\
 \widehat\pi_\epsilon'(r_{j,\epsilon})
 &=mr_j^{m-1}+O(\epsilon^\delta).
 \label{eq:app-corner-root-derivatives}
\end{align}
The second derivative under elimination is
\(\widehat\vartheta_\epsilon''/(\widehat\pi_\epsilon')^2\) at a turn, which
is precisely Eq.~\eqref{eq:corner-general-turn-curvature}. If \(R\) is
nonconstant, any positive multiple root satisfies \(R(r)=R'(r)=0\) and
therefore
\begin{equation}
 \operatorname{Res}(R,R')=0.
 \label{eq:app-corner-resultant-necessary}
\end{equation}
The resultant can also vanish at a nonreal or nonpositive multiple root. The resultant zero set is a
proper algebraic discriminant in the coefficient space of polynomials of
degree at most \(\kappa_N\), and the positive-multiple-root locus is a
semialgebraic subset of positive codimension, generically codimension one
along its smooth stratum. The positive-multiple-root locus is unstable under
\(C^1\) remainders, whereas the simple-root condition is open. This proves the
compact simple-root stability statement.
\end{proof}

\begin{proof}[Proof of Corollary~\ref{thm:maxwell-corner-universality}]
For \(R(r)=1-r^k\), Eq.~\eqref{eq:app-corner-integrable-primitive} gives
\begin{equation}
 \Phi_{\iota,R}(r;0)
 =\frac{r^\iota}{\iota}
  -\frac{r^{\iota+k}}{\iota+k}.
 \label{eq:app-corner-binomial-primitive}
\end{equation}
Multiplication by \(\iota(\iota+k)/k\) gives
\(H_{\iota,k}\) in Eq.~\eqref{eq:corner-universality-profile}. Its
derivative is
\begin{equation}
 H'_{\iota,k}(r)
 =\frac{\iota(\iota+k)}{k}r^{\iota-1}(1-r^k),
 \label{eq:app-corner-profile-derivative}
\end{equation}
so \(r=1\) is its only positive turn and
\begin{equation}
 H''_{\iota,k}(1)=-\iota(\iota+k).
 \label{eq:app-corner-binomial-second-derivative}
\end{equation}
The simple-root part of
Theorem~\ref{thm:maxwell-corner-Newton-classification} gives the unique perturbed
turn and its \(O(\epsilon^\delta)\) displacement. Since
\((r^m)'|_{r=1}=m\), elimination gives the curvature in
Eq.~\eqref{eq:corner-universality-turn}. The overall sign in
Eq.~\eqref{eq:corner-universality-normalization} is opposite to that of
the physical temperature, so the normalized maximum is a physical local
minimum.
\end{proof}

\begin{proof}[Proof of Proposition~\ref{prop:corner-Newton-realizability}]
Fix positive constants \(c_p,c_s,c_v,\rho_*,\lambda\), put
\(\iota=m-n\), \(\sigma=a+b-c\), and define on the punctured resolved
half-plane
\begin{align}
 \pi_M(\epsilon,r)
 &=\pi_t(\epsilon)
   +c_p\rho_*^m\epsilon^a r^m,
 \label{eq:app-corner-realizable-pressure}\\
 \vartheta_M(\epsilon,r)
 &=\vartheta_{\mathrm{ref}}(\epsilon)
   -\frac{\lambda m c_pc_v}{c_s}
    \rho_*^\iota\epsilon^\sigma
    \Phi_{\iota,R}(r;r_0).
 \label{eq:app-corner-realizable-temperature}
\end{align}
For \(\iota>0\), take \(r_0=0\) and
\(\vartheta_{\mathrm{ref}}=\vartheta_t\); for \(\iota\leq0\), take
\(r_0=1\). Since \(\pi_M\) is strictly increasing, it can be inverted on
\(r>0\). Denote the inverse by \(r=r(\pi,\epsilon)\), and set
\begin{equation}
 D(\vartheta,\pi;\epsilon)
 =-c_s\epsilon^{-b}
  \left[\vartheta-
  \vartheta_M(\epsilon,r(\pi,\epsilon))\right].
 \label{eq:app-corner-realizable-gap}
\end{equation}
Then \(D=0\) is the prescribed Maxwell component and
\begin{equation}
 D_{\vartheta}=-c_s\epsilon^{-b},
 \qquad
 D_{\pi}=-\lambda c_v\epsilon^{-c}
 (\rho_*r)^{-n}R(r).
 \label{eq:app-corner-realizable-jumps}
\end{equation}
Thus Eq.~\eqref{eq:corner-gap-envelope} realizes exactly the entropy and
volume jumps in Eqs.~\eqref{eq:corner-universality-entropy} and
\eqref{eq:corner-universality-volume}.

Choose an arbitrary smooth common value \(h\) and set
\begin{equation}
 g_S=h-\frac12D,
 \qquad
 g_L=h+\frac12D,
 \qquad
 \mathcal F_i(y_i)=g_i+\frac12y_i^2.
 \label{eq:app-corner-realizable-wells}
\end{equation}
The two local potentials have strict nondegenerate minima and value gap
\(D\). They define the required local two-sheet profile; on any punctured
compact control set they may also be joined through an arbitrarily high
smooth barrier without altering either selected critical value. This
optional joining is only a local thermodynamic completion and imposes no
gravitational equations. On the two disjoint critical sheets prescribe
\begin{equation}
 b_S=\upsilon_S\rho,
 \qquad b_L=\upsilon_L,
 \label{eq:app-corner-realizable-sheet-margins}
\end{equation}
with the positive \(C^2\) units specified in the proposition. Smooth
off-sheet extensions exist separately in the disjoint well neighborhoods,
but no choice of such an extension is part of the realized incidence class. Thus the
\(S\) sheet reaches the face transversely at \(r=0\), the \(L\) sheet stays
uniformly in the interior, and \(r>0\) is the physical side. The primitive in
Eq.~\eqref{eq:app-corner-realizable-temperature} is finite at that face
exactly when \(\iota>0\), which proves the last assertion.
\end{proof}

\subsection{Fold--value resonance and completeness data}
\label{appsubsec:fold-value-resonance}

In the fold normal form \eqref{eq:fold-normal-form-app}, the newborn values
are
\begin{align}
 v_+(t)
 &=\mathcal F_0(t,\eta)-\frac{2\sigma}{3}t^{3/2},
 \label{eq:fold-value-plus}\\
 v_-(t)
 &=\mathcal F_0(t,\eta)+\frac{2\sigma}{3}t^{3/2}.
 \label{eq:fold-value-minus}
\end{align}
Suppose a remote nondegenerate sheet has value
\(v_r(t,\delta)\), and use
\(\delta=v_r-\mathcal F_0\) as the second transverse control. Its equality
walls with the newborn pair have leading equations
\begin{equation}
 \delta=\mp\frac{2\sigma}{3}t^{3/2}+o(t^{3/2}).
 \label{eq:fold-remote-value-walls}
\end{equation}
These walls determine whether the newborn pair enters above or below the
remote value. A fold operator without this incidence information cannot
reconstruct \(\prec_{\mathcal F}\) across a homotopy that passes the
resonance.

The complete codimension-two incidence data for a selected control window must also
include every simultaneous event present there. Product strata, such as
two independent folds or two independent face crossings, carry commuting
relations. A stationary point meeting two physical faces carries the two
face labels and the corresponding corner incidence. A fold coincident with
an unrelated value or face event is a product only when their local branch
supports are disjoint. These cases explain why the short list of braid,
Maxwell--admissibility, cusp, and fold--boundary relations cannot serve as
a universal classification of all multi-control families.

\subsection{Proof of conditional reconstruction}
\label{appsubsec:reconstruction-proof}

\begin{proof}[Proof of Theorem~\ref{thm:path-reconstruction}]
Let \(\gamma\colon[0,1]\to U\) be a generic path beginning in \(C_0\).
Its inverse image of the regular codimension-one strata is a compact
zero-dimensional manifold and is therefore finite. It meets those strata
transversely. Between crossings, Theorem~\ref{thm:chamber-constancy} propagates
the chamber data. At a crossing, its orientation and incidence label
select the fold, value, or face operator. Their ordered composition
therefore reconstructs the chamber data on every subsequent path segment.

Now let
\begin{equation}
 H\colon[0,1]^2\longrightarrow U
 \label{eq:path-homotopy}
\end{equation}
be a generic homotopy relative to the endpoints between paths
\(\gamma_0\) and \(\gamma_1\). A generic two-dimensional map avoids wall
strata of codimension at least three. Its isolated nongeneric path movies
are of two kinds.
\begin{enumerate}
\item A path becomes tangent to a regular wall. On one side of the movie,
two consecutive crossings of opposite orientation are present; on the
other, neither is present. Their operators cancel by
\eqref{eq:operator-inverses}.
\item A path passes through a codimension-two wall stratum. The crossing
word on one side of the movie is replaced by the word on the other. The
recorded local continuation relation identifies the two composites.
\end{enumerate}
Local finiteness of the Whitney stratification and compactness of the
homotopy square reduce the comparison to finitely many such movies.
Subdivision of the square and composition of the local identities prove
that \(\gamma_0\) and \(\gamma_1\) induce the same endpoint chamber data. This is
the standard stratified continuation mechanism
\cite{Cerf1970Stratification,GoreskyMacPherson1988Stratified}.

If the sheet covering has nontrivial holonomy, continuation around a loop
can permute genealogical labels while preserving the unlabeled decorated
portrait. This is the stated monodromy action. In the global interval
setting, Lemma~\ref{lem:finite-critical-covering} removes that ambiguity.
A nonhomotopic path need not bound a homotopy square, so the argument gives
no equality for such paths.
\end{proof}

\section{Value-order obstructions and boundary selection}
\label{app:value-order-proofs}

\subsection{The corrected two-extrema criterion}
\label{appsubsec:two-extrema-proof}

\begin{proof}[Proof of Proposition~\ref{prop:corrected-two-extrema}]
Along a stationary sheet the envelope identity gives
\begin{equation}
 \frac{\mathrm d}{\mathrm dT}
 \mathcal F_T(x_i(T);\eta)
 =\left.\frac{\partial\mathcal F_T}{\partial T}\right|_{x=x_i(T)}
 =-S(x_i(T);\eta),
 \label{eq:stationary-envelope-identity}
\end{equation}
because the term proportional to \(\partial_x\mathcal F_T\) vanishes.
Subtracting the identities for \(x_L\) and \(x_S\) proves
Eq.~\eqref{eq:two-extrema-gap-monotonicity}; its sign is strict since
\(S_x>0\) and \(x_L>x_S\). Hence \(D\) is strictly decreasing and can
have at most one zero. Continuity on \(I\), together with the one-sided
limits at the ordinary folds, proves the equivalence with
Eq.~\eqref{eq:two-extrema-endpoint-signs} by the intermediate-value theorem.
At its zero,
\begin{equation}
 D'(T_M)=-\Delta S_{LS}\neq0,
 \label{eq:two-extrema-transverse-gap}
\end{equation}
so the crossing is transverse and has nonzero latent heat. Maxwell equality
compares only these two critical values. Admissibility and the comparison
with every other active local minimum, including every competing phase
admitted by the ensemble, are logically independent conditions and are
therefore necessary for a physical global transition.
\end{proof}

\subsection{An iso-stationary analytic family}
\label{appsubsec:iso-stationary-proof}

\begin{proof}[Proof of Theorem~\ref{thm:local-data-insufficient}]
Fix \(T_0\in\mathbb R\), take the ambient collar \(X=(-2,2)\), and use the
fixed physical set \(\mathcal A=(-3/2,3/2)\), with closure
\(K=[-3/2,3/2]\). For \(\varepsilon\) near zero, define
\begin{align}
 S_\varepsilon(x)
 &=S_0+\int_0^x e^{\varepsilon u}\,\dd u,
 \label{eq:epsilon-entropy}\\
 M_\varepsilon(x)
 &=M_0+\int_0^x e^{\varepsilon u}
       (T_0+u^3-u)\,\dd u,
 \label{eq:epsilon-mass}
\end{align}
where \(S_0\) is large enough that \(S_\varepsilon>0\) on the collar.
The functions are jointly real analytic at \(\varepsilon=0\), and
\(S_{\varepsilon,x}=e^{\varepsilon x}>0\). At reservoir temperature
\(T=T_0+t\), set
\begin{equation}
 \mathcal F_{\varepsilon,t}(x)
 =M_\varepsilon(x)-(T_0+t)S_\varepsilon(x).
 \label{eq:epsilon-potential}
\end{equation}
Then
\begin{equation}
 \partial_x\mathcal F_{\varepsilon,t}(x)
 =e^{\varepsilon x}(x^3-x-t).
 \label{eq:epsilon-gradient}
\end{equation}
The stationary equation
\begin{equation}
 x^3-x=t
 \label{eq:cubic-cover}
\end{equation}
is independent of \(\varepsilon\). Its folds are
\begin{equation}
 (x,t)=\left(\frac{1}{\sqrt3},-\frac{2}{3\sqrt3}\right),
 \qquad
 (x,t)=\left(-\frac{1}{\sqrt3},\frac{2}{3\sqrt3}\right).
 \label{eq:cubic-folds}
\end{equation}
They are ordinary because, at either fold,
\begin{equation}
 \mathcal F_{xxx}=6x e^{\varepsilon x}\neq0,
 \qquad
 \partial_t\mathcal F_x=-e^{\varepsilon x}\neq0.
 \label{eq:fold-genericity-check}
\end{equation}

For \(|t|<2/(3\sqrt3)\), denote the three roots by
\(x_-(t)<x_0(t)<x_+(t)\). At a root,
\begin{equation}
 \mathcal F_{xx}=e^{\varepsilon x}(3x^2-1).
 \label{eq:epsilon-index}
\end{equation}
The outer roots have index zero, the middle root has index one, and the
state-ordered Morse word is \((0,1,0)\) for every \(\varepsilon\). On
\(\mathcal A\), the gradient is negative at the left endpoint and positive
at the right endpoint throughout a fixed neighborhood of the fold window.
Thus, for every regular \(t\) in that window,
\begin{equation}
 \deg(\partial_x\mathcal F_{\varepsilon,t},
      (-3/2,3/2),0)=1.
 \label{eq:epsilon-degree}
\end{equation}
In the three-sheet region the same degree is the local sum
\(+1-1+1=1\). Properness is automatic on \(K\), and no stationary point
reaches its boundary.

The value difference between the outer stable sheets is
\begin{equation}
 D(\varepsilon,t)
 :=\mathcal F_{\varepsilon,t}(x_+(t))
   -\mathcal F_{\varepsilon,t}(x_-(t))
 =\int_{x_-(t)}^{x_+(t)}e^{\varepsilon x}
   (x^3-x-t)\,\dd x.
 \label{eq:D-epsilon-t}
\end{equation}
Endpoint terms vanish under differentiation because both endpoints solve
Eq.~\eqref{eq:cubic-cover}. Hence
\begin{equation}
 \partial_tD(\varepsilon,t)
 =-\int_{x_-(t)}^{x_+(t)}e^{\varepsilon x}\,\dd x<0.
 \label{eq:D-monotone}
\end{equation}
At \((\varepsilon,t)=(0,0)\), the outer roots are \(-1\) and \(1\), and
\begin{align}
 D(0,0)&=0,\nonumber\\
 \partial_tD(0,0)&=-2,
 \label{eq:D-derivatives-one}\\
 \partial_\varepsilon D(0,0)
 &=\int_{-1}^{1}x(x^3-x)\,\dd x
 =\frac25-\frac23
 =-\frac4{15}.
 \label{eq:D-derivatives-two}
\end{align}
The analytic implicit-function theorem gives a unique transverse Maxwell
curve \(t=t_M(\varepsilon)\) near the origin, with
\begin{equation}
 t_M'(0)
 =-\frac{\partial_\varepsilon D(0,0)}{\partial_tD(0,0)}
 =-\frac2{15},
 \qquad
 t_M(\varepsilon)=-\frac2{15}\varepsilon+O(\varepsilon^2).
 \label{eq:Maxwell-shift}
\end{equation}
The strict inequality \eqref{eq:D-monotone} makes the crossing unique in
any connected three-root subwindow in which it exists. At the origin the
middle stationary point is a strict maximum. The values at
\(\partial K\) exceed the outer minimum values by \(25/64\), so the outer
sheets remain the only global competitors after reducing the control
window. The shifted equality is therefore a physical Maxwell wall.

At \(t=0\), the two families with \(\varepsilon=\delta\) and
\(\varepsilon=-\delta\), for sufficiently small \(\delta>0\), select
opposite outer minima because
\(\partial_\varepsilon D(0,0)=-4/15\). Their entire stationary data remain
identical.
\end{proof}

\subsection{The clipped-cusp calculation}
\label{appsubsec:clipped-cusp-proof}

\begin{proof}[Proof of Proposition~\ref{prop:clipped-cusp}]
The stationary equation and Hessian are
\begin{equation}
 x^3-x+t=0,\qquad
 \mathcal F_{xx}=3x^2-1.
 \label{eq:clipped-cusp-stationary}
\end{equation}
For \(|t|<2/(3\sqrt3)\) there are three roots
\(x_-(t)<x_0(t)<x_+(t)\). The outer roots are minima and the middle root
is a maximum. At \(t=0\) the roots are \(-1,0,1\), and
\begin{equation}
 \mathcal F(-1;0)=\mathcal F(1;0)=-\frac14
 <\mathcal F(0;0)=0.
 \label{eq:clipped-cusp-Maxwell}
\end{equation}
The full interval contains all three points, so \(t=0\) is a physical
Maxwell crossing.

For \(|t|<1/8\),
\begin{equation}
 (x^3-x+t)\big|_{x=-3/2}=-\frac{15}{8}+t<0,\qquad
 (x^3-x+t)\big|_{x=-1/2}=\frac38+t>0.
 \label{eq:clipped-cusp-left-root}
\end{equation}
Thus \(x_-(t)\in(-3/2,-1/2)\) and is excluded by the clipped domain.
If \(t\geq0\), the middle root lies in \([0,1/\sqrt3)\); if \(t<0\),
the signs at \(-1/2\) and \(0\) place it in \((-1/2,0)\). The right
minimum remains in \((1/\sqrt3,3/2)\). Therefore the clipped family has
one active minimum and one active maximum, but no second active minimum.
The ambient cover, folds, indices, and critical values have not changed;
only the face-sign data have.
\end{proof}

\subsection{Closed-domain continuity and boundary deletion}
\label{appsubsec:closed-domain-proof}

\begin{proof}[Proof of Proposition~\ref{prop:no-jump}]
Let \(\lambda_n\to\lambda_0\). In the compact case choose
\begin{equation}
 x_n\in\operatorname*{argmin}_{x\in K(\lambda_n)}
 \mathcal F(x,\lambda_n).
 \label{eq:minimizer-sequence}
\end{equation}
Hausdorff convergence places the sequence in one compact neighborhood of
\(K(\lambda_0)\). Along a subsequence that realizes the lower limit,
\(x_n\to x_*\), and Hausdorff convergence gives
\(x_*\in K(\lambda_0)\). Joint continuity then yields
\begin{equation}
 \liminf_{n\to\infty}g_{\rm cl}(\lambda_n)
 =\mathcal F(x_*,\lambda_0)\geq g_{\rm cl}(\lambda_0).
 \label{eq:value-liminf}
\end{equation}
For a minimizer \(x_0\in K(\lambda_0)\), Hausdorff convergence supplies
\(y_n\in K(\lambda_n)\) with \(y_n\to x_0\). Hence
\begin{equation}
 \limsup_{n\to\infty}g_{\rm cl}(\lambda_n)
 \leq\lim_{n\to\infty}\mathcal F(y_n,\lambda_n)
 =g_{\rm cl}(\lambda_0).
 \label{eq:value-limsup}
\end{equation}

Under Eq.~\eqref{eq:uniform-inf-compactness}, every minimum is attained in
the common compact set \(C\). The closed-graph condition replaces the first
use of Hausdorff convergence, and lower hemicontinuity supplies a recovery
sequence for \(x_0\). The same two inequalities complete the proof. This
is the minimum version of Berge's maximum theorem
\cite{Berge1963Topological,AliprantisBorder2006Infinite}.
\end{proof}

\begin{proof}[Proof of Proposition~\ref{prop:boundary-deletion-jump}]
The derivative factorizes as
\begin{equation}
 \mathcal F_\kappa'(x)=(x^2-1)(4x-3\kappa).
 \label{eq:boundary-deletion-factorization}
\end{equation}
The stationary points are \(-1,3\kappa/4,1\). The outer points are strict
minima for \(0<\kappa<4/3\), and their values are
\begin{equation}
 \mathcal F_\kappa(-1)=-2\kappa,\qquad
 \mathcal F_\kappa(1)=2\kappa.
 \label{eq:boundary-deletion-values}
\end{equation}
For \(\lambda<-1\) both minima are active and \(-1\) wins. For
\(-1<\lambda<1\), the lower minimum is excluded and \(1\) is the only
active stationary minimum. The one-sided values of \(g_{\rm stat}\) thus
differ by \(4\kappa\). By contrast, the boundary point \(x=\lambda\) is
available in the closed problem and continuously replaces the departing
interior minimizer, so Proposition~\ref{prop:no-jump} applies.
\end{proof}

\section{C-metric identities and endpoint algebra}
\label{app:cmetric-identities}

\subsection{Reduction at fixed charge and tension}

Let \(z=eA\), \(d=1-2\mu\), and \(\Xi=1+z^2\).  North-pole
regularity and the single-string condition give
\begin{equation}
 K=\Xi+2mA,\qquad
 \mu=\frac{mA}{K}.
\end{equation}
Solving these equations yields
\begin{equation}
 mA=\frac{\mu\Xi}{d},\qquad K=\frac{\Xi}{d}.
 \label{eq:app-cmetric-mAK}
\end{equation}
Since \(Q=e/K=z/(AK)\),
\begin{equation}
 AQ=\frac{zd}{\Xi},\qquad
 \frac mQ=\frac{\mu\Xi^2}{zd^2}.
 \label{eq:app-cmetric-Am}
\end{equation}
Finally, if \(q=Q/\ell\),
\begin{equation}
 A\ell=\frac{zd}{q\Xi},\qquad
 \alpha^2=\Xi\left(1-A^2\ell^2\Xi\right)
 =\Xi-\frac{z^2d^2}{q^2}.
 \label{eq:app-cmetric-alpha}
\end{equation}

With \(u=Ar_+\), multiplication of \(f(r_+)=0\) by \(u^2\) gives
\begin{equation}
 H=(1-u^2)\left[u^2-\frac{(1-d)\Xi}{d}u+z^2\right]
 +\frac{u^4q^2\Xi^2}{z^2d^2}=0.
 \label{eq:app-cmetric-H}
\end{equation}
Since \(f=H/u^2\), its horizon derivative satisfies
\begin{equation}
 f_u\big|_{H=0}=\frac{H_u}{u^2}.
\end{equation}
Equations~\eqref{eq:app-cmetric-mAK}--\eqref{eq:app-cmetric-H} then reduce
the thermodynamic functions to Eqs.~\eqref{eq:cmetric-reduced-state-functions}.

\subsection{Restricted first-law identity}

On a simple horizon sheet, set
\begin{equation}
 \mathscr D_z=\partial_z-\frac{H_z}{H_u}\partial_u.
\end{equation}
Using
\begin{equation}
 \mathfrak m=\frac{\mu\alpha}{zd},\qquad
 s=\frac{u^2\Xi}{z^2d(1-u^2)},\qquad
 \tau=\frac{zdH_u}{\Xi u^2\alpha},
\end{equation}
direct differentiation and reduction modulo \(H=0\) give
\begin{equation}
 4\mathscr D_z\mathfrak m-\tau\mathscr D_zs=0.
 \label{eq:app-cmetric-first-law}
\end{equation}
This is also the fixed-\((P,Q,\mu)\) restriction of the
full-cohomogeneity first law.  Equation~\eqref{eq:app-cmetric-first-law}
couples the implicit variation of the outer horizon to the normalized mass
on the fixed-control thermodynamic sheet.

\subsection{Global entropy regularity}

The entropy coordinate is regular throughout the positive-temperature
physical sector.  Exact differentiation before imposing the horizon equation
gives
\begin{align}
 H_us_z-H_zs_u={}&
 -\frac{2(1-d)u^2\Xi}{z^3d^2}
 \nonumber\\
 &-\frac{4u(\Xi-2u^2)}{dz^3(1-u^2)^2}H.
 \label{eq:app-cmetric-entropy-offshell}
\end{align}
On \(H=0\), using \(1-d=2\mu\), this becomes
\begin{equation}
 H_us_z-H_zs_u
 =-\frac{4\mu u^2\Xi}{z^3d^2}.
 \label{eq:app-cmetric-entropy-wedge}
\end{equation}
For a simple positive-temperature horizon,
\(H_u=\tau\Xi u^2\alpha/(zd)\), and hence
\begin{equation}
 \mathscr D_zs
 =\frac{H_us_z-H_zs_u}{H_u}
 =-\frac{4\mu}{z^2d\alpha\tau}<0.
 \label{eq:app-cmetric-entropy-regularity}
\end{equation}
Equation~\eqref{eq:app-cmetric-entropy-wedge} also shows
\(dH\wedge ds\ne0\) when a particular graph coordinate fails at extremality;
entropy itself remains regular.  Therefore no entropy-coordinate zero is
missed by the Euclidean-equilibrium enumeration in the open sector
\(0<\mu<1/4\), \(\alpha>0\), and \(\tau>0\).  At equilibrium,
\begin{equation}
 \mathscr D_z^2\!\left(\frac{\mathcal G_T}{Q}\right)
 =\frac14(\mathscr D_z\tau)(\mathscr D_zs),
\end{equation}
so a canonical minimum is equivalently characterized by
\(\mathscr D_z\tau<0\).

\subsection{Slow-acceleration face}

At conformal infinity the relevant polynomial is
\begin{equation}
 C(v)=(v^2-1)(1+2mAv+z^2v^2)+\frac1{(A\ell)^2}.
\end{equation}
The generic boundary-horizon face obeys \(C=C_v=0\).  The derivative equation
is linear in \(mA\) and gives
\begin{equation}
 mA=\frac{v(1+2z^2v^2-z^2)}{1-3v^2}.
\end{equation}
Substitution into \(C=0\) gives
\begin{equation}
 \frac1{(A\ell)^2}
 =\frac{(1-v^2)^2(1-z^2v^2)}{1-3v^2},
\end{equation}
which proves Eq.~\eqref{eq:cmetric-slow-face}.  Squaring the latter equation
without retaining its sign and interval conditions produces extraneous
algebraic components; the unsquared expression is used in the physical
filter.

\subsection{The X point}

At \(X\), the equations \(\alpha=0\) and bulk extremality give
\begin{equation}
 A\ell=\frac1{\sqrt{1+z^2}},\qquad
 mA=z\sqrt{1+z^2}
\end{equation}
\cite{AbbasvandiEtAl2019Snapping}.  Combining them with
Eq.~\eqref{eq:app-cmetric-mAK} gives
\begin{equation}
 \frac{\mu}{d}=\frac{z}{\sqrt{1+z^2}}.
\end{equation}
Combining the first equation with \(A\ell=zd/(q\Xi)\) then gives
\begin{equation}
 q=\frac{zd}{\sqrt{1+z^2}}=\mu,
\end{equation}
and hence Eq.~\eqref{eq:cmetric-Pt}.

\subsection{The ordinary endpoint of the turning curve}

Let \(h=\sqrt{2\sqrt3-3}\).  The extension field
\begin{equation}
 \mathbb K=\mathbb Q(\sqrt3,h),\qquad h^2=2\sqrt3-3
\end{equation}
has basis \(1,\sqrt3,h,\sqrt3h\).  Define
\begin{align}
 d_c&=\frac1{1+h},
 &\mu_c&=\frac{h}{2(1+h)},
 &z_c&=2-\sqrt3,\nonumber\\
 q_c&=\frac1{(1+h)(1+\sqrt3)},
 &u_c&=\frac{z_c}{h}=\frac{\sqrt3h}{3}.
 \label{eq:app-cmetric-number-field-point}
\end{align}
At this point
\begin{equation}
 \Xi_c=4z_c,\qquad
 (A\ell)_c=\frac{1+\sqrt3}{4},\qquad
 \alpha_c=\sqrt3-1.
\end{equation}

Introduce a formal parameter \(\epsilon\), set \(z=z_c+\epsilon\), and solve
\(H(u,z;q_c,d_c)=0\) in \(\mathbb K[[\epsilon]]\).  The outer-horizon series
begins
\begin{equation}
 u=u_c-\frac{(71+41\sqrt3)h}{24}\epsilon^2
 -\frac{(97+56\sqrt3)h}{12}\epsilon^3+O(\epsilon^4).
 \label{eq:app-cmetric-u-series}
\end{equation}
In particular, \(\mathscr D_zu=0\).  Substitution into the exact state
functions gives
\begin{align}
 \tau&=4\mu_c+\tau_3\epsilon^3+O(\epsilon^4),\nonumber\\
 s&=\frac1{q_c}+s_1\epsilon+O(\epsilon^2),\nonumber\\
 \nu&=\nu_c+\nu_2\epsilon^2+O(\epsilon^3),
 \label{eq:app-cmetric-state-series}
\end{align}
where
\begingroup\small
\begin{align}
 \tau_3&=\frac{3(989+571\sqrt3)-(6393+3691\sqrt3)h}{48}
       =-57.8429\ldots,\nonumber\\
 s_1&=-\frac{(19+11\sqrt3)(1+h)}2
       =-31.9879\ldots,\nonumber\\
 \nu_2&=\frac{1668+963\sqrt3+(2631+1519\sqrt3)h}{6}
       =1153.45\ldots.
 \label{eq:app-cmetric-series-coefficients}
\end{align}
\endgroup
Thus
\begin{align}
 \mathscr D_z\tau=\mathscr D_z^2\tau
 &=\mathscr D_z\nu=0,\nonumber\\
 \mathscr D_z^3\tau&=6\tau_3\ne0,\nonumber\\
 \mathscr D_zs&=s_1\ne0.
\end{align}
The same substitution gives
\begin{equation}
 \frac{G_c}{Q}=0,\qquad
 \tau_c=4\mu_c,\qquad s_c=\frac1{q_c}.
\end{equation}

Since \(h^4+6h^2-3=0\) and \(h=2\mu_c/(1-2\mu_c)\),
\begin{equation}
 (1-2\mu_c)^4(h^4+6h^2-3)
 =64\mu_c^4-48\mu_c^2+24\mu_c-3=0.
\end{equation}
The monotonicity argument in Sec.~\ref{subsec:cmetric-exact-turning}
then isolates a unique root in \(0<\mu<1/4\).

The formal series identifies the ordinary endpoint of the finite-separation
curve in Eq.~\eqref{eq:cmetric-turning-polynomial}.  There the two roots
coalesce, and the vanishing quadratic discriminant gives the same endpoint
and upper tension in the global coordinate on the turning locus.

\subsection{Bicritical approach identities}

Let
\begin{align}
 \mu&=\frac{x}{\mathcal D_x},
 &\mathcal D_x&=1+2x-x^2,\nonumber\\
 d&=\frac{1-x^2}{\mathcal D_x},
 &q&=\mu,\nonumber\\
 0&<x<\sqrt2-1.
 \label{eq:app-bicritical-parameter}
\end{align}
For the two values of \(z\) in
Eq.~\eqref{eq:cmetric-bicritical-z}, direct substitution into the horizon
polynomial gives \(u=x\).  On the large branch,
\begin{equation}
 \alpha_L^2=\frac{(1-x^2)^3}{1-2x^2},
 \qquad
 \mathfrak m_L=\frac{\sqrt{1-x^2}}{x},
 \qquad
 s_L=\frac{\mathcal D_x}{x^2}.
 \label{eq:app-bicritical-large}
\end{equation}
On the one-sided limit of the disappearing interior sheet, \(\alpha_X=0\),
\(\mathfrak m_X=0\), and \(s_X=\mathcal D_x\).  Along the Maxwell approach,
the selected ratio \(H_u/\alpha\) gives
\begin{equation}
 \tau_t=\frac{4\mu}{\sqrt{1-x^2}}.
\end{equation}
It follows without further elimination that
\begin{align}
 \mathfrak g_L
 &=\frac{\sqrt{1-x^2}}{x}
 -\frac{1}{x\sqrt{1-x^2}}
 =-\frac{x}{\sqrt{1-x^2}},\nonumber\\
 \mathfrak g_X
 &=-\frac{\tau_t\mathcal D_x}{4}
 =-\frac{x}{\sqrt{1-x^2}}.
 \label{eq:app-bicritical-free-energies}
\end{align}
Finally, on the one-sided fixed-pressure cut \(P=P_t\), define
\(G_X^{\rm lim}(T;P_t)=-TS_X\).  The large phase is regular, so the envelope
identity gives
\(\mathrm d(G_L-G_X^{\rm lim})=-(S_L-S_X)\mathrm dT\).  With
\(\delta\tau=4\pi Q(T_t-T)=\tau_t-\tau\), this gives
\begin{equation}
 \frac{G_L-G_X^{\rm lim}}{Q}
 =\frac{\mathcal D_x(1-x^2)}{4x^2}\delta\tau+O(\delta\tau^2),
\end{equation}
which proves Eq.~\eqref{eq:cmetric-exact-jump-closure} and fixes its sign on
the fixed-pressure cut.

\section{Exact Maxwell turning and the surrounding wall atlas}
\label{app:cmetric-turning-certificate}

This appendix separates the exact proof for the finite-separation
Maxwell curve from the numerical continuation used to display the surrounding
wall complex.  All equalities and sign changes in
Secs.~\ref{subsec:cmetric-exact-turning} and
\ref{subsec:cmetric-bicritical-closure} follow from exact polynomial
arithmetic.  The three-dimensional point clouds give a numerical atlas over
the control range stated below.
We use resultants, subresultants, saturation, and Gr{\"o}bner bases in the
standard computational-algebra sense of
Ref.~\cite{CoxLittleOShea2025}; Sturm chains and exact real-root isolation
follow Ref.~\cite{BasuPollackRoy2006}.

\begin{theorem}[Exact finite-separation Maxwell turning]
\label{thm:app-cmetric-exact-turning}
In the standard single-string fixed-\((P,Q,\mu)\) ensemble, there is a
connected one-parameter locus of strict finite-separation physical Maxwell
turning points, with one turning point for each \(0<\mu<\mu_+\), where
\(\mu_+\) is given by Eq.~\eqref{eq:cmetric-mu-plus-strong}.  The two
equal-value states are admissible and entropy-regular, and each is a strict
canonical minimum.  No physical equilibrium has a lower potential.  Every
finite-separation physical Maxwell turning belongs to this locus.  The two
phases coalesce at \(\mu=\mu_+\); toward the other endpoint the locus extends
to \(\mu=0\).
\end{theorem}

\subsection{Polynomial proof for the coexisting pair}

Put \(w=z^2\) and use Eq.~\eqref{eq:cmetric-chi-mu}.  After clearing the
positive denominators in the horizon equation, define
\begin{align}
 \widehat H_\chi(u,w)={}&
 w(1+\chi^2)(1-u^2)(u^2+w)
 -2\chi w(1+w)u(1-u^2)\nonumber\\
 &+\chi^2(1+w)^2u^4.
 \label{eq:app-turning-cleared-horizon}
\end{align}
At \(u=\chi\), one finds identically
\begin{equation}
 \widehat H_\chi(\chi,w)=F_\chi(w),
 \label{eq:app-turning-horizon-F}
\end{equation}
with \(F_\chi\) defined in
Eq.~\eqref{eq:cmetric-turning-polynomial}.  Polynomial division of the
numerators of \(\tau^2-16\mu^2\), \(\alpha\mathfrak g\), and
\(\nu-\nu_0\) by \(F_\chi\) has zero remainder, where
\begin{equation}
 \nu_0=\frac{(1-\chi)(1+\chi)^5}
 {\chi^3(1+\chi^2)^2}.
\end{equation}
The positive root of the normalization factor is
\begin{equation}
 \alpha=
 \frac{\sqrt w(1-\chi^2)(\chi^2-w)}
 {(1+w)\chi^4}.
 \label{eq:app-turning-alpha}
\end{equation}
The inequalities below fix the positive-temperature sign of the squared
identity.  More explicitly, reduction modulo \(F_\chi\) gives
\begin{align}
 \tau&=4\mu,
 &\mathfrak m&=
 \frac{(1-\chi^2)(\chi^2-w)}
 {\chi^3(1+\chi^2)(1+w)},\nonumber\\
 s&=\frac{\chi^2(1+\chi)^2(1+w)}
 {w(1+\chi^2)(1-\chi^2)},
 &\nu&=\nu_0,
 \label{eq:app-turning-state-values}
\end{align}
and the critical free energy factors as
\begin{equation}
 \mathfrak g=
 \frac{F_\chi(w)}
 {\chi^3w(\chi-1)(\chi+1)(1+\chi^2)(1+w)}.
 \label{eq:app-turning-G-factor}
\end{equation}
Thus both roots have \(\tau=4\mu\), \(\mathfrak g=0\), and \(\nu=\nu_0\).

Write \(t=\chi^2\).  In the interval
Eq.~\eqref{eq:cmetric-chi-window},
\begin{equation}
 A_\chi=(1-t)^2+t^3>0,
 \qquad
 B_\chi=t(t^2+2t-1)<0.
\end{equation}
The product of the roots is \(\chi^6/A_\chi>0\), and their sum is positive.
Moreover,
\begin{equation}
 w_-+w_+-\chi^2
 =-\frac{\chi^6(\chi^2+2)}{A_\chi}<0.
\end{equation}
It follows that \(0<w_-<w_+<\chi^2\).  This proves
Eq.~\eqref{eq:app-turning-alpha} is real and positive and excludes the
normalization face.

The same inequalities prove that \(u=\chi\) is the outer horizon.  On
\(F_\chi=0\), division by \(u-\chi\) gives
\begin{equation}
 \widehat H_\chi(\chi+h,w)=h\sum_{k=0}^{3}c_kh^k.
 \label{eq:app-outer-horizon-quotient}
\end{equation}
In descending order, the coefficients are
\begin{align}
 c_3={}&(\chi^2-w)+\chi^2w(1+w),\nonumber\\
 c_2={}&2\chi\left[2\chi^2(1+w+w^2)-w(1-w)\right],\nonumber\\
 c_1={}&6\chi^4(1+w+w^2)+5\chi^2w^2+\chi^2w+w(1-w),\nonumber\\
 c_0={}&4\chi\left[\chi^4(1+w+w^2)+\chi^2w^2+w(\chi^2-w)\right].
 \label{eq:app-outer-horizon-coefficients}
\end{align}
The first, third, and fourth lines are manifestly positive.  In the second,
\(w<\chi^2\) implies
\begin{equation}
 2\chi^2(1+w+w^2)-w(1-w)
 >\chi^2(1+3w+2w^2)>0.
\end{equation}
Thus \(\widehat H_\chi(\chi+h,w)>0\) for every \(h>0\), and no larger
positive horizon exists.

The angular quadratic has its vertex at
\(v_*=-mA/w<-1\), because
\begin{equation}
 mA-w=
 \frac{(\chi-w)+\chi w(1-\chi)}{1+\chi^2}>0.
\end{equation}
Its minimum on \([-1,1]\) is therefore the positive value in
Eq.~\eqref{eq:cmetric-turning-margins}.  A contact with the slow-acceleration
face would require \(C=C_v=0\).  Eliminating \(v\) and then \(w\) gives,
apart from factors nonzero in the open interval,
\begin{align}
 R_{\rm slow}(\chi)={}&
 27\chi^{16}-108\chi^{14}+216\chi^{12}-108\chi^{10}
 \nonumber\\
 &+774\chi^8-612\chi^6+60\chi^2+7.
 \label{eq:app-slow-resultant}
\end{align}
More precisely, the resultant is
\begin{equation}
 2^{12}\chi^{56}(1-\chi^2)^8(1+\chi^2)^{10}
 R_{\rm slow}(\chi).
\end{equation}
In \(U=\chi^2\), the Sturm chain has the sign pattern
\begin{equation}
 (+,+,-,-,+,+,-,-,+)
\end{equation}
at both \(U=0\) and \(U=1/6\).  Its variation count is four at both
endpoints, so \(R_{\rm slow}\) has no zero on \(0<U<1/6\), which contains
the complete interval~\eqref{eq:cmetric-chi-window}.  A rational interior
point has positive slow margin on both root sheets; no slow-boundary contact
can occur before the critical merger.

\subsection{Exact nondegeneracy of the turning}

The condition \(\Delta V=0\) determines a stationary point of the
coexistence curve.  Its curvature can also be fixed algebraically.  Set
\(p=q^2\) and \(\rho=\alpha/\sqrt w\).  On one equilibrium phase, use the
three rational constraints
\begin{equation}
 H=0,
 \qquad
 \rho^2-\frac{1+w}{w}+\frac{d^2}{p}=0,
 \qquad
 \tau=\text{constant}
 \label{eq:app-turning-isothermal-constraints}
\end{equation}
to differentiate \((u,w,\rho)\) with respect to \(p\).  At \(u=\chi\), the
Jacobian determinant reduces modulo \(F_\chi\) to
\begin{equation}
 \mathcal J_\chi(w)=
 -\frac{4(\chi-1)(3\chi^4+6\chi^2-1)(A_\chi w+B_\chi)}
 {\chi^6(\chi+1)(1+\chi^2)}.
 \label{eq:app-turning-response-Jacobian}
\end{equation}
If \(w'\) is the other root of \(F_\chi\), Vieta's relation gives
\(A_\chi w+B_\chi=-A_\chi w'\ne0\).  All remaining factors in
Eq.~\eqref{eq:app-turning-response-Jacobian} are nonzero in the open
interval, so the isothermal response is regular on both phases.

The quotient-ring calculation gives
\begin{align}
 K_\chi(w):={}&\left(\partial_p\nu\right)_{\tau,\mu}\nonumber\\
 ={}&\frac{(1+\chi)^8}{2\chi^9(1+\chi^2)^4}
 \bigl[\chi^8w+5\chi^8-5\chi^6\nonumber\\
 &\hspace{16mm}-3\chi^4w-8\chi^4+3\chi^2w-w\bigr].
 \label{eq:app-turning-isothermal-volume}
\end{align}
The coefficient of \(w\) factors as
\begin{equation}
 \mathcal C_V(\chi)=
 \frac{(\chi-1)(\chi+1)^9A_\chi}
 {2\chi^9(1+\chi^2)^4}<0.
 \label{eq:app-turning-kappa}
\end{equation}
Therefore
\(K_\chi(w_+)-K_\chi(w_-)=\mathcal C_V(\chi)(w_+-w_-)<0\).
Equation~\eqref{eq:app-turning-state-values} gives the exact entropy
difference
\begin{equation}
 s_+-s_-=-\frac{A_\chi(1+\chi)^2(w_+-w_-)}
 {\chi^4(1+\chi^2)(1-\chi^2)}<0.
 \label{eq:app-turning-entropy-difference}
\end{equation}
In terms of \(q\), the dimensionless Clapeyron identity is
\begin{equation}
 \frac{\mathrm d\tau}{\mathrm dq}
 =4q\frac{\nu_+-\nu_-}{s_+-s_-}.
 \label{eq:app-turning-clapeyron-q}
\end{equation}
Differentiating at \(\nu_+=\nu_-\), inserting
\(q^2=\chi^2(1+\chi^2)/(1+\chi)^4\), and using
Eqs.~\eqref{eq:app-turning-kappa} and
\eqref{eq:app-turning-entropy-difference} yields
\begin{equation}
 \left.\frac{\mathrm d^2\tau}{\mathrm dq^2}\right|_{\rm turn}
 =\frac{4(1-\chi)^2(1+\chi)^4}
 {\chi^3(1+\chi^2)^2}>0.
 \label{eq:app-turning-positive-curvature}
\end{equation}
Thus every point of the exact locus is a nondegenerate local minimum of the
Maxwell temperature as a function of pressure.

\subsection{Canonical stability}

On the explicit horizon surface, let \(\widehat q^2(u,w)\) be obtained by
solving \(H=0\), and let \(\widehat\tau^2\) and \(s\) be the corresponding
state functions.  Define
\begin{equation}
 J_T=\det\frac{\partial(\widehat q^2,\widehat\tau^2)}
 {\partial(u,w)},
 \qquad
 J_S=\det\frac{\partial(\widehat q^2,s)}{\partial(u,w)}.
 \label{eq:app-turning-J-definitions}
\end{equation}
Their product has the sign of the fixed-\(q\) canonical Hessian.  At
\(u=\chi\), reduction modulo \(F_\chi\) gives
\begin{equation}
 J_S=\frac{2}{\chi w(1+\chi)^2(1+w)}>0.
 \label{eq:app-turning-JS}
\end{equation}
On a simple horizon sheet, the tangent to a fixed-\(q\) level is annihilated
by \(d\widehat q^2\).  Equation~\eqref{eq:app-turning-JS} therefore implies
that the entropy derivative along this tangent is nonzero.  Both exact phases
lie in entropy-regular coordinate patches.
Let \(\widetilde N_T\) and \(\widetilde D_T\) be the coprime numerator and
denominator of \(J_T|_{u=\chi}\).  Direct cancellation gives
\begin{align}
 \widetilde D_T(\chi,w)={}&\chi^9(1+\chi)^3(1+w)^6
 \mathcal Q_T(\chi,w)^2,
 \label{eq:app-stability-denominator}\\
 \mathcal Q_T(\chi,w)={}&
 w(\chi^6+2\chi^4-2\chi^2+1)
 +\chi^2(3\chi^2-1).
 \label{eq:app-stability-denominator-factor}
\end{align}
The factors outside \(\mathcal Q_T^2\) are positive in the physical
domain.  The remaining factor cannot vanish on either turning sheet because
\begin{align}
 \operatorname{Res}_w(F_\chi,\mathcal Q_T)
 ={}&\chi^8(1-\chi)^2(1+\chi)^2
 \nonumber\\
 &\times(1+\chi^2)^3>0
 \qquad (0<\chi<\chi_c).
 \label{eq:app-stability-denominator-resultant}
\end{align}
Thus \(\widetilde D_T>0\) at both roots of \(F_\chi\); in particular,
transport of the Hessian sign cannot be obstructed by a pole.  The integer
content of \(\widetilde N_T\) is \(16\).  If
\(N_T=\widetilde N_T/16\) is chosen with positive leading coefficient,
exact elimination gives
\begin{align}
 \operatorname{Res}_w(F_\chi,N_T)={}&
 \chi^{30}(\chi-1)^{14}(\chi+1)^4(1+\chi^2)^{11}
 \nonumber\\
 &\times(3\chi^4+6\chi^2-1)^2.
 \label{eq:app-stability-resultant-primitive}
\end{align}
Equivalently,
\(\operatorname{Res}_w(F_\chi,\widetilde N_T)\) is the right-hand side of
Eq.~\eqref{eq:app-stability-resultant-primitive} multiplied by \(2^8\).
Both versions are nonzero on the open interval
Eq.~\eqref{eq:cmetric-chi-window}.  At \(\chi=1/4\), exact root isolation
gives
\begin{equation}
 \frac1{250}<w_-<\frac1{200},
 \qquad
 \frac{57}{1000}<w_+<\frac{29}{500},
\end{equation}
and rational interval evaluation gives \(J_T>0\) on both boxes.  The two
simple root sheets are connected and cannot change sign without a zero of
Eq.~\eqref{eq:app-stability-resultant-primitive}; their denominator is
strictly positive by Eq.~\eqref{eq:app-stability-denominator-resultant}.
Consequently \(J_T>0\) on both sheets throughout
\(0<\chi<\chi_c\).  Together with Eq.~\eqref{eq:app-turning-JS}, this proves
that both coexisting states are strict canonical minima.

\subsection{Stationary eliminant at the exact controls}

It remains to exclude a physical Euclidean equilibrium below the common
value \(G=0\).  At the exact controls, retain the polynomial normalization
\begin{equation}
 \widehat E_\chi(u,w)
 =\bigl(\partial_u\widehat H_\chi(u,w)\bigr)^2
 -16w(1+w)^2u^4(\chi^2-w).
 \label{eq:app-stationary-generator-normalization}
\end{equation}
This is the cleared numerator of \(\tau^2-16\mu^2\) with its integer and
parameter content fixed.  Eliminating \(w\) gives the exact identity
\begin{align}
 \operatorname{Res}_w(\widehat H_\chi,\widehat E_\chi)
 ={}&16\chi^4u^8(u-\chi)^2(1+\chi^2)^2
 (u-1)^4(u+1)^4\mathcal R_\chi(u),
 \label{eq:app-stationary-resultant}
\end{align}
where
\begin{align}
 \mathcal R_\chi(u)={}&
 21\chi^4u^6+20\chi^2u^6
 +18\chi^5u^5+64\chi^3u^5\nonumber\\
 &+40\chi u^5+9\chi^6u^4+36\chi^4u^4
 +38\chi^2u^4\nonumber\\
 &+20u^4+24\chi^5u^3+4\chi^3u^3-16\chi u^3\nonumber\\
 &-2\chi^4u^2+4\chi^2u^2+5u^2-24\chi^3u\nonumber\\
 &-14\chi u+9\chi^2 .
 \label{eq:app-middle-polynomial}
\end{align}
At the ends of the outer-horizon interval,
\begin{align}
 \mathcal R_\chi(0)&=9\chi^2>0,\nonumber\\
 \mathcal R_\chi(\chi)
 &=16\chi^4(1+\chi^2)(3\chi^4+6\chi^2-1)<0,\nonumber\\
 \mathcal R_\chi(1)
 &=(3\chi^3+7\chi^2+\chi+5)^2>0 .
 \label{eq:app-middle-endpoint-signs}
\end{align}
The only nonmanifest polynomial factor in the discriminant used below,
written in \(U=\chi^2\), is
\begin{align}
 \mathcal P(U)={}&13689U^8-159192U^7+333180U^6\nonumber\\
 &+1275288U^5+1171510U^4+436120U^3\nonumber\\
 &+52028U^2+2600U+25 .
 \label{eq:app-middle-discriminant-factor}
\end{align}

\subsection{Saturated reconstruction of the remaining equilibria}
\label{app:saturated-reconstruction}

The projection in Eq.~\eqref{eq:app-stationary-resultant} may acquire roots
through degree loss or denominator clearing.  We therefore work in the
corresponding saturation.  With the normalization in
Eq.~\eqref{eq:app-stationary-generator-normalization}, the polynomials
\(\widehat H_\chi\) and \(\widehat E_\chi\) have degrees two and four.  The
leading coefficient of the first polynomial has the manifestly positive
form
\begin{equation}
 \operatorname{lc}_w\widehat H_\chi
 =(1-u^2)\bigl[(u-\chi)^2+1-u^2\bigr]+\chi^2u^4>0
 \qquad (0<u<1).
 \label{eq:app-H-leading-positive}
\end{equation}
In particular, a physical common zero cannot escape to \(w=\infty\).

The subresultant sequence of
\((\widehat H_\chi,\widehat E_\chi)\) with respect to \(w\) has degrees
\(4,2,1,0\).  For \(0<u<1\) with \(u\ne\chi\), its linear member is
\begin{equation}
 S_1=-4u(u-\chi)(u-1)(u+1)
 \bigl[a_\chi(u)w+b_\chi(u)\bigr],
 \label{eq:app-linear-subresultant}
\end{equation}
where \(b_\chi\) is the constant coefficient and the leading coefficient is
\begin{align}
a_\chi(u)={}&
 \chi^9(12u^9-18u^7+3u^3)
 +\chi^8(4u^{12}+18u^8-32u^6+9u^4)\nonumber\\
&+\chi^7(10u^{11}+24u^9-46u^7)
 +\chi^7(10u^5+17u^3-5u)\nonumber\\
 &+\chi^6(-2u^{12}+32u^{10}+6u^8-70u^6
 +27u^4-5u^2+2)\nonumber\\
&+\chi^5(-11u^{11}+74u^9-98u^7)
 +\chi^5(26u^5+31u^3-13u)\nonumber\\
 &+\chi^4(-13u^{12}+30u^{10}+10u^8-54u^6
 +27u^4-9u^2+4)\nonumber\\
&+\chi^3(-25u^{11}+85u^9-94u^7)
 +\chi^3(22u^5+23u^3-11u)\nonumber\\
 &+\chi^2(-3u^{12}-3u^{10}+24u^8-26u^6
 +9u^4-3u^2+2)\nonumber\\
&+\chi(-6u^{11}+21u^9-24u^7)
 +\chi(6u^5+6u^3-3u)\nonumber\\
&+4u^{12}-15u^{10}+20u^8-10u^6+u^2 .
 \label{eq:app-linear-subresultant-leading}
\end{align}
Put \(U=\chi^2\) and define
\begin{align}
 C(U)={}&81U^9+261U^8-494U^7-2263U^6\nonumber\\
 &-1144U^5+3844U^4+7768U^3\nonumber\\
 &+6340U^2+2704U+400,
 \label{eq:app-C-polynomial}\\
 E(U)={}&9U^{10}-107U^9-224U^8+7099U^7\nonumber\\
 &-22403U^6-33453U^5+145606U^4\nonumber\\
 &+252521U^3+301240U^2+287900U+122500.
 \label{eq:app-E-polynomial}
\end{align}
Exact elimination gives
\begin{align}
 \operatorname{Res}_u(\mathcal R_\chi,a_\chi)
 ={}&2^{12}\chi^{16}(1+\chi^2)^6
 (3\chi^3-7\chi^2+\chi-5)^2\nonumber\\
 &\times(3\chi^3+7\chi^2+\chi+5)^2 C(\chi^2)^4,
 \label{eq:app-R-a-resultant}\\
 \operatorname{Res}_u\!\left(
 \mathcal R_\chi,\operatorname{lc}_w\widehat E_\chi\right)
 ={}&2^{24}\chi^{10}(1+\chi^2)(1+9\chi^2)
 C(\chi^2)E(\chi^2).
 \label{eq:app-R-E-leading-resultant}
\end{align}
Neither right-hand side vanishes in the interval
Eq.~\eqref{eq:cmetric-chi-window}.  Indeed, for \(0<U<1/6\), the positive
terms omitted from the following lower bounds give
\begin{align}
 C(U)
 &>U^4\!\left(3844-1144U-2263U^2-494U^3\right)\nonumber\\
 &>U^4(3844-191-63-3)
 =3587U^4>0,\label{eq:app-C-positive-bound}\\
 E(U)
 &>U^4\!\left(145606-33453U-22403U^2-224U^4-107U^5\right)\nonumber\\
 &>U^4(145606-5576-623-1-1)
 =139405U^4>0.
 \label{eq:app-E-positive-bound}
\end{align}
Since \(\chi_c^2<1/6\), both polynomials are positive throughout the
required interval.  The two cubic factors in
Eq.~\eqref{eq:app-R-a-resultant} have fixed nonzero signs for
\(0<\chi<1/2\).

\begin{lemma}[Saturated reconstruction at the turning controls]
\label{lem:app-saturated-stationary-reconstruction}
Fix \(0<\chi<\chi_c\).  For \(0<u<1\) with \(u\ne\chi\), every finite
common zero of \(\widehat H_\chi(u,w)\) and
\(\widehat E_\chi(u,w)\) is obtained from a root of
\(\mathcal R_\chi(u)\) by the unique rational reconstruction
\(w=-b_\chi(u)/a_\chi(u)\).  The polynomial
\(\mathcal R_\chi\) has precisely two positive roots.  One lies in
\((0,\chi)\) and has negative Hawking temperature; the other lies in
\((\chi,1)\) and is the unique additional positive-temperature algebraic
candidate.  If the latter candidate is admissible, its Gibbs value is
strictly positive.  Hence no physical equilibrium lies below the common
value \(\mathfrak g=0\) of the two exact
turning phases.
\end{lemma}

\begin{proof}
Equations~\eqref{eq:app-H-leading-positive},
\eqref{eq:app-R-a-resultant}, and
\eqref{eq:app-R-E-leading-resultant} have three consequences.  The input
degrees do not drop at a root of \(\mathcal R_\chi\), the linear
subresultant does not become constant, and the common root is the unique
finite real number
\begin{equation}
 w=-\frac{b_\chi(u)}{a_\chi(u)}.
 \label{eq:app-w-reconstruction}
\end{equation}
There is no cleared root at \(w=0\), since
\(\widehat H_\chi(u,0)=\chi^2u^4>0\).

For completeness, the number and labels of the positive
\(\mathcal R_\chi\)-roots are also fixed without continuation assumptions.
Its discriminant is
\begin{align}
&\operatorname{Res}_u
 (\mathcal R_\chi,\partial_u\mathcal R_\chi)\nonumber\\
={}&-2^{20}\chi^8(1+\chi^2)(1+9\chi^2)
 (20+21\chi^2)\nonumber\\
&\times(\chi^4+2\chi^2+5)
 (3\chi^3-7\chi^2+\chi-5)^2\nonumber\\
&\times(3\chi^3+7\chi^2+\chi+5)^2
 \mathcal P(\chi^2),
 \label{eq:app-R-discriminant-complete}
\end{align}
where \(\mathcal P\) is given in
Eq.~\eqref{eq:app-middle-discriminant-factor}.  Its only negative term is
\(-159192U^7\), and hence, for \(0<U<1/6\),
\begin{equation}
 \mathcal P(U)>U^6(333180-159192U)
 >306648U^6>0.
 \label{eq:app-P-positive-bound}
\end{equation}
Thus Eq.~\eqref{eq:app-R-discriminant-complete} is nonzero.  At
\(\chi=1/4\), Sturm isolation gives precisely two positive roots,
\begin{equation}
 \frac{2039}{10000}<u_{\rm in}<\frac{51}{250}<\frac14,
 \qquad
 \frac{198}{625}<u_3<\frac{3169}{10000}<1.
 \label{eq:app-R-refined-boxes}
\end{equation}
The constant and leading coefficients of \(\mathcal R_\chi\) are
\(9\chi^2\) and \(\chi^2(20+21\chi^2)\), respectively.  A root can
therefore enter neither through \(u=0\) nor through infinity.  The endpoint
signs in Eq.~\eqref{eq:app-middle-endpoint-signs}, together with the
nonvanishing discriminant, transport one root in \((0,\chi)\) and one in
\((\chi,1)\) throughout the open interval.

Rational interval evaluation of Eq.~\eqref{eq:app-w-reconstruction} on both
boxes in Eq.~\eqref{eq:app-R-refined-boxes} gives
\(w>0\) and \(w<\chi^2\).  On the first box \(H_u<0\), whereas on the
second \(H_u>0\).  These signs cannot change because
\(\tau^2=16\mu^2>0\).  Nor can either sheet meet the normalization face.
After setting \(w=\chi^2\), remove the nonzero \(u\)-polynomial content by
defining
\begin{equation}
 \widehat H_\chi^{\mathrm{nf}}(u)
 =\frac{\widehat H_\chi(u,\chi^2)}
 {\chi^2(1+\chi^2)}.
 \label{eq:app-normalization-face-primitive}
\end{equation}
The finite-temperature equation contains the factor
\(2\chi u^2+u-\chi\), and
\begin{equation}
 \operatorname{Res}_u\!\left(
 \widehat H_\chi^{\mathrm{nf}},2\chi u^2+u-\chi\right)
 =\chi^6(1+9\chi^2)>0.
 \label{eq:app-normalization-face-resultant}
\end{equation}
The resultant of the unnormalized restriction
\(\widehat H_\chi(u,\chi^2)\) is
\(\chi^{10}(1+\chi^2)^2(1+9\chi^2)\), as required by homogeneity.
Thus \(u_{\rm in}\) is the negative-temperature sheet, while \(u_3\) is
the unique additional positive-temperature candidate.

Finally, a zero of its Gibbs value would satisfy the squared necessary
condition \(\mathfrak m^2=(\tau s/4)^2\).  At the exact controls define
\begin{equation}
 G_0(\chi,u,w)=\operatorname{num}\!\left[
 \frac{\alpha^2}{wd^2}-s^2\right].
 \label{eq:app-third-state-G-generator}
\end{equation}
The full resultant, with the elimination variable displayed explicitly, is
\begin{equation}
 \operatorname{Res}_w(\widehat H_\chi,G_0)
 ={}
 \chi^2u^4(u-\chi)^4(1+\chi)^8(1+\chi^2)
 (u-1)^4(u+1)^4,
 \label{eq:app-third-state-G-resultant}
\end{equation}
which is nonzero on the third sheet.  Exact interval evaluation at
\(\chi=1/4\) gives
\(2/25<\mathfrak g_3<81/1000\).  Since
\(0<u_3<1\), \(0<w<\chi^2\), and the reconstruction coefficient never
vanishes, no thermodynamic pole intervenes.  Hence
\(\mathfrak g_3>0\) on the complete connected sheet.  This proves that the
two \(\mathfrak g=0\) minima exhaust the physical global minima at the exact
turning controls.
\end{proof}

\subsection{Rigidity of common-horizon Maxwell turns}
\label{app:common-horizon-rigidity}

The thermodynamic equations rigidly fix \(u_-=u_+\) on the exact curve within
the common-horizon sector.

\begin{theorem}[Common-horizon rigidity]
\label{thm:common-horizon-rigidity}
Fix \(0<\mu<1/4\) and \(q>0\).  Suppose two distinct real equilibria have
the same horizon coordinate \(u\in(0,1)\), positive normalization, equal
temperature, equal Gibbs free energy, and equal thermodynamic volume.  Then
they lie on Eqs.~\eqref{eq:cmetric-chi-mu}--\eqref{eq:cmetric-turning-data},
with \(\chi=u\).  Such a pair exists if and
only if \(0<u<\chi_c\).  Consequently, for every fixed \(\mu\) there is at
most one finite-separation common-horizon Maxwell turning, and it exists
precisely for \(0<\mu<\mu_+\).
\end{theorem}

\begin{proof}
Set
\begin{align}
 a&=\frac{1-d}{d},
 &\mathfrak k&=\frac{d^2}{q^2},\nonumber\\
 \rho_i&=\frac{\alpha_i}{\sqrt{w_i}}>0,
 &y_i&=1+\frac1{w_i}=\rho_i^2+\mathfrak k .
 \label{eq:app-rigidity-variables}
\end{align}
Multiplication of the horizon equation by \(1/w=y-1\) gives the quadratic
\begin{equation}
 \mathcal L(y;u)=\frac{u^4}{\mathfrak k}y^2+(1-u^2)u(u-a)y
 +(1-u^2)^2=0.
 \label{eq:app-rigidity-horizon}
\end{equation}
The thermodynamic functions take the form
\begin{align}
 \tau_i&=\frac{d\,\partial_u\mathcal L(y_i;u)}{y_i u^2\rho_i},
 &s_i&=\frac{u^2y_i}{d(1-u^2)},
 &\mathfrak m_i&=\frac{\mu\rho_i}{d},\nonumber\\
 \nu_i&=\frac1{\rho_i}\left[
 \frac{u^3y_i^2}{d^2(1-u^2)^2}+\frac{\mu d}{q^4}\right].
 \label{eq:app-rigidity-state-functions}
\end{align}
Distinct phases have \(\rho_1\ne\rho_2\).  Put
\begin{equation}
 R=\rho_1+\rho_2>0,\qquad P=\rho_1\rho_2>0,\qquad
 X=R^2,\qquad t=u^2.
\end{equation}
Equal Gibbs values and equal temperatures first give
\begin{equation}
 \tau=\frac{4\mu(1-u^2)}{u^2R}.
 \label{eq:app-rigidity-temperature}
\end{equation}
Equating this expression to each Hawking temperature in
Eq.~\eqref{eq:app-rigidity-state-functions}, using Vieta's relations for
Eq.~\eqref{eq:app-rigidity-horizon}, and imposing \(\nu_1=\nu_2\) yields a
polynomial system.  The Vieta sum first fixes
\begin{equation}
 a=\frac{u\{t(-2P+X+\mathfrak k)+\mathfrak k\}}
 {\mathfrak k(1-t)}.
 \label{eq:app-rigidity-a-elimination}
\end{equation}
After this substitution and removal only of the positive factors
\(u,\mathfrak k,R,1-t,y_1+y_2\), the remaining equations are
\begin{align}
0={}&P^2t^2-2P\mathfrak kt^2+X\mathfrak kt^2+\mathfrak k^2t^2\nonumber\\
&-\mathfrak kt^2+2\mathfrak kt-\mathfrak k,
 \label{eq:app-rigidity-EH}\\
0={}&-2P^2t+2PXt-2P\mathfrak kt^2+6P\mathfrak kt\nonumber\\
&+X\mathfrak kt^2-X\mathfrak kt+\mathfrak k^2t^2
-2\mathfrak k^2t-\mathfrak k^2,
 \label{eq:app-rigidity-EV}\\
0={}&-2Pt^2+2Pt+Xt^2+Xt+\mathfrak kt^2-\mathfrak k,
 \label{eq:app-rigidity-ED}\\
0={}&16P^2t^2-8P^2t-6PXt^2-2PXt\nonumber\\
&-20P\mathfrak kt^2+4P\mathfrak k+X^2t^2+X^2t\nonumber\\
&+9X\mathfrak kt^2-X\mathfrak k+8\mathfrak k^2t^2
-8\mathfrak kt^2\nonumber\\
&+16\mathfrak kt-8\mathfrak k.
 \label{eq:app-rigidity-ES}
\end{align}
These four equations use no squared sign choice.  Their exact lexicographic
Gröbner basis contains
\begin{equation}
 -\mathfrak k(t-1)^3(\mathfrak kt-t-1)
 (-\mathfrak kt-\mathfrak k+t^2-6t+9).
 \label{eq:app-rigidity-branches}
\end{equation}
It also contains
\begin{equation}
 \mathfrak k(t-1)\bigl[Pt(t^2-6t+1)+2\mathfrak kt(t+1)
 -t^3+3t^2+t-3\bigr].
 \label{eq:app-rigidity-P-basis}
\end{equation}
Because \(\mathfrak k>0\) and \(0<t<1\), there are only two algebraic
branches.  If
\(t^2-6t+1\ne0\), the same basis and
Eq.~\eqref{eq:app-rigidity-ED} give
\begin{align}
\text{I}\quad
 \mathfrak k&=\frac{1+t}{t},
 &P&=\frac{1+t}{t},\nonumber\\
 X&=\frac{(1-t)^2}{t^2},
 &a&=\frac{2u}{1+u^2},
 \label{eq:app-rigidity-branch-I}\\
\text{II}\quad
 \mathfrak k&=\frac{(3-t)^2}{1+t},
 &P&=\frac{3-t}{t},\nonumber\\
 X&=\frac{(1-t)^2(3-t)}{t(1+t)},
 &a&=\frac{2u}{3-u^2}.
 \label{eq:app-rigidity-branch-II}
\end{align}
The second branch has
\begin{equation}
 X-4P=-\frac{(t-3)(t^2-6t-3)}{t(1+t)}<0
 \qquad (0<t<1),
\end{equation}
and therefore cannot contain two real \(\rho_i\).  At the only exceptional
value in \((0,1)\), \(t=3-2\sqrt2\), both branches have
\(\mathfrak k=4+2\sqrt2\).  Direct substitution leaves the two possibilities
\begin{equation}
 (P,X)=(4+2\sqrt2,12+8\sqrt2),\qquad
 (P,X)=(8+6\sqrt2,4+4\sqrt2),
\end{equation}
for which \(X-4P=-4\) and
\(-28-20\sqrt2\), respectively.  Neither is realizable by positive
\(\rho_1,\rho_2\).

On branch I,
\begin{equation}
 X-4P=\frac{1-6t-3t^2}{t^2}.
\end{equation}
Two distinct real positive roots exist exactly when
\(1-6t-3t^2>0\), or \(0<u<\chi_c\).  Moreover,
Eq.~\eqref{eq:app-rigidity-branch-I} gives
\begin{equation}
 d=\frac{1+u^2}{(1+u)^2},\qquad
 \mu=\frac{u}{(1+u)^2},\qquad
 q^2=\frac{u^2(1+u^2)}{(1+u)^4}.
\end{equation}
Equation~\eqref{eq:app-rigidity-temperature} then gives
\(\tau=4\mu\).  Recovering \(w_i=1/(y_i-1)\) reproduces precisely the two
roots in Eq.~\eqref{eq:cmetric-turning-roots}; substitution gives
\(\mathfrak g_1=\mathfrak g_2=0\) and the common volume in
Eq.~\eqref{eq:cmetric-turning-data}.  Finally,
\(u\mapsto u/(1+u)^2\) is strictly increasing for \(0<u<1\), which proves
the fixed-tension uniqueness statement.
\end{proof}

\subsection{Unrestricted positive-domain elimination}
\label{app:unrestricted-turning-classification}

The common-horizon restriction can be removed before any squaring of the
temperature.  Throughout this subsection, a finite-separation
positive-temperature Maxwell turning pair means two distinct physical phases
at common controls with equal temperature, equal Gibbs value, and
\(\Delta V=0\).

\begin{proposition}[Unrestricted common-horizon rigidity]
\label{prop:app-unrestricted-common-horizon-rigidity}
Every finite-separation positive-temperature Maxwell turning pair in the open
physical domain of the standard single-string fixed-\((P,Q,\mu)\) ensemble
has equal horizon coordinate.
\end{proposition}

The following coordinates give every cleared denominator a fixed sign.  Let
\begin{equation}
 a=\frac{1-d}{d},\qquad \mathfrak k=\frac{d^2}{q^2},\qquad
 \mathfrak q=\frac1{\mathfrak k}=\frac{q^2}{d^2}>0 .
 \label{eq:app-unrestricted-controls}
\end{equation}
For a physical equilibrium with horizon coordinate \(u\), put
\begin{align}
 \rho&=\frac{\alpha}{\sqrt w}>0,
 &y&=\rho^2+\mathfrak k,\nonumber\\
 x&=\frac{u^2y}{\mathfrak k(1-u^2)}>0,
 &h&=\frac{u\rho}{\sqrt{\mathfrak k}}>0.
 \label{eq:app-unrestricted-state-change}
\end{align}
The transformation is nonsingular in the open physical domain.  It gives
\begin{equation}
 h^2=x-(1+x)u^2,\qquad 0<h^2<x.
 \label{eq:app-unrestricted-lapse-relation}
\end{equation}
Substitution in the horizon equation
\eqref{eq:app-rigidity-horizon} reduces that equation to
\begin{equation}
 a=u\frac{\mathfrak q+x+x^2}{x}.
 \label{eq:app-unrestricted-horizon-solved}
\end{equation}
Thus the horizon constraint has been solved without choosing a polynomial
root.

Write \(\eta=h^2\) and introduce
\begin{equation}
 B_x=\mathfrak q+x+x^2,\qquad
 F_x=x(2x+1)(x+1)-\mathfrak q .
 \label{eq:app-unrestricted-BF}
\end{equation}
Three polynomial numerators will be used:
\begin{align}
 T_x(\eta)&=\eta B_x+(2x+1)(x^2+x-\mathfrak q),\nonumber\\
 G_x(\eta)&=\eta(x+2)B_x
 -x(2x+1)(x^2+x-\mathfrak q),\nonumber\\
 V_x(\eta)&=-\eta B_x
 +x\{\mathfrak q+x(2x+1)(x+1)\}.
 \label{eq:app-unrestricted-numerators}
\end{align}
Direct substitution in Eq.~\eqref{eq:app-rigidity-state-functions} gives
\begin{align}
 \frac{(1+a)\tau}{\sqrt{\mathfrak q}}
 &=\frac{T_x(\eta)}{hx(1+x)},\nonumber\\
 4\sqrt{\mathfrak q}\,\mathfrak g
 &=\frac{G_x(\eta)}{hx(1+x)},\nonumber\\
 2\mathfrak q^{3/2}\nu
 &=\frac{V_x(\eta)}{hx(1+x)}.
 \label{eq:app-unrestricted-scaled-state-functions}
\end{align}

\begin{lemma}[Nondegeneracy of the positive-domain reduction]
\label{lem:app-unrestricted-legal-reduction}
For two distinct physical phases \((x,h)\) and \((y,j)\) at common
\((a,\mathfrak q)\), relabeled if necessary so that \(x>y\), one has
\[
 x,y,h,j,F_x,F_y,A>0,\qquad x-y>0.
\]
Equality of temperature and volume fixes the positive lapse ratio \(j/h\)
and then a unique rational value \(\eta_*(x,y,\mathfrak q)\).  Every divisor
used in these two eliminations is nonzero and has fixed sign after the phase
ordering.
\end{lemma}

\begin{proof}
The identity
\begin{equation}
 V_x(\eta)+T_x(\eta)=(1+x)F_x
 \label{eq:app-unrestricted-VT-identity}
\end{equation}
and Eq.~\eqref{eq:app-unrestricted-scaled-state-functions} show that positive
temperature and positive thermodynamic volume imply \(F_x>0\), and likewise
\(F_y>0\).

Consider two distinct physical phases, denoted by \((x,h)\) and \((y,j)\),
at common \((a,\mathfrak q)\).  They have \(x\ne y\).  Indeed, if \(x=y\),
Eq.~\eqref{eq:app-unrestricted-horizon-solved} first gives equal positive
horizon coordinates, and Eq.~\eqref{eq:app-unrestricted-lapse-relation}
then gives \(h=j\), so the two states coincide.  Equality of temperature
and volume, together with
Eq.~\eqref{eq:app-unrestricted-VT-identity}, fixes the positive ratio
\begin{equation}
 \Lambda_h:=\frac jh=\frac{xF_y}{yF_x}>0.
 \label{eq:app-unrestricted-R}
\end{equation}
The remaining temperature equation is
\begin{equation}
 E_T=\Lambda_h y(1+y)T_x(\eta)
 -x(1+x)T_y(\Lambda_h^2\eta)=0.
 \label{eq:app-unrestricted-ET}
\end{equation}
It is linear in \(\eta\).  After clearing the positive denominator in
\(\Lambda_h\),
its leading coefficient is
\begin{equation}
 x(x-y)F_y A(x,y,\mathfrak q),
 \label{eq:app-unrestricted-ET-leading}
\end{equation}
where
\begin{align}
 A={}&\mathfrak q^2(x^2+xy+x+y^2+y)
 +2x^4y^3+2x^4y^2
 +2x^3y^4+7x^3y^3
 \nonumber\\
 &+5x^3y^2+2x^2y^4+5x^2y^3+3x^2y^2>0.
 \label{eq:app-unrestricted-A-positive}
\end{align}
Consequently \(E_T\) determines a unique rational value
\(\eta=\eta_*(x,y,\mathfrak q)\); no degree drop is possible in the physical
domain.
Every monomial in Eq.~\eqref{eq:app-unrestricted-A-positive} is positive, so
\(A>0\).  After relabeling so that \(x>y\), all factors
\(F_x,F_y,A,x-y\) have fixed positive sign.
The denominators generated by this substitution are, up to positive
monomials, \(F_yA\), \(F_xA\), and \(F_xF_yA\) in \(\eta_*\), the Gibbs
equation, and the common-tension equation, respectively.  They are therefore
nonzero before the primitive numerators are formed.
\end{proof}

\begin{lemma}[The only positive projected component]
\label{lem:app-unrestricted-primary-eliminant}
Every distinct positive physical solution of the equal-Gibbs and
common-tension equations after substitution of \(\eta_*\) satisfies
\(2xy+x+y=1\).
\end{lemma}

\begin{proof}
Equal Gibbs values give
\begin{equation}
 E_G=\Lambda_h y(1+y)G_x(\eta)
 -x(1+x)G_y(\Lambda_h^2\eta)=0.
 \label{eq:app-unrestricted-EG}
\end{equation}
Common positive tension, using
Eq.~\eqref{eq:app-unrestricted-horizon-solved}, also gives
\begin{equation}
 E_a=
 \frac{B_x^2(x-\eta)}{x^2(1+x)}
 -\frac{B_y^2(y-\Lambda_h^2\eta)}{y^2(1+y)}=0.
 \label{eq:app-unrestricted-Ea}
\end{equation}
This equation is equality of \(a^2\); the positive sign of \(a\) is already
fixed by the physical tension interval.  Substitute \(\eta_*\) into
Eqs.~\eqref{eq:app-unrestricted-EG} and
\eqref{eq:app-unrestricted-Ea}.  After canceling only powers of
\(x,y,1+x,1+y,F_x,F_y,x-y,A\), let
\(\mathcal P(x,y,\mathfrak q)\) and
\(\mathcal Q(x,y,\mathfrak q)\) be the primitive remaining numerators.  Their
degrees in \(\mathfrak q\) are four and six.  Equivalently, the elimination is
performed in the physical saturation
\begin{align}
 I_{\rm phys}
 ={}&\left\langle
 \operatorname{num}E_G(\eta_*),\operatorname{num}E_a(\eta_*)
 \right\rangle:S_{\rm phys}^{\infty},
 \label{eq:app-unrestricted-saturation-ideal}\\
 S_{\rm phys}
 ={}&xy(1+x)(1+y)F_xF_y(x-y)A .
 \label{eq:app-unrestricted-saturation-factor}
\end{align}
Every factor in \(S_{\rm phys}\) is nonzero and has fixed sign after the
phase ordering, with \(F_x,F_y,A>0\), so the saturation removes no physical
solution.  Exact elimination then gives
\begingroup\small
\begin{align}
 \operatorname{Res}_{\mathfrak q}(\mathcal P,\mathcal Q)
={}&-16x^{16}y^{16}(1+x)^2(1+y)^2(x-y)^8
 (2x+1)(2y+1)
 \nonumber\\
 &\times(1+x+y)^4(2xy+x+y)^2
 (2xy+x+y-1)
 \nonumber\\
 &\times(x^2+x+2y^3+3y^2+y)^2
 \nonumber\\
 &\times(2x^3+3x^2+x+y^2+y)^2
 \nonumber\\
 &\times(2x^2+2xy+3x+2y^2+3y+1)^9
 \nonumber\\
 &\times R_+(x,y),
 \label{eq:app-unrestricted-first-resultant}
\end{align}
\endgroup
with
\begin{align}
 R_+(x,y)={}&4x^3y+2x^3+4x^2y^2
 +18x^2y+5x^2+4xy^3
 \nonumber\\
 &+18xy^2+18xy+4x+2y^3
 \nonumber\\
 &+5y^2+4y+1.
 \label{eq:app-unrestricted-positive-tail}
\end{align}
Every factor in Eq.~\eqref{eq:app-unrestricted-first-resultant} is strictly
positive for \(x,y>0\), except the displayed nonzero phase-separation factor
and the single factor \(2xy+x+y-1\).  A physical turning pair must therefore
satisfy
\begin{equation}
 2xy+x+y=1.
 \label{eq:app-unrestricted-entropy-pair}
\end{equation}
\end{proof}

\begin{lemma}[Exceptional fibers and reconstruction]
\label{lem:app-unrestricted-exceptional-fibers}
On the positive locus \(2xy+x+y=1\), every distinct physical solution has
\(\mathfrak q=xy\); no exceptional value of \(x\) supports an additional
common \(\mathfrak q\)-root.
\end{lemma}

\begin{proof}
It remains to reconstruct \(\mathfrak q\), including the exceptional values at
which a generic polynomial gcd could increase.  Set
\begin{equation}
 y=\frac{1-x}{2x+1}.
 \label{eq:app-unrestricted-yx}
\end{equation}
Positivity gives \(0<x<1\), and \(x\ne y\) is equivalent to
\(2x^2+2x-1\ne0\).  On this locus the two primitive equations factor as
\begin{align}
 \mathcal P&=-(2x+1)L\mathcal P_3,
 &\deg_{\mathfrak q}\mathcal P_3&=3,\nonumber\\
 \mathcal Q&=-(2x+1)L\mathcal Q_5,
 &\deg_{\mathfrak q}\mathcal Q_5&=5,
 \label{eq:app-unrestricted-restricted-PQ}\\
 L&=(2x+1)\mathfrak q+x^2-x .
 \nonumber
\end{align}
The possible common roots of the two residual cofactors are excluded by
\begin{align}
 \operatorname{Res}_{\mathfrak q}(\mathcal P_3,\mathcal Q_5)
={}&98304x^8(x-1)^8(x+1)^2(x+2)^2
 \nonumber\\
 &\times(2x+1)^{10}(2x^2+1)^8
 \nonumber\\
 &\times(x^2+x+1)^2(2x^2+2x-1)^6
 \nonumber\\
 &\times(2x^2+4x+3)^5
 \nonumber\\
 &\times(4x^2+2x+3)A_5(x)^2
 \nonumber\\
 &\times B_5(x)^2,
 \label{eq:app-unrestricted-exceptional-resultant}\\
 A_5(x)={}&4x^5+10x^4+9x^3+3x^2+1,
 \nonumber\\
 B_5(x)={}&4x^5+10x^4+9x^3
 +2x^2-x+3.
 \nonumber
\end{align}
The right-hand side is nonzero on the distinct-phase interval.  In
particular, \(B_5(x)>3-x>0\).  Hence \(L=0\), and
\begin{equation}
 \mathfrak q=\frac{x(1-x)}{2x+1}=xy.
 \label{eq:app-unrestricted-kappa-xy}
\end{equation}

Finally, \(B_x=x(1+x+y)\) and \(B_y=y(1+x+y)\).  The common-\(a\) horizon
relation \eqref{eq:app-unrestricted-horizon-solved} therefore yields
\begin{equation}
 u_1=\frac{a}{1+x+y}=u_2.
 \label{eq:app-unrestricted-common-horizon-conclusion}
\end{equation}
\end{proof}

\begin{proof}[Proof of Proposition~\ref{prop:app-unrestricted-common-horizon-rigidity}]
Lemmas~\ref{lem:app-unrestricted-legal-reduction}--
\ref{lem:app-unrestricted-exceptional-fibers} show that every turning pair
in the unrestricted physical domain satisfies
Eq.~\eqref{eq:app-unrestricted-common-horizon-conclusion} and therefore
belongs to the common-horizon sector.
Theorem~\ref{thm:common-horizon-rigidity} then gives the closed locus in
Theorem~\ref{thm:cmetric-exact-turning-main}, including its exact tension
range.  The elimination has used the thermodynamic positivity conditions
only to remove factors with fixed sign; outer-horizon, slow-acceleration,
canonical-stability, and global-minimum conditions are established by the
exact arguments for the surviving locus.
\end{proof}

Combining Theorem~\ref{thm:common-horizon-rigidity} with the unrestricted
elimination in \ref{app:unrestricted-turning-classification} proves that the
same locus exhausts all finite-separation physical Maxwell turns in the
standard single-string ensemble.

The turning locus, its physicality and stability, the endpoint laws, and
global phase selection have therefore been established by exact polynomial
identities and real-root isolation.  The numerical continuation below is used
only to display the surrounding wall complex.

\subsection{Numerical wall construction}

For each control pair \((q,\mu)\), the construction enumerates all real roots
of the quartic horizon equation in \(0<u<1\) and then imposes the
positive-temperature outer-horizon and slow-acceleration conditions.  Fold
sheets are obtained from the zeros of \(\mathscr D_z\tau\) on each connected
physical interval.  Maxwell candidates solve the two horizon equations
together with equal temperature and equal Gibbs value, after which they are
compared with every physical equilibrium at the same controls.

Over \(0.015\leq\mu\leq0.245\), this continuation resolves both the
slow-acceleration and bulk-extremal exit faces and reproduces the chamber
structure shown in Fig.~\ref{fig:cmetric-wall-atlas}.  The sampling grids,
direct two-phase checks, and residual data are included in the Supplementary
Material.

\section{Small-tension Maxwell boundary-layer analysis}
\label{app:cmetric-blowup-proof}

This appendix proves Theorem~\ref{thm:cmetric-maxwell-blowup} directly from
the unsquared reduced C-metric equations.  The desingularization keeps the
two state-space charts separate, and every coefficient follows algebraically
from the exact equations.
Throughout, \(I=[r_0,r_1]\Subset(0,\infty)\), and \(O_I\) denotes a remainder
uniform on that fixed compact interval.

\subsection{Desingularized two-phase equations}

\begin{proof}[Proof of Theorem~\ref{thm:cmetric-maxwell-blowup}]
The horizon equation is first solved for \(p\), which avoids choosing a
branch of the quartic root during the expansion.  The two state-space charts
are
\begin{align}
 u_L&=\mu+2\mu^2+a\mu^3,
 &z_L&=\mu^2+4\mu^3+b\mu^4,\nonumber\\
 u_S&=\mu+2\mu^2+c\mu^3,
 &z_S&=\mu+2\mu^2+\frac92\mu^3+11\mu^4+k\mu^5.
 \label{eq:cmetric-blowup-charts}
\end{align}
Substitution in the exact horizon surface gives
\begin{align}
 p_L={}&\mu^2+(-4a+2b-10)\mu^4
 +8(a-b+3)\mu^5+O(\mu^6),\nonumber\\
 p_S={}&\mu^2+A\mu^4+A_1\mu^5+O(\mu^6),
 \label{eq:cmetric-blowup-pressure-expansion}\\
 A={}&-c^2+6c-2k+\frac{191}{4},\nonumber\\
 A_1={}&8c^2-44c+12k-129.
 \nonumber
\end{align}
The lapse factors have different orders,
\begin{equation}
 \alpha_L=1+O(\mu^2),
 \qquad
 \alpha_S=\mu\sqrt{A}\left(1+\frac{A_1}{2A}\mu+O(\mu^2)\right).
 \label{eq:cmetric-blowup-lapse}
\end{equation}
Writing the temperature as \(\tau=N/\alpha\), its required numerator
coefficients are
\begin{align}
 N_L={}&4\mu-2(5a-2b+7)\mu^3
 +4(5a-4b+10)\mu^4+O(\mu^5),\nonumber\\
 N_S={}&2(c-3)\mu^2-4(2c-5)\mu^3+O(\mu^4).
 \label{eq:cmetric-blowup-temperature-numerators}
\end{align}
It is convenient to write \(\mathfrak g=J/\alpha\).  Direct reduction by
the horizon equation gives
\begin{align}
 J_L={}&\frac{a-5}{2}\mu+(3-2a)\mu^2+O(\mu^3),\nonumber\\
 J_S={}&C_g\mu^2+C_{g,1}\mu^3+O(\mu^4),
 \label{eq:cmetric-blowup-value-numerators}\\
 C_g={}&-\frac{4c^2-22c+8k-197}{4},\nonumber\\
 C_{g,1}={}&8c^2-43c+12k-131.
 \nonumber
\end{align}

After division by the indicated powers of \(\mu\), the two pressure
equations, equal temperature, and equal critical value have the limiting
system
\begin{equation}
 \begin{split}
 -4a+2b-10-r^2&=0,\\
 A-r^2&=0,\\
 \frac{2(c-3)}{\sqrt A}-4&=0,\\
 \frac{C_g}{\sqrt A}-\frac{a-5}{2}&=0.
 \end{split}
 \label{eq:cmetric-blowup-leading-system}
\end{equation}
Its physical solution is
\begin{equation}
 a=c=3+2r,
 \qquad
 b=11+4r+\frac{r^2}{2},
 \qquad
 k=\frac{227}{8}-\frac52r^2.
 \label{eq:cmetric-blowup-leading-solution}
\end{equation}
The Jacobian of Eq.~\eqref{eq:cmetric-blowup-leading-system} with respect to
\((a,b,c,k)\) has determinant
\begin{equation}
 \det D_{(a,b,c,k)}\mathcal E_0=\frac4r.
 \label{eq:cmetric-blowup-jacobian}
\end{equation}
It is nonzero for every \(r>0\).  The analytic implicit-function theorem,
applied uniformly on \([r_0,r_1]\), therefore produces the unique nearby
Maxwell pair.  Solving the first correction gives
\begin{equation}
 (a_1,b_1,c_1,k_1)=
 \left(8r+2,\,2r^2+24r+24,\,8r+2,
 \frac{303}{4}-15r^2\right),
 \label{eq:cmetric-blowup-first-correction}
\end{equation}
which proves Eqs.~\eqref{eq:cmetric-blowup-large-state} and
\eqref{eq:cmetric-blowup-small-state}.  The low-\(z\) chart is regular in the
lapse and yields Eqs.~\eqref{eq:cmetric-blowup-tau} and
\eqref{eq:cmetric-blowup-gibbs} without squaring either equality.  Combining
Eq.~\eqref{eq:cmetric-blowup-tau} with
Eq.~\eqref{eq:cmetric-blowup-boundary-series} proves
Eq.~\eqref{eq:cmetric-universal-maxwell-curve}.

The solution lies in the physical chamber.  The low-\(z\) phase has
\(\alpha_L=1+O_I(\mu^2)\), and the positive term proportional to
\(p/z_L^2\) keeps its boundary polynomial strictly positive.  For the
high-\(z\) phase, write the angular coordinate in the boundary polynomial as
\(v=\mu y\).  Uniformly for bounded \(y\), direct substitution gives
\begin{equation}
 C(\mu y)=\mu^2\bigl[(y-1)^2+r^2\bigr]+O_I(\mu^3).
 \label{eq:cmetric-blowup-slow-margin}
\end{equation}
Outside this angular scale the leading \(v^2\) term is positive.  Hence the
slow-acceleration margin is positive for \(r\geq r_0\) and sufficiently small
\(\mu\).  The angular metric function tends uniformly to one.

\medskip
\noindent\textit{Outer-horizon selection.}
Put \(w=z^2\) and write
\begin{align}
 \mathcal H(u,w)
 &=(1-u^2)\bigl(u^2-2\mathfrak a u+w\bigr)
 +\frac{u^4p(1+w)^2}{wd^2},\nonumber\\
 \mathfrak a&=\frac{\mu(1+w)}d .
 \label{eq:cmetric-blowup-Hw}
\end{align}
At a root with \(0<u<1\), the last term in
Eq.~\eqref{eq:cmetric-blowup-Hw} is positive.  Consequently
\begin{equation}
 u^2-2\mathfrak a u+w<0,
 \qquad
 0<u<2\mathfrak a,
 \qquad
 0<w<\mathfrak a^2.
 \label{eq:cmetric-blowup-core-bound}
\end{equation}
These are exact inequalities, before any expansion.  They put every
possible larger root in a compact \(u/\mu\) interval.

Let \(w_L(\mu,r)\) and \(w_S(\mu,r)\) be the two functions in
Eq.~\eqref{eq:cmetric-blowup-w-states}, and set
\begin{equation}
 V_L=\frac{u_L}{\mu},
 \qquad
 C_S=\frac{u_S-\mu-2\mu^2}{\mu^3}.
 \label{eq:cmetric-blowup-root-coordinates}
\end{equation}
Direct expansion of the polynomial gives
\begin{equation}
 \mu^{-2}\mathcal H(\mu V,w_L)
 =V(V-1)(V^2+V+2)+O_I(\mu)
 \label{eq:cmetric-blowup-low-root-factor}
\end{equation}
uniformly on compact \(V\)-intervals.  Analytic division by the known
root defines the continuous quotient
\begin{align}
 \mathcal Q_L(V;\mu,r)
 &=\frac{\mathcal H(\mu V,w_L)}{\mu^2(V-V_L)},\nonumber\\
 \mathcal Q_L(V_L;\mu,r)
 &=
 \left.\partial_V\bigl[\mu^{-2}\mathcal H(\mu V,w_L)\bigr]
 \right|_{V=V_L} .
 \label{eq:cmetric-blowup-low-root-quotient}
\end{align}
Equation~\eqref{eq:cmetric-blowup-low-root-factor} implies
\begin{equation}
 \mathcal Q_L(V;\mu,r)
 \longrightarrow V(V^2+V+2)>0
 \label{eq:cmetric-blowup-low-root-sign}
\end{equation}
uniformly for \(r\in I\), a fixed \(0<\varepsilon<1\), and
\(V\in[1-\varepsilon,2+\varepsilon]\).
In view of Eq.~\eqref{eq:cmetric-blowup-core-bound}, every root larger
than \(u_L\) would lie in this interval and would make the left-hand side
of Eq.~\eqref{eq:cmetric-blowup-low-root-quotient} vanish.  This is
impossible for sufficiently small \(\mu\), so \(u_L\) is the outer root.

The high-\(z\) polynomial has a double leading root and must be resolved
one weight further.  Put \(u_Y=\mu+Y\mu^2\) and
\(u_C=\mu+2\mu^2+C\mu^3\).  Uniformly for bounded \(Y\) and \(C\),
\begin{align}
 \mu^{-4}\mathcal H(u_Y,w_S)
 &=(Y-2)^2-6\mu(Y-2)+O_I(\mu^2),
 \label{eq:cmetric-blowup-high-root-first}\\
 \mu^{-6}\mathcal H(u_C,w_S)
 &=[C-(3+2r)][C-(3-2r)]+O_I(\mu).
 \label{eq:cmetric-blowup-high-root-factor}
\end{align}
The preceding \(u=\mu V\) expansion starts with \((V-1)^2\).
Successive domination in these three displayed expansions shows that
every root in the high-\(z\) cluster has, in order,
\(V=1+O_I(\mu)\), \(Y=2+O_I(\mu)\), and bounded \(C\).  Thus there
are precisely two roots in this cluster,
\begin{equation}
 C_-=3-2r+O_I(\mu),
 \qquad
 C_+=3+2r+O_I(\mu),
 \label{eq:cmetric-blowup-high-roots}
\end{equation}
and the expansion of \(u_S\) identifies \(C_S=C_+\).  Dividing by this
root gives
\begin{equation}
 \mathcal Q_S(C;\mu,r)=
 \frac{\mathcal H(\mu+2\mu^2+C\mu^3,w_S)}
 {\mu^6(C-C_S)}
 \longrightarrow C-(3-2r).
 \label{eq:cmetric-blowup-high-root-quotient}
\end{equation}
The quotient is again extended at \(C=C_S\) by differentiation.  For
\(C\geq C_S\), its limiting value is at least \(4r\geq4r_0\).
Hence it is uniformly positive for small \(\mu\), and \(u_S\) has no
larger root.  The two selected horizons are therefore outer roots.  The same
expansion gives
\begin{equation}
 \frac{\tau(C;\mu,r)}\mu
 =\frac{2(C-3)}r+O_I(\mu)
 \label{eq:cmetric-blowup-fixed-w-temperature}
\end{equation}
on the two roots in Eq.~\eqref{eq:cmetric-blowup-high-roots}.  The inner
root \(C_-\) has limiting temperature \(-4\), whereas the outer root
\(C_+\) has limiting temperature \(+4\).  Thus the choice of the outer
root is also fixed by the unsquared positive-temperature equation.

At fixed \(p\),
\begin{align}
 \left(\frac{\partial\tau}{\partial z}\right)_{p,L}
 &=-\frac1\mu+O(1),\nonumber\\
 \left(\frac{\partial\tau}{\partial z}\right)_{p,3}
 &=\frac{25}{81\mu}+O(1),\nonumber\\
 \left(\frac{\partial\tau}{\partial z}\right)_{p,S}
 &=-\frac{1}{r^2\mu^4}\bigl(1+O_I(\mu)\bigr).
 \label{eq:cmetric-blowup-stability}
\end{align}
The entropy derivative is negative on the positive-temperature physical
sheet.  Thus \(L\) and \(S\) are strict minima, while the third state is a
maximum.  Lemma~\ref{lem:cmetric-uniform-root-exhaustion} below and the strict
Gibbs-value gap show that the two minima exhaust the physical global minima
in this boundary layer.  Finally,
\(s_S-s_L=-\mu^{-2}(1+O_I(\mu))\), so the envelope identity implies that
\(L\) wins above coexistence and \(S\) below it.
\end{proof}

\begin{lemma}[Uniform exhaustion of the boundary-layer equilibria]
\label{lem:cmetric-uniform-root-exhaustion}
Let \(I=[r_0,r_1]\Subset(0,\infty)\),
\begin{equation}
 p=\mu^2(1+\mu^2r^2),\qquad r\in I,
\end{equation}
and let \(\tau_M(\mu,r)\) be the Maxwell temperature constructed above.
There exist constants
\begin{equation}
 \mu_0>0,\qquad c_0>0,\qquad C_0<\infty,\qquad \delta_0>0,
 \label{eq:cmetric-uniform-exhaustion-constants}
\end{equation}
depending only on \(I\), with the following property.  For every
\(0<\mu<\mu_0\) and every \(r\in I\), each solution of the unsquared system
\begin{equation}
 \begin{aligned}
  \mathcal H(u,z^2)&=0, & \alpha^2&>0, & 0<u&<1,\quad z>0,\\
  \tau(u,z;p,\mu)&=\tau_M(\mu,r)>0
 \end{aligned}
 \label{eq:cmetric-uniform-unsquared-system}
\end{equation}
lies in exactly one of the following two scale-separated charts:
\begin{align}
 \mathcal U_H={}&\left\{
 \left|\frac u\mu-1\right|\le C_0\mu,
 \quad
 \left|\frac{z^2}{\mu^2}-1\right|\le C_0\mu
 \right\},
 \label{eq:cmetric-uniform-high-chart}\\
 \mathcal U_L={}&\left\{
 c_0\le\frac u\mu\le2-c_0,
 \quad
 c_0\le\frac z{\mu^2}\le C_0
 \right\}.
 \label{eq:cmetric-uniform-low-chart}
\end{align}
The high chart contains exactly one positive-temperature solution, the
boundary-linked branch \(S\).  The low chart contains exactly two solutions,
the large-entropy branch \(L\) and the third branch \(3\).  More precisely,
in the disjoint union of the two coordinate charts, they lie in pairwise
disjoint \(\delta_0\)-neighborhoods of
\begin{align}
 (C,\Omega)_S&=(3+2r,121-5r^2),
 \label{eq:cmetric-uniform-high-center}\\
 (V,W)_L&=(1,1),\qquad
 (V,W)_3=\left(\frac95,\frac{27}{5}\right),
 \label{eq:cmetric-uniform-low-centers}
\end{align}
where
\begin{align}
 \mathcal U_H:\quad
 u&=\mu+2\mu^2+C\mu^3,\nonumber\\
 z^2&=\mu^2+4\mu^3+13\mu^4+40\mu^5+\Omega\mu^6,
 \label{eq:cmetric-uniform-high-coordinates}\\
 \mathcal U_L:\quad
 u&=\mu V,\qquad z=\mu^2W.
 \label{eq:cmetric-uniform-low-coordinates}
\end{align}
In particular, no equilibrium occurs at an intermediate scale between
\(z=O(\mu)\) and \(z=O(\mu^2)\), and the smallness threshold \(\mu_0\) is
uniform for \(r\in I\).
\end{lemma}

\begin{proof}
Uniform compactness follows without selecting a root of the horizon quartic.
The lapse inequality gives
\begin{equation}
 0<z^2<\frac{p}{d^2-p}=O_I(\mu^2),
 \label{eq:cmetric-uniform-lapse-bound}
\end{equation}
for sufficiently small \(\mu\).  At a root of \(\mathcal H\), positivity of
the charge term gives the exact inequalities
\begin{align}
 u^2-2\mathfrak a u+z^2&<0,\qquad
 0<u<2\mathfrak a,\nonumber\\
 0<z^2&<\mathfrak a^2,\qquad
 \mathfrak a=\frac{\mu(1+z^2)}d.
 \label{eq:cmetric-uniform-core-bounds}
\end{align}
Thus \(V=u/\mu\) and \(Z=z/\mu\) range in a compact set independent of
\(r\in I\) and small \(\mu\).

Consider a sequence of solutions of
Eq.~\eqref{eq:cmetric-uniform-unsquared-system} with \(\mu\to0\).  After
passing to a subsequence, \(r\to r_*\in I\) and \((V,Z)\) converges.  If
\(Z\to Z_*\in(0,1)\), the leading horizon and unsquared-temperature
relations are
\begin{equation}
 Z_*^2=V_*(2-V_*),\qquad
 \tau\longrightarrow
 \frac{2Z_*(V_*-1)}{V_*^2\sqrt{1-Z_*^2}}.
 \label{eq:cmetric-uniform-interior-limit}
\end{equation}
Because \(\tau_M=4\mu+O_I(\mu^3)\), the second limit must vanish.  It gives
\(V_*=1\), after which the first relation would give \(Z_*=1\), a
contradiction.  The lapse inequality gives \(Z_*\le1\).  If \(Z_*=1\), the
leading horizon equation reads \(1=V_*(2-V_*)\) and hence \(V_*=1\).
Thus every sequence with a nonzero \(Z\)-limit converges to the single corner
\((V,Z)=(1,1)\).

There is no fractional scale hidden at this corner.  Put
\(x=V-1\) and \(y=z^2/\mu^2-1\).  Exact Taylor division of the rational
equations first gives the coarse uniform balances
\begin{align}
 \frac{\mathcal H}{\mu^2}
 &=x^2+y-4\mu+o_K(x^2+|y|+\mu),\nonumber\\
 \alpha^2&=4\mu-y+o_K(|y|+\mu),\nonumber\\
 \frac{H_u}{\mu}
 &=2x-4\mu+o_K(|x|+\mu).
 \label{eq:cmetric-uniform-coarse-bootstrap}
\end{align}
The horizon identity expresses \(y\) as
\(4\mu-x^2+o_K(x^2+\mu)\).  If \(x^2/\mu\) were unbounded, the lapse
balance would give \(\alpha/|x|\to1\), while the unsquared temperature
would tend to \(2\operatorname{sgn}x\), contrary to
\(\tau_M\to0^+\).  Hence
\begin{equation}
 x^2=O_I(\mu),\qquad y=O_I(\mu).
 \label{eq:cmetric-uniform-coarse-trap}
\end{equation}
On this bootstrapped set, Taylor division with integral remainder gives the
sharper identities below.  After the displayed vanishing powers have been
extracted, every denominator in the desingularized rational coefficients is
bounded away from zero uniformly on \(I\):
\begin{align}
 \frac{\mathcal H}{\mu^2}
 &=x^2+y-4\mu-4\mu x+O_I(\mu^2),\nonumber\\
 \alpha^2&=4\mu-y+4\mu y+O_I(\mu^2),\nonumber\\
 \frac{H_u}{\mu}
 &=2x-4\mu+O_I(\mu^2).
 \label{eq:cmetric-uniform-first-balance}
\end{align}
The horizon equation first gives
\begin{equation}
 y=4\mu+4\mu x-x^2+O_I(\mu^2).
 \label{eq:cmetric-uniform-y-first}
\end{equation}
If \(|x|/\mu\) were unbounded, substitution of
Eq.~\eqref{eq:cmetric-uniform-y-first} in the second line of
Eq.~\eqref{eq:cmetric-uniform-first-balance} would give
\(\alpha/|x|\to1\).  The third line, together with \(u/\mu,z/\mu\to1\),
would then give
\(\tau\to2\operatorname{sgn}x\), contrary to \(\tau_M\to0^+\).
Therefore \(x=O_I(\mu)\), and
Eq.~\eqref{eq:cmetric-uniform-y-first} gives
\(y=4\mu+O_I(\mu^2)\).

Set \(X=x/\mu\), which is now uniformly bounded.  Substitution in the exact
horizon, lapse, and unsquared-temperature formulas sharpens the balance to
\begin{align}
 \frac{\alpha^2}{\mu^2}
 &=r^2+(X-2)^2+O_I(\mu),\nonumber\\
 \frac{H_u}{\mu^2}
 &=2(X-2)+O_I(\mu),\nonumber\\
 \tau&=\frac{2(X-2)}
 {\sqrt{r^2+(X-2)^2}}+O_I(\mu).
 \label{eq:cmetric-uniform-X-balance}
\end{align}
The target temperature is \(O_I(\mu)\).  Since \(r_0\le r\le r_1\), the last
line of Eq.~\eqref{eq:cmetric-uniform-X-balance} is uniformly bounded away
from zero whenever \(|X-2|\) is bounded away from zero, and its derivative at
\(X=2\) is \(2/r\).  Hence
\begin{equation}
 X-2=O_I(\mu),\qquad
 C:=\frac{X-2}{\mu}=O_I(1).
 \label{eq:cmetric-uniform-C-bound}
\end{equation}

The remaining weights are fixed by successive exact division of the horizon
equation.  With
\begin{align}
 \omega_4&=\frac{w-\mu^2-4\mu^3}{\mu^4},\nonumber\\
 \omega_5&=\frac{w-\mu^2-4\mu^3-13\mu^4}{\mu^5},\nonumber\\
 \Omega&=\frac{w-\mu^2-4\mu^3-13\mu^4-40\mu^5}{\mu^6},
 \label{eq:cmetric-uniform-successive-coordinates}
\end{align}
the weight-four and weight-five equations give
\begin{equation}
 \omega_4=13-(X-2)^2+O_I(\mu),
 \qquad
 \omega_5=40+O_I(\mu).
 \label{eq:cmetric-uniform-successive-bounds}
\end{equation}
Thus \(\Omega\) is bounded, and the weight-six equation is
\begin{equation}
 \Omega=112+6C-C^2-r^2+O_I(\mu).
 \label{eq:cmetric-uniform-Omega-bound}
\end{equation}
Consequently
\begin{equation}
 z^2=\mu^2+4\mu^3+13\mu^4+40\mu^5+O_I(\mu^6),\qquad
 u=\mu+2\mu^2+O_I(\mu^3).
 \label{eq:cmetric-uniform-high-trap}
\end{equation}
This proves the uniform high-chart bounds in
Eq.~\eqref{eq:cmetric-uniform-high-chart} and the boundedness of the
coordinates \((C,\Omega)\) in
Eq.~\eqref{eq:cmetric-uniform-high-coordinates}.  Let \(B_H\) be a fixed
compact rectangle containing all such \((C,\Omega)\) for
\(r\in I\) and sufficiently small \(\mu\).

In those coordinates define the desingularized residual map
\begin{equation}
 \Phi_H(C,\Omega;\mu,r)=
 \left(
 \frac{\mathcal H}{\mu^6},
 \frac{\tau-\tau_M}{\mu}
 \right).
 \label{eq:cmetric-uniform-high-residual}
\end{equation}
On every compact set containing all trapped candidates, it converges in
\(C^1\), uniformly for \(r\in I\), to
\begin{equation}
 \Phi_{H,0}(C,\Omega;r)=
 \left(
 C^2-6C+r^2+\Omega-112,
 \frac{2(C-3)}r-4
 \right).
 \label{eq:cmetric-uniform-high-limit-map}
\end{equation}
This map has the unique zero in
Eq.~\eqref{eq:cmetric-uniform-high-center}, and
\begin{equation}
 \det D_{(C,\Omega)}\Phi_{H,0}=-\frac2r,
 \qquad
 \left|\det D\Phi_{H,0}\right|\ge\frac2R>0.
 \label{eq:cmetric-uniform-high-jacobian}
\end{equation}
The uniform implicit-function theorem produces one high-chart solution.
To prove that it is the only one, choose a fixed tubular neighborhood
\(N_H\) of the compact zero graph in
Eq.~\eqref{eq:cmetric-uniform-high-center}.  On the compact complement of
\(N_H\),
\begin{equation}
 \gamma_H:=\min
 \bigl\{\|\Phi_{H,0}(C,\Omega;r)\|:
 r\in I,\ (C,\Omega)\in B_H\setminus N_H\bigr\}>0.
 \label{eq:cmetric-uniform-high-gap}
\end{equation}
Uniform convergence makes the finite-\(\mu\) residual larger than
\(\gamma_H/2\) there.  Thus the high chart contains exactly the branch
\(S\), with no additional isolated or intermediate-scale root.

It remains to analyze sequences with \(Z\to0\).  Set
\begin{equation}
 W=\frac z{\mu^2},\qquad
 \Delta_\mu=
 \frac{2(1+\mu^4W^2)}d-V-\frac{\mu^2W^2}{V}>0.
 \label{eq:cmetric-uniform-low-delta}
\end{equation}
The horizon equation and its \(u\)-derivative give the exact on-shell
relations
\begin{align}
 1+\mu^2r^2
 &=\frac{W^2d^2(1-\mu^2V^2)}
 {V^3(1+\mu^4W^2)^2}\,\Delta_\mu,
 \label{eq:cmetric-blowup-low-exact-horizon}\\
 \frac{H_u}{\mu}
 &=2\mu^2V^2\Delta_\mu
 +2(1-\mu^2V^2)
 \left[V-\frac{1+\mu^4W^2}{d}+2\Delta_\mu\right],
 \label{eq:cmetric-blowup-low-exact-Hu}\\
 \frac{\tau}{\mu}
 &=\frac{Wd}{(1+\mu^4W^2)V^2\alpha}\,
 \frac{H_u}{\mu}.
 \label{eq:cmetric-blowup-low-exact-temperature}
\end{align}
The second identity holds exactly on the horizon because its left-hand side
minus its displayed right-hand side is
\(4\mathcal H/(\mu^2V)\).  Every prefactor relating \(\tau\) to \(H_u\) in the last line
is positive, so \(\tau>0\) fixes \(H_u>0\).

Suppose first that \(W\to0\).  Since
\(0<\Delta_\mu<2/d\), the horizon identity forces \(V\to0\).
Positivity of \(H_u\) in the second exact identity gives
\(\liminf\Delta_\mu\ge1/2\); the alternative
\(\Delta_\mu\to0\) would instead give \(H_u/\mu\to-2\).
The horizon identity now gives \(V\asymp W^{2/3}\), and hence
\[
 \frac{\mu^2W^2}{V}=O(\mu^2W^{4/3})\longrightarrow0.
\]
It follows from the definition of \(\Delta_\mu\) that
\(\Delta_\mu\to2\).  Thus the positive-temperature degeneration at
\(W=0\) is uniquely \(V\to0,\Delta_\mu\to2\).

Suppose next that \(W\to\infty\) while \(Z=\mu W\to0\).  The horizon
identity first gives \(\Delta_\mu\to0\), but this alone does not select a
limit for \(V\).  Its definition gives
\[
 V+\frac{Z^2}{V}\longrightarrow2.
\]
Besides \(V\to2\), the only possible endpoint degeneration is
\(V\to0\) with \(Z^2/V\to2\).  On that alternative, the second exact
identity gives \(H_u/\mu\to-2\), again contradicting the unsquared positive
temperature.  Therefore \(V\to2\).  At either allowed endpoint,
\begin{equation}
 \frac{\tau}{\mu}\asymp
 \frac1{\sqrt{V\Delta_\mu}}\longrightarrow+\infty,
 \label{eq:cmetric-uniform-low-end-gap}
\end{equation}
whereas \(\tau_M/\mu=4+O_I(\mu^2)\).  A compactness contradiction therefore
provides the constants \(c_0,C_0\), independent of \(r\in I\), in
Eq.~\eqref{eq:cmetric-uniform-low-chart}.  This excludes every scale with
\(z/\mu\to0\) but \(z/\mu^2\to\infty\).

On the compact low chart, \(\mu^{-2}\mathcal H(\mu V,\mu^4W^2)\)
converges in \(C^1\) to
\begin{equation}
 \mathcal H_{L,0}(V,W)=V^2-2V+\frac{V^4}{W^2}.
 \label{eq:cmetric-uniform-low-horizon}
\end{equation}
Its positive zero set is the graph
\begin{equation}
 W_0(V)=\sqrt{\frac{V^3}{2-V}},\qquad 0<V<2,
\end{equation}
and
\(\partial_W\mathcal H_{L,0}=-2V^4/W^3\ne0\).  A uniform tubular
implicit-function theorem therefore gives a unique finite-\(\mu\) horizon
graph \(W_\mu(V;r)\) containing every low-chart candidate.  Restricted to
this graph, the temperature residual converges in \(C^1\), uniformly for
\(r\in I\), to
\begin{equation}
 f_0(V)=\frac{2(3-V)}{\sqrt{V(2-V)}}-4.
 \label{eq:cmetric-uniform-low-temperature-residual}
\end{equation}
Its only zeros in \((0,2)\) are
\begin{equation}
 V=1,\qquad V=\frac95,
\end{equation}
and they are simple because
\begin{equation}
 f_0'(1)=-2,\qquad f_0'\!\left(\frac95\right)=\frac{50}{9}.
 \label{eq:cmetric-uniform-low-simple}
\end{equation}
Choose disjoint neighborhoods \(N_L\) and \(N_3\) of these two roots.  On
the compact interval \(I_L=[c_0,2-c_0]\), shrink the neighborhoods if
necessary and set
\begin{equation}
 \gamma_L:=\min
 \bigl\{|f_0(V)|:V\in I_L\setminus(N_L\cup N_3)\bigr\}>0.
 \label{eq:cmetric-uniform-low-gap}
\end{equation}
Uniform convergence excludes a finite-\(\mu\) zero outside those
neighborhoods, while the uniform implicit-function theorem produces one
and only one zero inside each.  The corresponding limiting values of
\(W_0\) are \(1\) and \(27/5\), proving
Eq.~\eqref{eq:cmetric-uniform-low-centers}.

The inverse-Jacobian bounds in the high chart and the two nonzero derivatives
in Eq.~\eqref{eq:cmetric-uniform-low-simple} are uniform on \(I\).  Hence the
three implicit-function neighborhoods may be chosen with positive radii
independent of \(r\).  Taking \(\delta_0\) smaller than all three radii and
smaller than half the separation of the two low-chart centers makes the
isolating neighborhoods in the statement mutually disjoint.  Shrinking
\(\mu_0\) once more makes the high and low charts themselves disjoint.

If the lemma were false for every \(\mu_0\), a sequence of counterexamples
with \(\mu\to0\) would, after taking a subsequence in the compact parameter
set \(I\), enter either the high or the low chart.  The residual gaps
Eqs.~\eqref{eq:cmetric-uniform-high-gap} and
\eqref{eq:cmetric-uniform-low-gap} exclude an extra zero in either chart.
This contradiction proves all assertions with a single threshold
\(\mu_0(I)\).
\end{proof}

As a direct consequence, the Gibbs comparison is global within the stated
black-hole sector.  On the third branch,
\begin{equation}
 \mathfrak g_3=\frac{2}{27\mu}+O_I(1),
\end{equation}
whereas \(\mathfrak g_M=O_I(\mu)\).  Hence the third branch lies strictly
above the Maxwell pair for all sufficiently small \(\mu\), uniformly on
\(I\).  Since Lemma~\ref{lem:cmetric-uniform-root-exhaustion} leaves no other
positive-temperature equilibrium candidate, the two admissible strict
minima \(L\) and \(S\) exhaust the global minima among admissible
positive-temperature black-hole equilibria in this boundary layer.

\section{Action normalization and exact scheme robustness}
\label{app:cmetric-renormalization}

The absolute Euclidean action of an accelerating black hole depends on the
renormalization prescription.  Phase selection depends instead on differences
between branches at common reservoir data.  For the two prescriptions that
have been worked out explicitly for the charged AdS C-metric, this distinction
can be made exact.

\subsection{Standard fixed-charge scheme}

The quantities in Eq.~\eqref{eq:cmetric-thermodynamics} use the normalized
Killing field and standard asymptotically AdS counterterms of
Ref.~\cite{AnabalonEtAl2019Thermodynamics}.  With the Maxwell boundary term
appropriate to fixed electric charge, the canonical potential is
\begin{equation}
 G_{\rm can}^{\rm std}=M-TS
 \label{eq:app-standard-canonical-potential}
\end{equation}
at fixed \((P,Q,\mu)\).  The normalization
\(\alpha^2=\Xi(1-A^2\ell^2\Xi)\) enters the physical definitions of
\(M\) and \(T\).  Covariant phase-space
charges give the same thermodynamic variables after the conformal-boundary
corner term is included \cite{KimEtAl2023CovariantPhaseSpace}.

For later comparison, write the north and south polar tensions as
\begin{equation}
 \mu_N=\frac14\left(1-\frac{\Xi+2mA}{K}\right),
 \qquad
 \mu_S=\frac14\left(1-\frac{\Xi-2mA}{K}\right),
 \label{eq:app-polar-tensions}
\end{equation}
and set \(\mu_\Sigma=\mu_N+\mu_S\).  The north-regular single-string
sector has \(K=\Xi+2mA\), hence
\begin{equation}
 \mu_N=0,
 \qquad
 \mu_S=\frac{mA}{K}=\mu,
 \qquad
 \mu_\Sigma=\mu.
 \label{eq:app-single-string-tension-map}
\end{equation}
Thus the overall tension used in Ref.~\cite{HaleEtAl2025ChargedAccelerating}
equals the tension held fixed in the present ensemble, fixing the factor
convention.

\subsection{Topological prescription}

At fixed polar tensions and with the same normalized time, Eq.~(93) of
Ref.~\cite{HaleEtAl2025ChargedAccelerating} relates the grand-canonical
potentials obtained from topological and standard renormalization by
\begin{equation}
 \Omega^{\rm top}=\Omega^{\rm std}+2\pi\ell^2\mu_\Sigma T.
 \label{eq:app-hale-grand-shift}
\end{equation}
Because the added term is independent of the electric potential, the
fixed-charge Legendre transformation commutes with this shift.  Using
\(P=3/(8\pi\ell^2)\) and Eq.~\eqref{eq:app-single-string-tension-map} gives
\begin{equation}
 \boxed{
 G_{\rm can}^{\rm top}(x;T,P,Q,\mu)
 =G_{\rm can}^{\rm std}(x;T,P,Q,\mu)
 +\frac{3\mu T}{4P}}
 \label{eq:app-canonical-scheme-shift}
\end{equation}
for every admissible black-hole branch \(x\).

\begin{proposition}[Exact scheme robustness of phase selection]
\label{prop:exact-scheme-robustness}
Compare the standard and topological prescriptions on the same family of
black-hole solutions, at fixed \((T,P,Q,\mu)\), fixed polar tensions, and the
same normalization of the timelike Killing field.  Then the stationary
cover, its response-derived Morse labels, every branch-value difference, the
Maxwell set, the global winner among physical black-hole equilibria, and the
complete Maxwell-turning curve are identical
in the two prescriptions.  The jumps \(\Delta S\) and \(\Delta V\), the
Clapeyron slope, latent heat, turning curvature, and reduced
stationary-saddle barrier are also identical.  Only absolute thermodynamic
potentials and their common branch-independent offsets change.
\end{proposition}

\begin{proof}
At common reservoir data, Eq.~\eqref{eq:app-canonical-scheme-shift} adds the
same number to every solution.  Hence, for any two branches \(i,j\),
\begin{equation}
 G_i^{\rm top}-G_j^{\rm top}
 =G_i^{\rm std}-G_j^{\rm std}.
 \label{eq:app-scheme-difference-invariance}
\end{equation}
The underlying solutions and reservoir map are unchanged, so the stationary
cover and its local response data are unchanged.  Equation
\eqref{eq:app-scheme-difference-invariance} preserves all critical-value
equalities and inequalities, including comparisons with the unstable
stationary sheet.  It therefore preserves Maxwell coexistence, the winner
order, and every barrier defined by a stationary-value difference.

Differentiating Eq.~\eqref{eq:app-canonical-scheme-shift} at fixed values of
the remaining controls gives the branch-independent shifts
\begin{align}
 S^{\rm top}&=S^{\rm std}-\frac{3\mu}{4P},
 &
 V^{\rm top}&=V^{\rm std}-\frac{3\mu T}{4P^2},
 \label{eq:app-scheme-extensive-shifts}\\
 \Phi^{\rm top}&=\Phi^{\rm std},
 &
 \mathcal Y_\mu^{\rm top}
 &=\mathcal Y_\mu^{\rm std}+\frac{3T}{4P},
 \nonumber
\end{align}
where \(\mathcal Y_\mu=(\partial G/\partial\mu)_{T,P,Q}\); its relation to a
signed thermodynamic length depends on the first-law convention.
The common entropy and volume offsets cancel between phases.  Consequently
\(\Delta S\), \(\Delta V\), \({\rm d}T/{\rm d}P=\Delta V/\Delta S\), and
\(T\Delta S\) are invariant.  Since the coexistence curve itself is
unchanged, so are its turns, curvature, and endpoints.
\end{proof}

The distinction between absolute normalization and phase selection is visible
on the exact turning curve.  In the standard prescription both coexisting
branches have \(G_{\rm can}^{\rm std}=0\).  Equations
\eqref{eq:cmetric-turning-controls} and
\eqref{eq:app-canonical-scheme-shift} instead give
\begin{equation}
 G_{\rm can}^{\rm top}\big|_{\rm turn}
 =\frac{3\mu T_{\rm turn}}{4P_{\rm turn}}
 =\frac{2Q}{1+\chi^2}.
 \label{eq:app-topological-turning-value}
\end{equation}
The choice of zero changes between the two schemes.  The equality of the two
values and their order relative to all other admissible black-hole branches
remain invariant.

Proposition~\ref{prop:exact-scheme-robustness} applies with the polar tensions,
time normalization, admissible solution set, and thermodynamic ensemble held
fixed.  The reduced barrier used here is the difference between stationary
values on the one-dimensional equilibrium family.  A gravitational
tunneling exponent would additionally require a fully off-shell Euclidean
bounce calculation.

\section*{Funding}

This research did not receive any specific grant from funding agencies in the
public, commercial, or not-for-profit sectors.

\section*{Data availability}

All data generated or analyzed in this study are included in the article and
its Supplementary Material.  The supplementary archive contains the symbolic
computations, high-precision numerical checks, root-isolation data, source
tables, and figure-generation code needed to reproduce the reported results,
together with dependency and execution instructions.

\section*{CRediT authorship contribution statement}

Ruiliang Li: Conceptualization, Methodology, Formal analysis, Investigation,
Software, Validation, Visualization, Writing--original draft,
Writing--review and editing.

\section*{Declaration of competing interest}

The author declares that he has no known competing financial interests or
personal relationships that could have appeared to influence the work
reported in this paper.


\begin{thebibliography}{10}
\expandafter\ifx\csname url\endcsname\relax
  \def\url#1{\texttt{#1}}\fi
\expandafter\ifx\csname urlprefix\endcsname\relax\def\urlprefix{URL }\fi
\expandafter\ifx\csname href\endcsname\relax
  \def\href#1#2{#2} \def\path#1{#1}\fi

\bibitem{HawkingPage1983}
S.~W. Hawking, D.~N. Page, Thermodynamics of black holes in anti-de {Sitter}
  space, Communications in Mathematical Physics 87 (1983) 577--588.
\newblock \href {https://doi.org/10.1007/BF01208266}
  {\path{doi:10.1007/BF01208266}}.

\bibitem{Chamblin1999Fluctuations}
A.~Chamblin, R.~Emparan, C.~V. Johnson, R.~C. Myers, Holography,
  thermodynamics, and fluctuations of charged {AdS} black holes, Physical
  Review D 60 (1999) 104026.
\newblock \href {http://arxiv.org/abs/hep-th/9904197}
  {\path{arXiv:hep-th/9904197}}, \href
  {https://doi.org/10.1103/PhysRevD.60.104026}
  {\path{doi:10.1103/PhysRevD.60.104026}}.

\bibitem{GregoryScoins2019Chemistry}
R.~Gregory, A.~Scoins, Accelerating black hole chemistry, Physics Letters B 796
  (2019) 191--195.
\newblock \href {http://arxiv.org/abs/1904.09660} {\path{arXiv:1904.09660}},
  \href {https://doi.org/10.1016/j.physletb.2019.06.071}
  {\path{doi:10.1016/j.physletb.2019.06.071}}.

\bibitem{AnabalonEtAl2019Thermodynamics}
A.~Anabal{\'o}n, F.~Gray, R.~Gregory, D.~Kubiz{\v n}{\'a}k, R.~B. Mann,
  Thermodynamics of charged, rotating, and accelerating black holes, Journal of
  High Energy Physics 04 (2019) 096.
\newblock \href {http://arxiv.org/abs/1811.04936} {\path{arXiv:1811.04936}},
  \href {https://doi.org/10.1007/JHEP04(2019)096}
  {\path{doi:10.1007/JHEP04(2019)096}}.

\bibitem{AbbasvandiEtAl2019Snapping}
N.~Abbasvandi, W.~Cong, D.~Kubiz{\v n}{\'a}k, R.~B. Mann, Snapping swallowtails
  in accelerating black hole thermodynamics, Classical and Quantum Gravity
  36~(10) (2019) 104001.
\newblock \href {http://arxiv.org/abs/1812.00384} {\path{arXiv:1812.00384}},
  \href {https://doi.org/10.1088/1361-6382/ab129f}
  {\path{doi:10.1088/1361-6382/ab129f}}.

\bibitem{GunasekaranKubiznakMann2012Extended}
S.~Gunasekaran, D.~Kubiz{\v n}{\'a}k, R.~B. Mann, Extended phase space
  thermodynamics for charged and rotating black holes and {Born--Infeld} vacuum
  polarization, Journal of High Energy Physics 11 (2012) 110.
\newblock \href {http://arxiv.org/abs/1208.6251} {\path{arXiv:1208.6251}},
  \href {https://doi.org/10.1007/JHEP11(2012)110}
  {\path{doi:10.1007/JHEP11(2012)110}}.

\bibitem{AltamiranoEtAl2014Review}
N.~Altamirano, D.~Kubiz{\v n}{\'a}k, R.~B. Mann, Z.~Sherkatghanad,
  Thermodynamics of rotating black holes and black rings: Phase transitions and
  thermodynamic volume, Galaxies 2~(1) (2014) 89--159.
\newblock \href {http://arxiv.org/abs/1401.2586} {\path{arXiv:1401.2586}},
  \href {https://doi.org/10.3390/galaxies2010089}
  {\path{doi:10.3390/galaxies2010089}}.

\bibitem{AltamiranoKubiznakMann2013Reentrant}
N.~Altamirano, D.~Kubiz{\v n}{\'a}k, R.~B. Mann, Reentrant phase transitions in
  rotating anti--de {Sitter} black holes, Physical Review D 88 (2013) 101502.
\newblock \href {http://arxiv.org/abs/1306.5756} {\path{arXiv:1306.5756}},
  \href {https://doi.org/10.1103/PhysRevD.88.101502}
  {\path{doi:10.1103/PhysRevD.88.101502}}.

\bibitem{AltamiranoEtAl2014Triple}
N.~Altamirano, D.~Kubiz{\v n}{\'a}k, R.~B. Mann, Z.~Sherkatghanad, {Kerr--AdS}
  analogue of triple point and solid/liquid/gas phase transition, Classical and
  Quantum Gravity 31~(4) (2014) 042001.
\newblock \href {http://arxiv.org/abs/1308.2672} {\path{arXiv:1308.2672}},
  \href {https://doi.org/10.1088/0264-9381/31/4/042001}
  {\path{doi:10.1088/0264-9381/31/4/042001}}.

\bibitem{FrassinoEtAl2014MultipleReentrant}
A.~M. Frassino, D.~Kubiz{\v n}{\'a}k, R.~B. Mann, F.~Simovic, Multiple
  reentrant phase transitions and triple points in {Lovelock} thermodynamics,
  Journal of High Energy Physics 09 (2014) 080.
\newblock \href {http://arxiv.org/abs/1406.7015} {\path{arXiv:1406.7015}},
  \href {https://doi.org/10.1007/JHEP09(2014)080}
  {\path{doi:10.1007/JHEP09(2014)080}}.

\bibitem{AbbasvandiEtAl2019FinelySplit}
N.~Abbasvandi, W.~Ahmed, W.~Cong, D.~Kubiz{\v n}{\'a}k, R.~B. Mann, Finely
  split phase transitions of rotating and accelerating black holes, Physical
  Review D 100 (2019) 064027.
\newblock \href {http://arxiv.org/abs/1906.03379} {\path{arXiv:1906.03379}},
  \href {https://doi.org/10.1103/PhysRevD.100.064027}
  {\path{doi:10.1103/PhysRevD.100.064027}}.

\bibitem{SpallucciSmailagic2013Maxwell}
E.~Spallucci, A.~Smailagic, {Maxwell}'s equal-area law for charged {Anti-de
  Sitter} black holes, Physics Letters B 723 (2013) 436--441.
\newblock \href {http://arxiv.org/abs/1305.3379} {\path{arXiv:1305.3379}},
  \href {https://doi.org/10.1016/j.physletb.2013.05.038}
  {\path{doi:10.1016/j.physletb.2013.05.038}}.

\bibitem{WeiLiu2015Clapeyron}
S.-W. Wei, Y.-X. Liu, Clapeyron equations and fitting formula of the
  coexistence curve in the extended phase space of charged {AdS} black holes,
  Physical Review D 91 (2015) 044018.
\newblock \href {http://arxiv.org/abs/1411.5749} {\path{arXiv:1411.5749}},
  \href {https://doi.org/10.1103/PhysRevD.91.044018}
  {\path{doi:10.1103/PhysRevD.91.044018}}.

\bibitem{LiWang2022Landscape}
R.~Li, J.~Wang, Generalized free energy landscape of a black hole phase
  transition, Physical Review D 106 (2022) 106015.
\newblock \href {http://arxiv.org/abs/2206.02623} {\path{arXiv:2206.02623}},
  \href {https://doi.org/10.1103/PhysRevD.106.106015}
  {\path{doi:10.1103/PhysRevD.106.106015}}.

\bibitem{XuEtAl2024ComplexFreeEnergy}
Z.-M. Xu, Y.-S. Wang, B.~Wu, W.-L. Yang, Generalized {Maxwell} equal area law
  and black holes in complex free energy, Physics Letters B 850 (2024) 138528.
\newblock \href {http://arxiv.org/abs/2305.05916} {\path{arXiv:2305.05916}},
  \href {https://doi.org/10.1016/j.physletb.2024.138528}
  {\path{doi:10.1016/j.physletb.2024.138528}}.

\bibitem{AnabalonEtAl2018Holographic}
A.~Anabal{\'o}n, M.~Appels, R.~Gregory, D.~Kubiz{\v n}{\'a}k, R.~B. Mann,
  A.~{\"O}vg{\"u}n, Holographic thermodynamics of accelerating black holes,
  Physical Review D 98 (2018) 104038.
\newblock \href {http://arxiv.org/abs/1805.02687} {\path{arXiv:1805.02687}},
  \href {https://doi.org/10.1103/PhysRevD.98.104038}
  {\path{doi:10.1103/PhysRevD.98.104038}}.

\bibitem{ZhangLiYu2019Accelerating}
J.~Zhang, Y.~Li, H.~Yu, Thermodynamics of charged accelerating {AdS} black
  holes and holographic heat engines, Journal of High Energy Physics 02 (2019)
  144.
\newblock \href {http://arxiv.org/abs/1808.10299} {\path{arXiv:1808.10299}},
  \href {https://doi.org/10.1007/JHEP02(2019)144}
  {\path{doi:10.1007/JHEP02(2019)144}}.

\bibitem{KastorRayTraschen2009}
D.~Kastor, S.~Ray, J.~Traschen, Enthalpy and the mechanics of {AdS} black
  holes, Classical and Quantum Gravity 26 (2009) 195011.
\newblock \href {http://arxiv.org/abs/0904.2765} {\path{arXiv:0904.2765}},
  \href {https://doi.org/10.1088/0264-9381/26/19/195011}
  {\path{doi:10.1088/0264-9381/26/19/195011}}.

\bibitem{Dolan2011PressureVolume}
B.~P. Dolan, Pressure and volume in the first law of black hole thermodynamics,
  Classical and Quantum Gravity 28 (2011) 235017.
\newblock \href {http://arxiv.org/abs/1106.6260} {\path{arXiv:1106.6260}},
  \href {https://doi.org/10.1088/0264-9381/28/23/235017}
  {\path{doi:10.1088/0264-9381/28/23/235017}}.

\bibitem{KubiznakMann2012PV}
D.~Kubiz{\v n}{\'a}k, R.~B. Mann, {$P$--$V$} criticality of charged {AdS} black
  holes, Journal of High Energy Physics 07 (2012) 033.
\newblock \href {http://arxiv.org/abs/1205.0559} {\path{arXiv:1205.0559}},
  \href {https://doi.org/10.1007/JHEP07(2012)033}
  {\path{doi:10.1007/JHEP07(2012)033}}.

\bibitem{CveticEtAl2011Volume}
M.~Cveti{\v c}, G.~W. Gibbons, D.~Kubiz{\v n}{\'a}k, C.~N. Pope, Black hole
  enthalpy and an entropy inequality for the thermodynamic volume, Physical
  Review D 84 (2011) 024037.
\newblock \href {http://arxiv.org/abs/1012.2888} {\path{arXiv:1012.2888}},
  \href {https://doi.org/10.1103/PhysRevD.84.024037}
  {\path{doi:10.1103/PhysRevD.84.024037}}.

\bibitem{KubiznakMannTeo2017Chemistry}
D.~Kubiz{\v n}{\'a}k, R.~B. Mann, M.~Teo, Black hole chemistry: thermodynamics
  with {$\Lambda$}, Classical and Quantum Gravity 34 (2017) 063001.
\newblock \href {http://arxiv.org/abs/1608.06147} {\path{arXiv:1608.06147}},
  \href {https://doi.org/10.1088/1361-6382/aa5c69}
  {\path{doi:10.1088/1361-6382/aa5c69}}.

\bibitem{AppelsGregoryKubiznak2016Thermodynamics}
M.~Appels, R.~Gregory, D.~Kubiz{\v n}{\'a}k, Thermodynamics of accelerating
  black holes, Physical Review Letters 117 (2016) 131303.
\newblock \href {http://arxiv.org/abs/1604.08812} {\path{arXiv:1604.08812}},
  \href {https://doi.org/10.1103/PhysRevLett.117.131303}
  {\path{doi:10.1103/PhysRevLett.117.131303}}.

\bibitem{AppelsGregoryKubiznak2017Conical}
M.~Appels, R.~Gregory, D.~Kubiz{\v n}{\'a}k, Black hole thermodynamics with
  conical defects, Journal of High Energy Physics 05 (2017) 116.
\newblock \href {http://arxiv.org/abs/1702.00490} {\path{arXiv:1702.00490}},
  \href {https://doi.org/10.1007/JHEP05(2017)116}
  {\path{doi:10.1007/JHEP05(2017)116}}.

\bibitem{ZhangEtAl2026Geometric}
S.-H. Zhang, Z.-Y. Li, J.-F. Zhang, X.~Zhang, Universal geometric framework for
  black hole phase transitions: from multivaluedness to classification,
  European Physical Journal C 86 (2026) 642.
\newblock \href {http://arxiv.org/abs/2512.16629} {\path{arXiv:2512.16629}},
  \href {https://doi.org/10.1140/epjc/s10052-026-15851-5}
  {\path{doi:10.1140/epjc/s10052-026-15851-5}}.

\bibitem{ZhangEtAl2026Unifying}
S.-H. Zhang, S.-W. Wei, J.-F. Zhang, X.~Zhang, Unifying topological, geometric,
  and complex classifications of black hole thermodynamics, preprint (2026).
\newblock \href {http://arxiv.org/abs/2604.08315} {\path{arXiv:2604.08315}},
  \href {https://doi.org/10.48550/arXiv.2604.08315}
  {\path{doi:10.48550/arXiv.2604.08315}}.

\bibitem{Chamblin1999Catastrophic}
A.~Chamblin, R.~Emparan, C.~V. Johnson, R.~C. Myers, Charged {AdS} black holes
  and catastrophic holography, Physical Review D 60 (1999) 064018.
\newblock \href {http://arxiv.org/abs/hep-th/9902170}
  {\path{arXiv:hep-th/9902170}}, \href
  {https://doi.org/10.1103/PhysRevD.60.064018}
  {\path{doi:10.1103/PhysRevD.60.064018}}.

\bibitem{Cerf1970Stratification}
J.~Cerf, La stratification naturelle des espaces de fonctions
  diff{\'e}rentiables r{\'e}elles et le th{\'e}or{\`e}me de la pseudo-isotopie,
  Publications Math{\'e}matiques de l'IH{\'E}S 39 (1970) 5--173.
\newblock \href {https://doi.org/10.1007/BF02684687}
  {\path{doi:10.1007/BF02684687}}.

\bibitem{GolubitskyGuillemin1973}
M.~Golubitsky, V.~Guillemin, Stable Mappings and Their Singularities, Vol.~14
  of Graduate Texts in Mathematics, Springer, New York, 1973.
\newblock \href {https://doi.org/10.1007/978-1-4615-7904-5}
  {\path{doi:10.1007/978-1-4615-7904-5}}.

\bibitem{ArnoldGuseinZadeVarchenko2012}
V.~I. Arnold, S.~M. Gusein-Zade, A.~N. Varchenko, Singularities of
  Differentiable Maps, Volume 1: Classification of Critical Points, Caustics
  and Wave Fronts, Modern Birkh{\"a}user Classics, Birkh{\"a}user, Boston,
  2012.
\newblock \href {https://doi.org/10.1007/978-0-8176-8340-5}
  {\path{doi:10.1007/978-0-8176-8340-5}}.

\bibitem{Vassiliev2025BifurcationSets}
V.~A. Vassiliev, Real function singularities and their bifurcation sets, in:
  J.~L. Cisneros-Molina, L.~D. Tr{\'a}ng, J.~Seade (Eds.), Handbook of Geometry
  and Topology of Singularities VII, Springer, Cham, 2025, pp. 71--119.
\newblock \href {https://doi.org/10.1007/978-3-031-68711-2_2}
  {\path{doi:10.1007/978-3-031-68711-2_2}}.

\bibitem{WeiLiu2022Topology}
S.-W. Wei, Y.-X. Liu, Topology of black hole thermodynamics, Physical Review D
  105 (2022) 104003.
\newblock \href {http://arxiv.org/abs/2112.01706} {\path{arXiv:2112.01706}},
  \href {https://doi.org/10.1103/PhysRevD.105.104003}
  {\path{doi:10.1103/PhysRevD.105.104003}}.

\bibitem{WeiLiuMann2022Defects}
S.-W. Wei, Y.-X. Liu, R.~B. Mann, Black hole solutions as topological
  thermodynamic defects, Physical Review Letters 129 (2022) 191101.
\newblock \href {http://arxiv.org/abs/2208.01932} {\path{arXiv:2208.01932}},
  \href {https://doi.org/10.1103/PhysRevLett.129.191101}
  {\path{doi:10.1103/PhysRevLett.129.191101}}.

\bibitem{WeiLiu2026Review}
S.-W. Wei, Y.-X. Liu, Topology of black hole thermodynamics: A brief review,
  Science China Physics, Mechanics \& Astronomy 69 (2026) 260401.
\newblock \href {http://arxiv.org/abs/2605.00037} {\path{arXiv:2605.00037}},
  \href {https://doi.org/10.1007/s11433-025-2923-3}
  {\path{doi:10.1007/s11433-025-2923-3}}.

\bibitem{Wu2023AcceleratingTopology}
D.~Wu, Topological classes of thermodynamics of the four-dimensional static
  accelerating black holes, Physical Review D 108 (2023) 084041.
\newblock \href {http://arxiv.org/abs/2307.02030} {\path{arXiv:2307.02030}},
  \href {https://doi.org/10.1103/PhysRevD.108.084041}
  {\path{doi:10.1103/PhysRevD.108.084041}}.

\bibitem{WuYangWei2025Extended}
S.-P. Wu, S.-J. Yang, S.-W. Wei, Extended thermodynamical topology of black
  hole, European Physical Journal C 85 (2025) 1372.
\newblock \href {http://arxiv.org/abs/2508.01614} {\path{arXiv:2508.01614}},
  \href {https://doi.org/10.1140/epjc/s10052-025-15098-6}
  {\path{doi:10.1140/epjc/s10052-025-15098-6}}.

\bibitem{Wu2026BranchStructure}
D.~Wu, Topological changes and response singularities in black hole
  thermodynamic branch structure, preprint (2026).
\newblock \href {http://arxiv.org/abs/2607.07364} {\path{arXiv:2607.07364}},
  \href {https://doi.org/10.48550/arXiv.2607.07364}
  {\path{doi:10.48550/arXiv.2607.07364}}.

\bibitem{Li2026LegendreCovariant}
R.~Li, Legendre-covariant thermodynamic topology of black holes, Nuclear
  Physics B (2026) 117597.
\newblock \href {https://doi.org/10.1016/j.nuclphysb.2026.117597}
  {\path{doi:10.1016/j.nuclphysb.2026.117597}}.

\bibitem{HaleEtAl2025ChargedAccelerating}
T.~Hale, D.~Kubiz{\v n}{\'a}k, J.~Men{\v s}{\'i}kov{\'a}, R.~B. Mann, J.~Yang,
  Thermodynamics of charged and accelerating black holes, Physical Review D 111
  (2025) 104004.
\newblock \href {http://arxiv.org/abs/2501.13679} {\path{arXiv:2501.13679}},
  \href {https://doi.org/10.1103/PhysRevD.111.104004}
  {\path{doi:10.1103/PhysRevD.111.104004}}.

\bibitem{KinnersleyWalker1970}
W.~Kinnersley, M.~Walker, Uniformly accelerating charged mass in general
  relativity, Physical Review D 2 (1970) 1359--1370.
\newblock \href {https://doi.org/10.1103/PhysRevD.2.1359}
  {\path{doi:10.1103/PhysRevD.2.1359}}.

\bibitem{PlebanskiDemianski1976}
J.~F. Pleba{\'n}ski, M.~Demia{\'n}ski, Rotating, charged, and uniformly
  accelerating mass in general relativity, Annals of Physics 98~(1) (1976)
  98--127.
\newblock \href {https://doi.org/10.1016/0003-4916(76)90240-2}
  {\path{doi:10.1016/0003-4916(76)90240-2}}.

\bibitem{GriffithsPodolsky2006}
J.~B. Griffiths, J.~Podolsk{\'y}, A new look at the
  {Pleba{\'n}ski--Demia{\'n}ski} family of solutions, International Journal of
  Modern Physics D 15 (2006) 335--369.
\newblock \href {http://arxiv.org/abs/gr-qc/0511091}
  {\path{arXiv:gr-qc/0511091}}, \href
  {https://doi.org/10.1142/S0218271806007742}
  {\path{doi:10.1142/S0218271806007742}}.

\bibitem{GriffithsPodolsky2009Book}
J.~B. Griffiths, J.~Podolsk{\'y}, Exact Space-Times in Einstein's General
  Relativity, Cambridge Monographs on Mathematical Physics, Cambridge
  University Press, Cambridge, 2009.
\newblock \href {https://doi.org/10.1017/CBO9780511635397}
  {\path{doi:10.1017/CBO9780511635397}}.

\bibitem{Arnold1978Boundary}
V.~I. Arnol'd, Critical points of functions on a manifold with boundary, the
  simple {Lie} groups {$B_k$}, {$C_k$}, {$F_4$}, and singularities of evolutes,
  Russian Mathematical Surveys 33~(5) (1978) 99--116.
\newblock \href {https://doi.org/10.1070/RM1978v033n05ABEH002515}
  {\path{doi:10.1070/RM1978v033n05ABEH002515}}.

\bibitem{NisseLimChang2024ADiscriminants}
M.~Nisse, Y.-K. Lim, L.~Chang, Black hole thermodynamic free energy as
  {A}-discriminants, International Journal of Theoretical Physics 63 (2024)
  176.
\newblock \href {http://arxiv.org/abs/2311.11801} {\path{arXiv:2311.11801}},
  \href {https://doi.org/10.1007/s10773-024-05711-x}
  {\path{doi:10.1007/s10773-024-05711-x}}.

\bibitem{BorodzikNemethiRanicki2016Boundary}
M.~Borodzik, A.~N{\'e}methi, A.~Ranicki, Morse theory for manifolds with
  boundary, Algebraic \& Geometric Topology 16~(2) (2016) 971--1023.
\newblock \href {http://arxiv.org/abs/1207.3066} {\path{arXiv:1207.3066}},
  \href {https://doi.org/10.2140/agt.2016.16.971}
  {\path{doi:10.2140/agt.2016.16.971}}.

\bibitem{BorodzikBuczynska2025Families}
M.~Borodzik, W.~Buczy{\'n}ska, Families of {Morse} functions for manifolds with
  boundary, preprint (2025).
\newblock \href {http://arxiv.org/abs/2507.15847} {\path{arXiv:2507.15847}},
  \href {https://doi.org/10.48550/arXiv.2507.15847}
  {\path{doi:10.48550/arXiv.2507.15847}}.

\bibitem{GoreskyMacPherson1988Stratified}
M.~Goresky, R.~MacPherson, Stratified Morse Theory, Springer, Berlin and
  Heidelberg, 1988.
\newblock \href {https://doi.org/10.1007/978-3-642-71714-7}
  {\path{doi:10.1007/978-3-642-71714-7}}.

\bibitem{Berge1963Topological}
C.~Berge, Topological Spaces: Including a Treatment of Multi-Valued Functions,
  Vector Spaces and Convexity, Macmillan, New York, 1963.

\bibitem{AliprantisBorder2006Infinite}
C.~D. Aliprantis, K.~C. Border, Infinite Dimensional Analysis: A Hitchhiker's
  Guide, 3rd Edition, Springer, Berlin and Heidelberg, 2006.
\newblock \href {https://doi.org/10.1007/3-540-29587-9}
  {\path{doi:10.1007/3-540-29587-9}}.

\bibitem{GhoshBhamidipati2019Contact}
A.~Ghosh, C.~Bhamidipati, Contact geometry and thermodynamics of black holes in
  {AdS} spacetimes, Physical Review D 100 (2019) 126020.
\newblock \href {http://arxiv.org/abs/1909.11506} {\path{arXiv:1909.11506}},
  \href {https://doi.org/10.1103/PhysRevD.100.126020}
  {\path{doi:10.1103/PhysRevD.100.126020}}.

\bibitem{Bravetti2019Contact}
A.~Bravetti, Contact geometry and thermodynamics, International Journal of
  Geometric Methods in Modern Physics 16~(supp01) (2019) 1940003.
\newblock \href {https://doi.org/10.1142/S0219887819400036}
  {\path{doi:10.1142/S0219887819400036}}.

\bibitem{Golubitsky1978Catastrophe}
M.~Golubitsky, An introduction to catastrophe theory and its applications, SIAM
  Review 20~(2) (1978) 352--387.
\newblock \href {https://doi.org/10.1137/1020043} {\path{doi:10.1137/1020043}}.

\bibitem{CoxLittleOShea2025}
D.~A. Cox, J.~Little, D.~O'Shea, Ideals, Varieties, and Algorithms: An
  Introduction to Computational Algebraic Geometry and Commutative Algebra, 5th
  Edition, Undergraduate Texts in Mathematics, Springer, Cham, 2025.
\newblock \href {https://doi.org/10.1007/978-3-031-91841-4}
  {\path{doi:10.1007/978-3-031-91841-4}}.

\bibitem{BasuPollackRoy2006}
S.~Basu, R.~Pollack, M.-F. Roy, Algorithms in Real Algebraic Geometry, 2nd
  Edition, Vol.~10 of Algorithms and Computation in Mathematics, Springer,
  Berlin and Heidelberg, 2006.
\newblock \href {https://doi.org/10.1007/3-540-33099-2}
  {\path{doi:10.1007/3-540-33099-2}}.

\bibitem{KimEtAl2023CovariantPhaseSpace}
H.~Kim, N.~Kim, Y.~Lee, A.~Poole, Thermodynamics of accelerating {AdS}$_4$
  black holes from the covariant phase space, European Physical Journal C 83
  (2023) 1095.
\newblock \href {http://arxiv.org/abs/2306.16187} {\path{arXiv:2306.16187}},
  \href {https://doi.org/10.1140/epjc/s10052-023-12266-4}
  {\path{doi:10.1140/epjc/s10052-023-12266-4}}.

\end{thebibliography}
\end{document}